\documentclass[11pt,letterpaper]{article}
\usepackage{multirow}
\usepackage[utf8]{inputenc}
\usepackage{etex}
\usepackage{xspace,enumerate}
\usepackage[dvipsnames]{xcolor}
\usepackage[T1]{fontenc}
\usepackage[full]{textcomp}
\usepackage[american]{babel}
\usepackage{mathtools}
\usepackage{amsthm}
\newtheorem{theorem}{Theorem}[section]
\newtheorem*{theorem*}{Theorem}

\newtheorem*{proposition*}{Proposition}
\newtheorem{lemma}[theorem]{Lemma}
\newtheorem*{lemma*}{Lemma}
\newtheorem{corollary}[theorem]{Corollary}
\newtheorem*{conjecture*}{Conjecture}

\newtheorem*{fact*}{Fact}

\newtheorem*{hypothesis*}{Hypothesis}

\theoremstyle{definition}
\newtheorem{definition}[theorem]{Definition}
\newtheorem*{definition*}{Definition}

\theoremstyle{remark}

\newtheorem*{claim*}{Claim}
\newtheorem{remark}[theorem]{Remark}
\newtheorem*{remark*}{Remark}

\newtheorem*{observation*}{Observation}

\usepackage[
letterpaper,
top=1.2in,
bottom=1.2in,
left=1in,
right=1in]{geometry}
\usepackage{newpxtext} 
\usepackage{textcomp} 
\usepackage[varg,bigdelims]{newpxmath}
\usepackage[scr=rsfso]{mathalfa}
\usepackage{bm} 
\let\mathbb\varmathbb
\usepackage{microtype}
\usepackage[
pagebackref,
colorlinks=true,
urlcolor=blue,
linkcolor=blue,
citecolor=OliveGreen,
]{hyperref}
\usepackage[capitalise,nameinlink]{cleveref}
\crefname{lemma}{Lemma}{Lemmas}
\crefname{fact}{Fact}{Facts}
\crefname{theorem}{Theorem}{Theorems}
\crefname{corollary}{Corollary}{Corollaries}
\crefname{claim}{Claim}{Claims}
\crefname{example}{Example}{Examples}
\crefname{algorithm}{Algorithm}{Algorithms}
\crefname{problem}{Problem}{Problems}
\crefname{definition}{Definition}{Definitions}
\usepackage{paralist}
\usepackage{turnstile}
\usepackage{tikz}
\usepackage{caption}
\DeclareCaptionType{Algorithm}
\usepackage{newfloat}

\usepackage{url}

\usepackage{comment}

\global\long\def\argmin{\mathrm{argmin}}
\def\AA{\mathbf{A}}
\def\BB{\mathbf{B}}

\def\MM{\mathbf{M}}
\def\HH{\mathbf{H}}
\def\Om{\mathbf{\Omega}}
\def\NN{\mathbf{N}}

\def\WW{\mathbf{W}}
\def\CC{\mathbf{C}}
\def\DD{\mathbf{D}}
\def\II{\mathbb{I}}
\def\YY{\mathbf{Y}}
\def\XX{\mathbf{X}}
\def\PP{\mathbf{P}}

\def\UU{\mathbf{U}}

\def\VV{\mathbf{V}}
\def\EE{\mathbf{E}}
\def\SS{\mathbf{S}}
\def\TT{\mathbf{T}}
\def\WW{\mathbf{W}}
\def\PI{\mathbf{\Pi}}
\def\QQ{\mathbf{Q}}

\def\MM{\mathbf{M}}

\def\ZZ{\mathbf{Z}}
\def\RR{\mathbf{R}}
\def\II{\mathbf{I}}
\def\Y{\mathbf{Y}}
\def\X{\mathcal{X}}
\def\S{\mathcal{S}}

\def\R{\mathcal{R}}
\def\T{\mathcal{T}}

\def\Sig{\mathbf{\Sigma}}

\usepackage{ amssymb }

\def\poly{\textrm{poly}}

\def\prob#1#2{\mbox{\bf Pr}_{#1}\left[ #2 \right]}

\def\expec#1#2{{\bf \mathbb{E}}_{#1}[ #2 ]}
\def\expecf#1#2{{\bf \mathbb{E}}_{#1}\left[ #2 \right]}
\def\var#1{\mbox{\bf Var}[ #1 ]}
\def\varf#1{\mbox{\bf Var}\left[ #1 \right]}

\def\trace#1{\mathrm{Tr} \left(#1 \right)}

\def\ceil#1{\left\lceil #1 \right\rceil}

\def\expec#1#2{{\bf \mathbb{E}}_{#1}[ #2 ]}
\def\expecf#1#2{{\bf \mathbb{E}}_{#1}\left[ #2 \right]}
\def\var#1{\mbox{\bf Var}[ #1 ]}
\def\varf#1{\mbox{\bf Var}\left[ #1 \right]}

\numberwithin{equation}{section}
\numberwithin{figure}{section}

\usepackage[framemethod=TikZ]{mdframed}
\newcounter{Frame}
\newenvironment{Frame}[1][htb]{%
\refstepcounter{Frame}
    \begin{mdframed}[%
        frametitle={#1},
        skipabove=\baselineskip plus 2pt minus 1pt,
        skipbelow=\baselineskip plus 2pt minus 1pt,
        linewidth=1.0pt,
        frametitlerule=true,
        nobreak=true
    ]%
}{%
    \end{mdframed}
}

\title{Robust and Sample Optimal Algorithms for\\ PSD Low-Rank Approximation}
\author{
Ainesh Bakshi
\thanks{ Ainesh Bakshi and David Woodruff acknowledge support in part from  NSF  No.  CCF-1815840.}
\\
CMU\\
abakshi@cs.cmu.edu
\and
Nadiia Chepurko
\\
MIT\\
nadiia@mit.edu
\and David P. Woodruff\footnotemark[1]
\\
CMU\\
dwoodruf@cs.cmu.edu}
\begin{document}
\date{}
\maketitle
\thispagestyle{empty}
\setcounter{page}{0}
\begin{abstract}
    Recently, Musco and Woodruff (FOCS, 2017) showed that given an $n \times n$ positive semidefinite (PSD) matrix $\AA$, it is possible to compute a  $(1+\epsilon)$-approximate relative-error low-rank approximation to $\AA$ by querying $\widetilde{O}(nk/\epsilon^{2.5})$ entries of $\AA$ in time $\widetilde{O}(nk/\epsilon^{2.5} +n k^{\omega-1}/\epsilon^{2(\omega-1)})$. They also showed that any relative-error low-rank approximation algorithm must query ${\Omega}(nk/\epsilon)$ entries of $\AA$, this gap has since remained open. Our main result is to resolve this question by obtaining an optimal algorithm that queries  $\widetilde{O}(nk/\epsilon)$ entries of $\AA$ and outputs a relative-error low-rank approximation in $\widetilde{O}(n\cdot(k/\epsilon)^{\omega-1})$ time. Note, our running time improves that of Musco and Woodruff, and matches the information-theoretic lower bound if the matrix-multiplication exponent $\omega$ is $2$.  
    
    We then extend our techniques to negative-type distance matrices. Here, our input is a pair-wise distance matrix $\AA$ corresponding to a point set $\mathcal{P} = \{x_1, x_2, \ldots, x_n\}$ such that $\AA_{i,j}=\| x_i - x_j \|^2_2$. Bakshi and Woodruff (NeurIPS, 2018) showed a bi-criteria, relative-error low-rank approximation for negative-type metrics. Their algorithm queries $\widetilde{O}(nk/\epsilon^{2.5})$ entries and outputs a rank-$(k+4)$ matrix. We show that the bi-criteria guarantee is not necessary and obtain an  $\widetilde{O}(nk/\epsilon)$ query algorithm, which is optimal. Our algorithm applies to all distance matrices that arise from metrics satisfying negative-type inequalities, including $\ell_1, \ell_2,$ spherical metrics, hypermetrics and effective resistances on a graph. We also obtain faster algorithms for ridge regression. 
   
    Next, we introduce a new robust low-rank approximation model which captures PSD matrices that have been corrupted with noise. We assume that the Frobenius norm of the corruption is bounded. Here, we relax the notion of approximation to additive-error, since it is information-theoretically impossible to obtain a relative-error approximation in this setting. While a sample complexity lower bound precludes sublinear algorithms for arbitrary PSD matrices, we provide the first sublinear time and query algorithms when the corruption on the diagonal entries is bounded. As a special case, we show sample-optimal sublinear time algorithms for low-rank approximation of correlation matrices corrupted by noise. 
\end{abstract}

\newpage

\thispagestyle{empty}
\setcounter{page}{0}
\tableofcontents 
\newpage

\section{Introduction}

Low-rank approximation is one of the most common dimensionality reduction techniques, whereby one replaces a large matrix $\AA$ with a low-rank factorization $\UU \cdot \VV \approx \AA$. Such a factorization provides a compact way of storing $\AA$ and allows one to multiply $\AA$ quickly by a vector. It is used as an algorithmic primitive in clustering \cite{drineas2004clustering,mcsherry2001spectral}, recommendation systems \cite{drineas2002competitive}, web search \cite{achlioptas2001web,kleinberg1999authoritative}, and learning mixtures of distributions \cite{achlioptas2005spectral, kannan2005spectral}, and has numerous other applications. 

A large body of recent work has looked at {\it relative-error} low-rank approximation, whereby given an $n \times n$ matrix $\AA$, an accuracy parameter $\epsilon > 0$, and a rank parameter $k$, one seeks to output a rank-$k$ matrix $\BB$ for which 
\begin{eqnarray}\label{eqn:relError}
\|\AA-\BB\|_F^2 \leq (1+\epsilon)\|\AA-\AA_k\|_F^2,
\end{eqnarray}
where for a matrix $\CC$, $\|\CC\|_F^2 = \sum_{i,j} \CC_{i,j}^2$, and $\AA_k$ denotes the best rank-$k$ approximation to $\AA$ in Frobenius norm. $\AA_k$ can be computed exactly using the singular value decomposition, but takes time $O(n^{\omega})$, where $\omega$ is the matrix multiplication constant.  
We refer the reader to the survey \cite{w14} and references therein.

For worst-case matrices, it is not hard to see that any algorithm achieving (\ref{eqn:relError}) must spend at least $\Omega(\textrm{nnz}(\AA))$ time, where $\textrm{nnz}(\AA)$ denotes the number of non-zero entries (sparsity) of $\AA$. Indeed, without reading most of the non-zero entries of $\AA$, one could fail to read a single large entry, thus making one's output matrix $\BB$ an arbitrarily bad approximation. 

A flurry of recent work  \cite{kerenidis2016quantum, mw17,  bakshi2018sublinear, chia2018quantum, tang2019quantum, rebentrost2018quantum, gilyen2018quantum,   indyk2019sample, sw19, gilyen2019quantum} has looked at the possibility of achieving {\it sublinear} time algorithms (classical and quantum) for low-rank approximation. In particular, Musco and Woodruff \cite{mw17} consider the important case of positive-semidefinite (PSD) matrices. PSD matrices include as special cases covariance matrices, correlation matrices, graph Laplacians, kernel matrices and random dot product models. Further, the special case where the input itself is low-rank (PSD Matrix Completion) has applications in quantum state tomography \cite{gross2010quantum}. Subsequently, Bakshi and Woodruff \cite{bakshi2018sublinear} considered low-rank approximation of the closely related family of Negative-type (Euclidean Squared) distance matrices. Negative-type metrics include as special cases  $\ell_1$ and $\ell_2$ metrics, spherical metrics and hypermetrics, as well as effective resistances in graphs \cite{deza2009geometry,terwilliger1987classification, chandra1996electrical,christiano2011electrical}. Negative-type metrics have found various applications in algorithm design and optimization \cite{arora2008euclidean,spielman2011graph, koutis2014approaching, madry2015fast}. 

Musco and Woodruff show that it is possible to output a low-rank matrix $\BB$ in factored form achieving (\ref{eqn:relError}) in $\widetilde{O}(nk/\epsilon^{2.5} +n k^{\omega-1}/\epsilon^{2(\omega-1)})$ time, while reading only $\widetilde{O}(nk/\epsilon^{2.5})$ entries of $\AA$. They also showed a lower bound that any algorithm achieving (\ref{eqn:relError}) must read $\Omega(nk/\epsilon)$ entries, and closing the gap between these bounds has remained an open question. 
Similarly, Bakshi and Woodruff exploit the structure of Negative-type metrics to reduce to the PSD case and obtain a bi-criteria algorithm that requires $\widetilde{O}(nk/\epsilon^{2.5})$ queries. The gap in the sample complexity and the requirement of a bi-criteria guarantee remained open. 
We resolve these questions here, and describe our novel technical contributions in Section \ref{sec:tech}.



Next we consider PSD matrices that have been corrupted by a small amount of noise. A drawback of algorithms achieving (\ref{eqn:relError}) is that they cannot tolerate any amount of unstructured noise. For instance, if one slightly corrupts a few off-diagonal entries, making the input matrix $\AA$ no longer PSD, then it is impossible to detect such corruptions in sublinear time, making the relative-error guarantee (\ref{eqn:relError}) information-theoretically impossible. Motivated by this, we  also introduce a new framework where an adversary corrupts the input by adding a noise matrix $\NN$ to a psd matrix $\AA$. We assume that the Frobenius norm of the corruption is bounded relative to the Frobenius norm of $\AA$, i.e., $\|\NN\|^2_F \leq \eta \|\AA \|^2_F$.  We also assume the corruption is well-spread, i.e., each row of $\NN$ has $\ell^2_2$-norm at most a fixed constant factor larger than $\ell^2_2$-norm of the corresponding row of $\AA$. 

This model captures small perturbations to PSD matrices that we may observe in real-world datasets, as a consequence of round-off or numerical errors in tasks such as computing Laplacian pseudoinverses, and systematic measurement errors when computing a covariance matrix. One important application captured by our model is low-rank approximation of corrupted \textit{correlation matrices}. Finding a low-rank approximation of such matrices occurs when measured correlations are asynchronous or incomplete, or when models are stress-tested by adjusting individual correlations. Low-rank approximation of correlation matrices also has many applications in finance \cite{h02}. 

Given that it is information-theoretically impossible to obtain the relative-error guarantee (\ref{eqn:relError}) in the \emph{robust model}, we relax our notion of approximation to the following well-studied additive-error guarantee:
\begin{eqnarray}\label{eqn:addError}
\|\AA-\BB\|_F^2 \leq \|\AA-\AA_k\|_F^2 + (\epsilon+\eta) \|\AA\|_F^2. 
\end{eqnarray}
This additive-error guarantee was introduced by the seminal work of Frieze et. al. \cite{fkv04}, and triggered a long line of work on low-rank approximation from a computational perspective. Frieze et al. showed that it is possible to achieve (\ref{eqn:addError}) in $O(\textrm{nnz}(\AA))$ time. Further, given access to an oracle for computing row norms of $\AA$, \ref{eqn:addError} is achievable in sublinear time. More recently, the same notion of approximation was used to obtain sublinear sample complexity and running time algorithms for \textit{distance matrices} \cite{bakshi2018sublinear,indyk2019sample}, and
a quantum algorithm for recommendation systems \cite{kerenidis2016quantum}, which was subsequently dequantized \cite{tang2019quantum}. 

This raises the question of how robust are our sublinear low-rank approximation algorithms for structured matrices, if we relax to additive-error guarantees and allow for corruption. In particular, can we obtain additive-error low-rank approximation algorithms for PSD matrices that achieve sublinear time and sample complexity in the presence of noise? We characterize when such robust algorithms are achievable in sublinear time. 



\subsection{Our Results}

\begin{table}[]
\small
\centering
\begin{tabular}{|c|c|c|c|c|c|}
\hline
\multirow{2}{*}{\textbf{Problem}} & \multicolumn{2}{c|}{\textbf{Prior Work}} & \multicolumn{2}{c|}{\textbf{Our Results}} & \textbf{Query Lower } \\ \cline{2-5} 
 & \textbf{Query} & \textbf{Run Time} & \textbf{Query} & \textbf{Run Time} & \textbf{Bound}  \\ \hline
\multirow{2}{*}{\textbf{PSD} \textbf{LRA}} & $O\left(\frac{nk}{\epsilon^{2.5}}\right)$ & $O\left(\frac{nk^{\omega-1}}{\epsilon^{2\omega-2}}+ \frac{nk}{\epsilon^{2.5}}\right)$ &  $O^*\left(\frac{nk}{\epsilon}\right)$ & $O^{\dagger}\left(\frac{nk^{\omega-1}}{\epsilon^{\omega-1}}\right)$ & $\Omega\left(\frac{nk}{\epsilon}\right)$\\ \cline{2-6} 
 & \multicolumn{2}{c|}{\cite{mw17}} & \multicolumn{2}{c|}{Thm. \ref{thm:sample_opt_psd_lra}} & \cite{mw17}\\ \hline
\multirow{2}{*}{\shortstack{\textbf{PSD LRA}\\ \textbf{PSD Output}}} & $O\left(nk\left(\frac{k}{\epsilon^{2}}+\frac{1}{\epsilon^3}\right)\right)$ & $O\left(nk^{\omega-1}\left(\frac{k}{\epsilon^{\omega}}+\frac{1}{\epsilon^{3\omega-3}}\right)\right)$ &  $O^*\left(\frac{nk}{\epsilon}\right)$ & $O^{\dagger}\left(\frac{nk^{\omega-1}}{\epsilon^{\omega-1}}\right)$ & $\Omega\left(\frac{nk}{\epsilon}\right)$\\ \cline{2-6} 
 & \multicolumn{2}{c|}{\cite{mw17}} & \multicolumn{2}{c|}{Thm. \ref{thm:sample_opt_psd_lra}} & \cite{mw17}\\ \hline
\multirow{2}{*}{\shortstack{\textbf{Negative-Type}\\ \textbf{LRA} }} &  $O\left(\frac{nk}{\epsilon^{2.5}}\right)$ & $O\left(\frac{nk^{\omega-1}}{\epsilon^{2\omega-2}} + \frac{nk}{\epsilon^{2.5}}\right)$ &  $O^*\left(\frac{nk}{\epsilon}\right)$ & $O^{\dagger}\left(\frac{nk^{\omega-1}}{\epsilon^{\omega-1}}\right)$ & $\Omega\left(\frac{nk}{\epsilon}\right)$  \\ \cline{2-6} 
 & \multicolumn{2}{c|}{Bi-criteria, \cite{bakshi2018sublinear} } & \multicolumn{2}{c|}{No Bi-criteria, Thm. \ref{thm:optimal_euclidean_lra}} & \cite{bakshi2018sublinear} \\ \hline
\multirow{2}{*}{\shortstack{\textbf{Coreset Ridge}\\ \textbf{Regression}}} & $O\left(\frac{n s_{\lambda}^2}{\epsilon^4}\right)$ & $O\left(\frac{n s_{\lambda}^{\omega}}{\epsilon^{\omega}}\right)$ & $O^*\left(\frac{ns_{\lambda}}{\epsilon^2}\right)$ & $O^{\dagger}\left(\frac{ns^{\omega-1}_{\lambda}}{\epsilon^{2\omega-2}}\right)$ & $\Omega\left(\frac{n s_{\lambda}}{\epsilon^2}\right)$ \\ \cline{2-6} 
 & \multicolumn{2}{c|}{\cite{mw17}} & \multicolumn{2}{c|}{Thm. \ref{thm:ridge_regression}} & Thm \ref{thm:coreset_lb_for_ridge} \\ \hline
\end{tabular}
\caption{Comparison with prior work. The notation $O^*$ and $O^{\dagger}$ represent existence of matching lower bounds for query complexity and running time (assuming the fast matrix multiplication exponent $\omega$ is $2$) respectively. The notation $s_{\lambda}$ is used to denote the statistical dimension of ridge regression. All bounds are stated ignoring polylogarithmic factors in $n,k$ and $\epsilon$.}
\end{table}

We begin with stating our results for low-rank approximation for structured matrices. 
Our main result is an optimal algorithm for low-rank approximation of PSD matrices: 

\vspace{0.1in} 
\noindent \textbf{Theorem \ref{thm:sample_opt_psd_lra}} \textit{(Informal Sample-Optimal PSD LRA.) 
Given a PSD matrix $\AA$, there exists an algorithm that queries $\widetilde{O}(nk/\epsilon)$ entries in $\AA$ and outputs a rank $k$ matrix $\BB$ such that with probability $99/100$, $\|\AA - \BB \|^2_F \leq (1+\epsilon)\|\AA  - \AA_k\|^2_F$, and the algorithm runs in time $\widetilde{O}(n \cdot(k/\epsilon)^{\omega-1})$.
}
\vspace{0.1in} 
\begin{remark}
Our algorithm matches the sample complexity lower bound of Musco and Woodruff, up to logarithmic factors, which shows that any randomized algorithm that outputs a $(1+\epsilon)$-relative-error low-rank approximation for a PSD matrix $\AA$ must read $\Omega(nk/\epsilon)$ entries. Our running time also improves that of Musco and Woodruff and is optimal if the matrix multiplication exponent $\omega$ is $2$. 
\end{remark}

\begin{remark}
We can extend our algorithm such that the low-rank matrix $\BB$ we output is also PSD with the same query complexity and running time. In comparison, the algorithm of Musco and Woodruff accesses $\widetilde{O}(nk/\epsilon^3+ nk^2/\epsilon^2)$ entries in $\AA$ and runs in time $\widetilde{O}(n (k/\epsilon)^{\omega} + nk^{\omega-1}/\epsilon^{3(\omega-1)})$.
\end{remark}
 
At the core of our analysis is a sample optimal algorithm for Spectral Regression: $\min_{\XX}\|\DD \XX - \EE \|^2_2$. 
We show that when $\DD$ has orthonormal columns and $\EE$ is arbitrary, we can sketch the problem by sampling rows proportional to the leverage scores of $\DD$ and approximately preserve the minimum cost.  
This is particularly surprising since our sketch only computes sampling probabilities by reading entries in $\DD$, while being completely agnostic to the entries in  $\EE$. 
Here, we also prove a spectral approximate matrix product guarantee for our one-sided leverage score sketch, which may be of independent interest. We note that such a guarantee for leverage score sampling does not appear in prior work, and we discuss the technical challenges we need to overcome in the subsequent section.

The techniques we develop for PSD low-rank approximation also extend to computing a low-rank approximation for distance matrices that arise from negative-type (Euclidean-squared) metrics. Here, our input is a pair-wise distance matrix $\AA$ corresponding to a point set $\mathcal{P}=\{x_1, x_2, \ldots, x_n \} \in \mathbb{R}^d$ such that $\AA_{i,j}= \|x_i - x_j \|^2_2$. We obtain an optimal algorithm for computing a low-rank approximation of such matrices:


\vspace{0.1in} 
\noindent \textbf{Theorem \ref{thm:optimal_euclidean_lra}} \textit{(Informal Sample-Optimal LRA for Negative-Type Metrics.) 
Given a negative-type distance matrix $\AA$, there exists an algorithm that queries $\widetilde{O}(nk/\epsilon)$ entries in $\AA$ and outputs a rank $k$ matrix $\BB$ such that with probability $99/100$, $\|\AA - \BB \|^2_F \leq (1+\epsilon)\|\AA  - \AA_k\|^2_F$, and the algorithm runs in time $\widetilde{O}(n \cdot(k/\epsilon)^{\omega-1})$.
}

\begin{remark}
Prior work of Bakshi and Woodruff \cite{bakshi2018sublinear} obtains a $\widetilde{O}(nk/\epsilon^{2.5})$ query algorithm that outputs a rank-$(k+4)$ matrix $\BB$ such that $\|\AA - \BB \|^2_F \leq (1+\epsilon)\|\AA - \AA_k \|^2_F$. We show that the bi-criteria guarantee is not necessary, thereby resolving an open question in their paper. 
\end{remark}

\paragraph{Structured Regression.}
The sample-optimal algorithm for PSD Low-Rank Approximation also leads to a faster algorithm for Ridge Regression, when the design matrix is PSD. Given a PSD matrix $\AA$, a vector $y$ and a regularization parameter $\lambda$, we consider the following optimization problem:
$ \min_{x \in \mathbb{R}^n}\|\AA x - y \|^2_2 + \lambda\|x \|^2_2$.
This problem is often referred to as Ridge Regression and has been the focus of numerous theoretical and practical works (see \cite{gruber2017improving} and references therein).  

\vspace{0.1in} 
\noindent \textbf{Theorem \ref{thm:ridge_regression}} \textit{(Informal Ridge Regression.) 
Given a PSD matrix $\AA$, a regularization parameter $\lambda$ and statistical dimension $s_{\lambda} = \trace{ (\AA^2 +\lambda \II)^{-1} \AA^2}$,
there exists an algorithm that queries $\widetilde{O}(n s_{\lambda}/\epsilon^2)$ entries of $\AA$ and with probability $99/100$ outputs a $(1+\epsilon)$ approximate solution to the Ridge Regression objective and runs in $\widetilde{O}(n (s_{\lambda}/\epsilon^2)^{\omega-1})$ time. 
}
\vspace{0.1in}

\begin{remark}
Our result improves on prior work by Musco and Woodruff \cite{mw17}, who obtain an algorithm that queries $\widetilde{O}(ns_{\lambda}^2/\epsilon^4)$ entries in $\AA$ and runs in $\widetilde{O}(n (s_{\lambda}/\epsilon^2)^{\omega} )$ time. 
\end{remark}

\begin{remark}
Since our algorithm works for all $y$ simultaneously, we obtain a low-rank \textit{coreset} of the design matrix (in factored form) that preserves the Ridge Regression cost up to a $(1+\epsilon)$ factor. Further, in Theorem \ref{thm:coreset_lb_for_ridge}, we prove a matching lower bound on the query complexity for any coreset construction.
\end{remark}
  



\paragraph{Robust Low-Rank Approximation.}
Next, we consider a robust form of low-rank approximation problem, where the input is a PSD matrix corrupted by noise. In this setting, we have query access to the corrupted matrix $\AA + \NN$, where $\AA$ is PSD and $\NN$ is such that $\|\NN\|^2_F \leq \eta \|\AA\|^2_F$. Further, for all $i \in [n]$ $\|\NN_{i,*} \|^2_2\leq c \| \AA_{i,*}\|^2_2$, for a fixed constant $c$. 
The diagonal of a PSD matrix carries crucial information since the largest diagonal entry upper bounds all off-diagonal entries. Therefore, a reasonable adversarial strategy is to corrupt the largest diagonal entries and make them close to the small diagonal entries, which enables the resulting matrix to have large off-diagonal entries that are hard to find. Capturing this intuition we parameterize our algorithms and lower bounds by the largest ratio between a diagonal entry of $\AA$ and $\AA+\NN$, denoted by $\phi_{\max} = \max_{j \in [n]} \AA_{j,j} / |(\AA+\NN)_{j,j}|$.

\vspace{0.1in} 
\noindent \textbf{Theorem \ref{thm:robust_lb}.} \textit{(Informal lower bound.) 
Let $\epsilon>\eta>0$.
Given $\AA + \NN$ such that $\AA$ is PSD and  $\NN$ is a corruption matrix as defined above, any randomized algorithm that with probability at least $2/3$ outputs a rank-$k$ approximation up to additive error $(\epsilon+\eta)\|\AA \|^2_F$ must read  $\Omega\left( \phi_{\max}^2 nk/\epsilon \right)$ entries of $\AA+\NN$.
}
\vspace{0.1in} 

\begin{remark}
Any algorithm must incur additive error $\eta\|\AA \|^2_F$, since $\AA$ is not even identifiable below additive-error $\eta\|\AA \|^2_F$.
\end{remark}
\begin{remark}
In our hard instance, $\phi^2_{\max}$  can be as large as $\epsilon n/k$, which implies a sample-complexity lower bound of $\Omega(n^2)$. While this lower bound precludes sublinear algorithms for arbitrary PSD matrices, we observe that in many applications $\phi_{\max}$ can be significantly smaller. For instance, if $\AA$ is a correlation matrix, we know that the true diagonal entries of $\AA+\NN$ are $1$ and can ignore any corruption on them to bound $\phi_{\max}$ by $1$. 
\end{remark}

Motivated by the aforementioned observation, we introduce algorithms for robust low-rank approximation, parameterized by the corruption on the diagonal entries. We obtain the following theorem:

\vspace{0.1in} 
\noindent \textbf{Theorem \ref{thm:main_thm} } \textit{(Informal Robust LRA.)
Given $\AA + \NN$, which satisfies our noise model, there exists an algorithm that queries $\widetilde{O}(\phi_{\max}^2 nk/\epsilon)$ entries in $\AA+\NN$ and computes a rank $k$ matrix $\BB$ such that with probability at least $99/100$, $
    \|\AA - \BB \|^2_F \leq \|\AA - \AA_k \|^2_F + (\epsilon + \sqrt{\eta} )\|\AA\|^2_F
$.
}
\begin{remark}
While the sample complexity of this algorithm matches the sample complexity in the lower bound, it incurs additive-error $\sqrt{\eta}\| \AA\|^2_F$ as opposed to $\eta \|\AA\|^2_F$. An interesting open question here is whether we can achieve additive-error $o(\sqrt{\eta}\|\AA\|^2_F)$, though we note that when $\eta^2 \leq \epsilon$, this just changes the additive error guarantee of our low-rank approximation by a constant factor.
\end{remark}
   
\begin{remark}
Our techniques extend to low-rank approximation of correlation matrices, and we obtain a sample complexity of $\widetilde{O}(nk/\epsilon)$, which is optimal. In fact, the hard instance in \cite{mw17} implies an $\Omega(nk/\epsilon)$ lower bound on the sample complexity, even in 
the presence of no noise. 
Surprisingly, corrupting a correlation matrix does not increase the sample complexity and only incurs an additive error of $\sqrt{\eta}\|\AA \|^2_F$ (see Corollary \ref{cor:cor_matrix_lra} for a formal statement). 
\end{remark}


\section{Technical Overview}\label{sec:tech}

In this section, we provide an overview of our techniques and supply intuition for our proofs. As a first step, it is easy to see that the $\Omega(\textsf{nnz}(\AA))$ lower bound for general matrices does not apply to PSD matrices, since it proceeds by hiding arbitrarily large entries. Observe that, reading the diagonal of a PSD matrix certifies an upper bound on all entries of the matrix and thus off-diagonals cannot be arbitrarily large. With this intuition in mind, we focus on sublinear algorithms.

\subsection{Sample-Optimal Low-Rank Approximation} 
At a high level, our algorithm consists of two stages: first, we use the existing machinery developed by Musco and Woodruff \cite{mw17} to obtain \textit{weak projection-cost preserving sketches} for $\AA$. Our sketches are smaller than those obtained by Musco and Woodruff, albeit satisfying a weaker guarantee. Recall that such sketches reduce the dimensionality of the columns (rows), while ensuring that the norm of \textit{all} projections of rows of the form $\II - \PP$, for rank $k$ $\PP$, space are simultaneously preserved. Constructing such sketches for both the column and row spaces of $\AA$ results in a much smaller matrix, which we can afford to query. 

At this point our approach diverges from that of Musco and Woodruff, since it is not possible to follow their strategy and recover a  $(1+\epsilon)$-relative-error low-rank approximation from the weaker sketch we constructed above. However, we show that our sketch has enough information to extract a \textit{structured subspace} (represented by an orthonormal basis) for $\AA$ such that the projection onto the orthogonal complement of this subspace is comparable to the optimal low-rank approximation cost, in spectral norm. Note, this guarantee is stronger than the span of the structured subspace containing a low-rank approximation comparable to the optimal in Frobenius norm, and indeed the latter does not suffice. In the second stage, we show that we can recover a rank-$k$ matrix in the span of the structured subspace such that it is a $(1+\epsilon)$-relative-error low-rank approximation for $\AA$. Further, we show that all these steps can be performed in sublinear time and by reading only $\widetilde{O}(nk/\epsilon)$ entries in $\AA$ (see Theorem \ref{thm:sample_opt_psd_lra} for a precise statement).


We begin by providing a bird's eye view of the Musco-Woodruff algorithm and how to adapt parts of it to obtain \textit{weak projection-cost preserving sketches (PCPs)}. For ease of exposition, we ignore polylogarithmic factors in the subsequent discussion. Their algorithm begins with computing the so-called \textit{ridge leverage scores} (Definition \ref{def:ridge_leverage_scores}) for $\AA^{1/2}$, which approximate the \textit{ridge leverage scores} of $\AA$ up to a $\sqrt{n/k}$-factor. The ridge leverage scores of $\AA^{1/2}$ can be approximated efficiently since we can compute the row norms of $\AA^{1/2}$ by simply reading the diagonal of $\AA$. 
It is well known \cite{cohenmm17} that sampling $k/\epsilon^2_0$ columns of $\AA$ proportional to its ridge leverage scores results in a sketch $\CC$ that preserves the cost of all rank-$k$ projections $\PP$:
\begin{equation}
\label{eqn:rel_pcp}
    \|\CC - \PP\CC \|^2_F = (1\pm\epsilon_0)\|\AA - \PP \AA\|^2_F 
\end{equation}

In prior work, $\CC$ is referred to as a projection-cost preserving sketch (PCP).
PCP constructions are useful since a low-rank approximation for $\CC$ translates to a low-rank approximation for $\AA$, while $\CC$ has much smaller dimension.  Observe that oversampling columns of $\AA$ proportional to the \textit{ridge leverage scores} of $\AA^{1/2}$, by a $\sqrt{n/k}$ factor, suffices to obtain the guarantee of \ref{eqn:rel_pcp} (see Lemma \ref{lem:rel_col_pcp} for a precise statement). Note, $\CC$ may have $\Omega(n^{1.5}/\epsilon^2_0)$ non-zeros but the algorithm need not query any entries in $\CC$ at this stage.   
Musco and Woodruff then construct a row PCP for $\CC$ by sampling $\sqrt{nk}/\epsilon^{2.5}_0$ rows of $\CC$ proportional to the rank-$(k/\epsilon_0)$ \textit{ridge leverage scores} of $\AA$. The resulting matrix $\RR$ is a $\sqrt{nk}/\epsilon^{2.5}_0 \times \sqrt{nk}/\epsilon^{2.5}_0$ matrix such that for any rank-$k$ projection $\PP$, 
\begin{equation}
\label{eqn:rel_pcp2}
\|\RR - \RR \PP \|^2_F + O(\|\AA - \AA_k\|^2_F ) =  (1\pm \epsilon_0) \|\CC - \CC \PP \|^2_F 
\end{equation}

Since $\RR$ is a much smaller matrix, they run an input-sparsity time algorithm to compute a low-rank approximation for it  \cite{clarkson2013low}. Using standard regression techniques (described in \cite{cohenmm17, mw17, bakshi2018sublinear}) along with equations \ref{eqn:rel_pcp} and  \ref{eqn:rel_pcp2}, setting $\epsilon_0 =\epsilon$ results in a $(1+\epsilon)$-low-rank approximation of $\AA$ by querying $O(nk/\epsilon^{5})$ entries. Musco and Woodruff instead use a more complicated algorithm to get a $1/\epsilon^{2.5}$ dependence.

Diverging from Musco and Woodruff, we begin by showing that a Spectral-Frobenius  projection for $\AA$ suffices to obtain a low-rank approximation with $O(nk/\epsilon)$ queries. A projection matrix $\PP$ is a  $(\epsilon, k)$ Spectral-Frobenius projection for a matrix $\AA$ if $\|\AA - \AA \PP \|_2^2 \leq \frac{\epsilon }{k} \|\AA - \AA_k \|_F^2$ (see Definition \ref{def:sf_projection}).  
Assuming we are handed a $(k/\epsilon)$-dimensional \textit{structured subspace} such that the projection on to this subspace is given by $\PP = \QQ \QQ^\top$, where  $\QQ$ is a $n\times k/\epsilon$ matrix with orthonormal columns satisfying $\|\AA - \QQ\QQ^{\top}\AA  \|^2_2 \leq \epsilon/k \cdot \|\AA - \AA_{k/\epsilon}\|^2_F$, we show that we can extract a rank-$k$ relative-error approximation by reading only $nk/\epsilon$ entries in $\AA$ (Theorem \ref{thm:struc_proj_to_lra} in Section \ref{subsec:structured_projection_to_lra}). We provide an overview of the proof here.

The \textsf{SF} projection property implies $\|\AA - \QQ\QQ^{\top} \AA_k \QQ\QQ^{\top} \|^2_F \leq (1+\epsilon)\|\AA - \AA_k\|^2_F$. Therefore, it suffices to solve the following optimization problem:
\begin{equation}
\label{eqn:sf_cons_regression}
    \min_{\textrm{rank}(\XX) \leq k} \|\AA - \QQ \XX \QQ^{\top} \|^2_F
\end{equation}
since $\XX = \QQ^{\top} \AA_k \QQ$ is always feasible.
While we are now optimizing over a $k/\epsilon \times k/\epsilon$ matrix $\XX$, with rank at most $k$, the problem still seems intractable to solve \textit{optimally} in sublinear time and queries to $\AA$. The key idea here is that $\QQ$ has orthonormal columns and thus the leverage scores are precomputed for us. We can then sample columns and rows proportional to the leverage scores of $\QQ$ and consider a significantly smaller sketched problem. 
Therefore, we create sampling matrices $\SS$ and $\TT$ that sample $\textrm{poly}(k/\epsilon)$ rows proportional to the leverage scores of $\QQ$ and consider the resulting optimization problem:
\begin{equation}
\label{eqn:sketched_gen_lra}
    \min_{\textrm{rank}(X)\leq k}\|\SS\AA\TT - \SS\QQ\XX\QQ^{\top}\TT\|^2_F
\end{equation}
Here, we are faced with an intriguing phenomenon: our sketched optimization problem  does not have the property that the minimum cost for Equation \ref{eqn:sketched_gen_lra} is a $(1+\epsilon)$-approximation to the minimum cost for Equation \ref{eqn:sf_cons_regression}. The reason is that our sketch incurs a fixed additive shift term, which we cannot approximate in sublinear time. We note that this is the bottleneck in approximating the cost of the optimal low-rank approximation, and as mentioned in \cite{mw17}, it is open to estimate this cost in $o(n^{3/2})$ time.   

However, we can apply the structural result in Lemma \ref{lem:affine_embedding} twice, to show that the optimal solution to Equation \ref{eqn:sketched_gen_lra}, when plugged in to Equation \ref{eqn:sf_cons_regression}  obtains a $(1+\epsilon)$-approximation to the minimum cost.
Formally, $\SS$ and $\TT$ have the property that if $\widehat{\XX}= \arg\min_{\XX}\|\SS \AA \TT - \SS\QQ\XX\QQ^{\top}\TT \|^2_F$, then 
\begin{equation*}
      \|\AA - \QQ\widehat{\XX}\QQ^{\top}\|^2_F \leq (1+O(\epsilon)) \min_{\textrm{rank}(X)\leq k}\|\AA - \QQ\XX\QQ^{\top}\|^2_F.
\end{equation*}

The optimization problem in Equation \ref{eqn:sketched_gen_lra} is called \textit{Generalized Low-Rank Approximation} and admits a closed form solution \cite{friedland2007generalized} (Theorem \ref{thm:gen_lra}). Further, since the problem now has all dimensions independent of $n$, we can afford to explicitly compute $\SS \AA\TT$ by querying the corresponding entries in $\AA$.  The resulting closed-form solution can also be computed  in $\textrm{poly}(k/\epsilon)$ time (and queries) which only contributes a lower order term. We obtain one factor for the low-rank approximation for $\AA$ by simply computing an orthonormal basis for $\QQ\widehat{\XX}$. In order to compute the second factor, we set up and approximately solve a regression problem, the details of which can be found in Algorithm \ref{alg:proj_to_rank}. Efficiently solving such a regression problem is now standard in low-rank approximation literature \cite{cohenmm17, mw17, bakshi2018sublinear}. Therefore, we can output a low-rank approximation to $\AA$ by querying only $\widetilde{O}(nk/\epsilon)$ entries.

Now, we focus on obtaining a Spectral-Frobenius projection for $\AA$. 
Our starting point here is to observe that the PCP construction from equations \eqref{eqn:rel_pcp} and \eqref{eqn:rel_pcp2} above allows to preserve the projection of columns of $\AA$ on all $(k/\epsilon)$-dimensional subspaces, albeit up to a constant factor by setting $\epsilon_0 = O(1)$. 
Therefore, a natural approach is to set the error parameter $\epsilon_0$ in the PCP constructions to be a small fixed constant, say $0.1$, and the rank parameter $k$ to be $k/\epsilon$, where $\epsilon$ is the desired input accuracy.  Further, we observe that the guarantee obtained in Equation \ref{eqn:rel_pcp} can be strengthen to a mixed \textit{Spectral-Frobenius PCP} guarantee (also introduced by \cite{mw17}): for all rank-$(k/\epsilon)$ projection matrices $\PP$, the column PCP $\CC$ satisfies :
\begin{equation}
\label{eqn:weak_pcp1}
    (1-0.1)\|\AA -  \PP\AA \|^2_2 -\frac{\epsilon}{10k}\|\AA -\AA_{k/\epsilon}\|^2_F \leq \|\CC - \PP\CC\|^2_2 \leq (1+0.1)\|\AA - \PP \AA  \|^2_2 +\frac{\epsilon}{10k}\|\AA -\AA_{k/\epsilon}\|^2_F
\end{equation}
For a formal statement refer to Lemma \ref{lem:rel_row_pcp_mixed}. Sampling rows of $\CC$ proportional to the same distribution results in a row PCP for $\CC$ such that for all rank-$(k/\epsilon)$ projections $\PP$,
\begin{equation}
\label{eqn:weak_pcp2}
    (1-0.1)\|\CC -  \CC\PP \|^2_2 -\frac{\epsilon}{10k}\|\AA -\AA_{k/\epsilon}\|^2_F \leq \|\RR - \RR\PP\|^2_2 \leq (1+0.1)\|\CC - \CC \PP  \|^2_2 +\frac{\epsilon}{10k}\|\AA -\AA_{k/\epsilon}\|^2_F
\end{equation}



We then use an \textit{input-sparsity} spectral-low-rank approximation algorithm by \cite{cohen2015dimensionality} (Lemma \ref{lem:input_sparsity_spectral_lra}), to obtain a low-dimensional subspace, represented by a $ \sqrt{nk/\epsilon} \times k/\epsilon$ matrix $\ZZ$ with orthonormal columns such that 
\begin{equation}
\label{eqn:struct}
    \|\RR - \RR \ZZ \ZZ^{\top} \|^2_2 \leq \frac{\epsilon}{k} \|\RR  - \RR_{k/\epsilon}\|^2_F
\end{equation}
Observe the projection matrix $\ZZ \ZZ^{\top}$ as a $(\epsilon,  k)$-Spectral-Frobenius projection for $\RR$. 
Now, using the fact that $\RR$ is a Spectral-Frobenius PCP for $\CC$ and plugging in $\PP = \ZZ \ZZ^{\top}$ in Equation \ref{eqn:weak_pcp2}, we can bound $\|\CC - \CC\ZZ \ZZ^{\top} \|^2_2$ as follows :
\begin{equation}
\label{eqn:sf_c}
\begin{split}
    \|\CC - \CC \ZZ \ZZ^{\top} \|^2_2  \leq \frac{10}{9} \|\RR - \RR\ZZ\ZZ^{\top} \|^2_2 + \frac{\epsilon}{9k}\|\RR - \RR_{k/\epsilon}\|^2_F 
    & \leq \frac{10\epsilon}{k}\|\RR - \RR_k \|^2_F + \frac{\epsilon}{9k}\|\RR - \RR_{k/\epsilon}\|^2_F\\
    & \leq O\left(\frac{\epsilon}{k}\right) \|\CC - \CC_{k} \|^2_F
\end{split}
\end{equation}
where the second inequality follows from Equation \ref{eqn:struct} and the third follows from the fact that PCPs preserve Frobenius low-rank approximation cost up to a constant factor, i.e., $\|\RR - \RR_{k/\epsilon} \|^2_F =\Theta(\|\CC - \CC_{k/\epsilon} \|^2_F) = \Theta(\|\AA - \AA_{k/\epsilon} \|^2_F)$. Therefore, $\ZZ\ZZ^{\top}$ is also a $(\epsilon, k)$- Spectral-Frobenius projection for $\CC$.
Here, we are faced with a few challenges. First, the relative-error approximation spanned by the subspace has rank $k/\epsilon$. Second, it is unclear how to obtain any reasonable result for $\AA$ from the above structural property, given that even the dimensions of $\ZZ\ZZ^{\top}$ do not match $\AA$.


Here, we observe that if we had a such a projection for the \textit{column-space} of $\CC$, by Equation \ref{eqn:weak_pcp1}, it would also work for $\AA$ and we would be done. To this end, we consider the following optimization problem:
\begin{equation}
\label{eqn:tech_spectral_regression}
\min_{\WW \in \mathbb{R}^{n \times k'}} \|\CC - \WW \ZZ^{\top} \|^2_2
\end{equation}
We show that an orthonormal basis $\QQ$ for an approximate minimizer to  Equation \ref{eqn:tech_spectral_regression} is an \textsf{SF} projection for $\CC$ and in turn $\AA$ (since $\CC$ is a column PCP for $\AA$). Therefore, we focus on optimizing Equation \ref{eqn:tech_spectral_regression} and refer to this problem as \textit{Spectral Regression}. We note that unlike standard regression, here we minimize the Spectral (Operator) norm.  While the corresponding problem for minimizing Frobenius norm is extensively studied and well understood, to the best of our knowledge the only relevant related work on Spectral Regression is in the streaming model, by Clarkson and Woodruff \cite{clarkson2009numerical}. They construct an oblivious sketch, consisting of random entries in $\{-1,1\}$, for Equation \ref{eqn:tech_spectral_regression} that preserves the optimal solution up to a  $(1\pm\epsilon)$ factor. Unfortunately, we cannot use oblivious sketching here, since $\CC$ may be a dense matrix and we cannot afford to read all of it. 

Here, we emphasize that obtaining a sample-optimal algorithm for the aforementioned Spectral Regression problem is crucial for our main algorithmic result. Given that $\CC$ is an $n \times \sqrt{nk/\epsilon}$ matrix, we cannot query most of it and thus approximating its leverage scores is infeasible.  
A natural approach here would be to follow the \textit{Affine Embedding idea} for Frobenius norm (refer to Lemma \ref{lem:affine_embedding}) and hope a similar guarantee holds for the spectral norm as well. Here, one might hope to obtain a small sketch that preserves the spectral norm cost of all $\WW$ up to a $(1\pm \epsilon)$ factor. While such a guarantee would suffice, we note that $\ZZ$ could have rank as large as $k/\epsilon$ and we can no longer afford a $(1+\epsilon)$-approximate affine embedding even for Frobenius norm, without incurring a larger dependence on $\epsilon$. This precludes all known approaches for sketching Equation \ref{eqn:tech_spectral_regression} to preserve the optimal cost.

Instead, we relax the notion of approximation for our sketch. We observe that it suffices to construct a sketch $\SS$ such that if $ \widehat{\WW} = \arg\min_{\WW}\|\CC\SS- \WW\ZZ^{\top} \SS\|^2_2$, 
then 
\begin{equation}
\label{eqn:sketch_error}
    \|\CC -\widehat{\WW}\ZZ^{\top} \|^2_2 \leq O(1)\left(\min_{\WW \in \mathbb{R}^{n \times k/\epsilon}} \|\CC - \WW\ZZ^{\top}\|^2_2 + \frac{\epsilon}{k}\|\CC- \CC_{k/\epsilon} \|^2_F\right)
\end{equation}
Note, this is a weaker guarantee for the sketch $\SS$, since we only need to preserve the cost of the optimal solution up to a mixed relative and additive error. First, we observe such a guarantee suffices, since we can upper bound the cost from Equation \ref{eqn:sketch_error} by $O(\epsilon/k) \cdot \|\CC-\CC_{k/\epsilon}\|^2_F$  and the Spectral-Frobenius PCP from Equation \ref{eqn:weak_pcp1} incurs this term anyway. In Theorem \ref{thm:spectral_regression},  we show that we can construct such a sketch $\SS$ satisfying Equation \ref{eqn:sketch_error} by sampling $k/\epsilon$ columns of $\CC$ proportional to the leverage scores of $\ZZ^{\top}$.
This is surprising since we completely ignore all information about $\CC$ and our sketch is not an oblivious sketch. 

The key technical lemma (Lemma \ref{lem:weak_amm}) we prove here is a \textit{weak approximate matrix product} for $\CC^*$ and $\ZZ^{\top}$  where $\CC^* = \CC(\II - \PP_{\ZZ^{\top}})$ is the projection onto the orthogonal complement of $\ZZ^{\top}$. While  \textit{ approximate matrix product} has been extensively studied \cite{drineas2004clustering, sarlos2006improved, clarkson2013low}, even for spectral norm \cite{cohen2015optimal}, it is important to emphasize here that all known constructions are either oblivious sketches or require sketches that are sampled proportional to both $\CC^*$ and $\ZZ^{\top}$. Since $\ZZ^{\top}$ has no information about the spectrum of $\CC^*$, the main challenge here is to control the spectrum of $\CC^* \SS\SS^{\top}\ZZ^{\top}$.

In order to bound $\|\CC^* \SS\SS^{\top}\ZZ^{\top}\|_2$, we analyze how sampling columns of $\CC^*$ proportional to the \textit{leverage scores} of $\ZZ^{\top}$ affects the spectrum of $\CC^*$. 
An important tool in our analysis is the following result by Rudelson and Vershynin on how the spectral norm of a matrix degrades when we sample a uniformly random subset of rows \cite{rudelson2007sampling}. They show that sampling $q$ rows of a matrix $\MM$ uniformly at random, indexed by the set $\mathcal{Q}$, results in a matrix $\MM_{|\mathcal{Q}}$ such that
\begin{equation*}
     \expecf{}{\big\|\MM_{|\mathcal{Q}}\big\|_2} = O\left( \sqrt{\frac{q}{n}}\|\MM \|_2 + \sqrt{\log(q)}\|\MM \|_{(n/q)}\right)
\end{equation*}
where $\|\AA \|_{(n/q)}$ is the average of the largest $n/q$ $\ell_2$-norms of columns of $\AA$. 
Here, we prove that expected spectral norm of $\CC$ restricted to the columns sampled by $\SS$ proportional to the leverage scores of $\ZZ^T$ only exceeds that of a random subset by a polylogarithmic factor. 
This result may be of independent interest in applications where we would want to bound the spectrum of random submatrices, where the rows or columns are \textit{not} sampled uniformly.

Intuitively, there are two technical challenges we overcome in order to apply the Rudelson-Vershynin result in our setting. First, a leverage score sampling matrix $\SS$ need not sample columns uniformly at random, since we have no control over the squared column norms of $\ZZ^{\top}$. Given that the squared column norms of $\ZZ^{\top}$ may be lopsided, the subset of columns we select could be far from a uniform sample in the worst case. Second, the matrix we apply it to is not square and $\|\cdot\|_{(n/q)}$ norm only shrinks substantially when the columns of $\AA$ have the same $\ell^2_2$ norm, up to a constant. 

We therefore obtain a variant of Spectral norm decay for rectangular matrices, i.e. for any $n\times m$ matrix $\MM$ with roughly the same squared column norms, we show that 
\begin{equation}
\label{eqn:spectral_decay_rv}
    \expecf{}{\big\|\MM_{|\mathcal{Q}}\big\|_2} = O\left( \sqrt{\frac{q}{n}}\|\MM \|_2 + \sqrt{\log(q)/b}\|\MM \|_{(n/q)}\right)
\end{equation}
where $b = \ceil{n/m}$.  
To apply the above result, we then partition the rows of $\CC$ (since $\SS$ samples columns of $\CC$ as opposed to rows) into $\log(n)$ groups such that within each group, all rows have roughly the same squared norm. We then analyze leverage score sampling proportional to the column norms of $\ZZ^{\top}$ on each group independently. 
We show that we can obtain a coupling between the two random processes, namely uniform sampling and leverage score sampling, such that we obtain a decay bound similar to Equation \ref{eqn:spectral_decay_rv}, up to log factors. We describe our solution in more detail in Section \ref{sec:rel}. We note that our results extend to outputting a low-rank PSD matrix as well. 

\paragraph{Negative-Type Matrices.} We then use the techniques developed above to obtain an optimal relative-error low-rank approximation for Negative-Type distance matrices. While arbitrary metrics do not admit sublinear time algorithms for relative-error low-rank approximation (see Theorem 7.1 in \cite{bakshi2018sublinear}) Bakshi and Woodruff provided a sublinear time algorithm for metrics that satisfy negative-type inequalities. They obtain a $(1+\epsilon)$-relative-error approximation, that queries $\widetilde{O}(nk/\epsilon^{2.5})$ entries in the input. However, this algorithm outputs a bi-criteria solution, i.e., given a negative-type matrix $\AA$, it outputs a rank-$(k+4)$ matrix $\MM$ such that $\|\AA - \BB \|^2_F \leq (1+\epsilon)\|\AA - \AA_k \|^2_F$.

The key observation they make is that negative-type metrics can be realized as the distances corresponding to a point set $\mathcal{P}=\{x_1, x_2, \ldots x_n\}$ such that $\AA_{i,j}= \|x_i - x_j\|^2_2 = \|x_i \|^2_2 + \| x_j\|^2_2 - 2\langle x_i,x_j\rangle$. Therefore, $\AA$ admits the following decomposition: $\AA = \RR_1 + \RR_2 - 2\BB$, where for all $j\in[n]$, $(\RR_1)_{i,j}=\|x_i \|^2_2$, $\RR_2 = \RR_1^{\top}$ and $\BB$ is PSD. Observe that query access to $\AA$ suffices to obtain query access to $\BB$ by simply assuming w.l.o.g. that $x_1$ is centered at the origin and the $i$-th entry in the first row corresponds to $\|x_i\|^2_2$. Therefore, any PSD low-rank approximation algorithm can be simulated on the matrix $\BB$ by only having query access to $\AA$. 
Bakshi and Woodruff show that obtaining the low-rank approximation for $\BB$ and appending the column span of $\RR_1$ and $\RR_2$ to it results in a rank-$(k+4)$ bi-criteria approximation to $\AA$. The bi-criteria algorithm can be improved to $k+2$ using Cauchy's Interlacing Theorem \cite{fisk2005very} and observing $\RR_1, \RR_2$ are rank-$1$ updates to $\BB$, but this seems to be the limit of such approaches. 

We show here that our \textsf{SF} projection framework can be used to obtain a sample-optimal algorithm for negative-type metrics, and the bi-criteria approximation is not necessary. Recall, from our discussion above, that an \textsf{SF} projection for $\AA$ suffices to obtain a low-rank approximation for $\AA$. Our key observation is that we can use the techniques we developed for PSD matrices to obtain an \textsf{SF} projection, $\QQ \QQ^T$ , for $\BB$ (in the decomposition above), to which we append the column span of $\RR_1, \RR_2$ to $\QQ$, and the resulting projection (denoted by $\mathbf{\Omega}$) is an \textsf{SF} projection for $\AA$. To see this, observe, $\|\AA - \mathbf{\Omega} \AA_k\mathbf{\Omega} \|^2_F = \|\AA -\AA_k \|^2_F + \|\AA_k - \mathbf{\Omega}\AA_k\mathbf{\Omega} \|^2_F + 2 \trace{(\AA-\AA_k)(\II-\mathbf{\Omega})\AA_k\mathbf{\Omega}}$. A simple calculation using Von-Neumann's trace inequality bounds $\|\AA_k - \mathbf{\Omega}\AA_k\mathbf{\Omega} \|^2_F + 2 \trace{(\AA-\AA_k)(\II-\mathbf{\Omega})\AA_k\mathbf{\Omega}}$ by $O(k\|\AA(\II -\mathbf{\Omega}) \|^2_2)$. Since $\mathbf{\Omega}$ spans $\RR_1$ and $\RR_2$, and is an  \textsf{SF} projection for $\BB$, we can bound the above cost by $O(\epsilon/k)\|\BB- \BB_{k+2}\|_F$. It is easy to see that $\|\AA - \AA_k \|^2_F = O(\|\BB - \BB_{k+2} \|^2_F)$ and therefore, we conclude $\|\AA - \mathbf{\Omega} \AA_k\mathbf{\Omega} \|^2_F \leq (1 +O(\epsilon)) \| \AA - \AA_k \|^2_F$ (see Lemma \ref{lem:struct_proj_distance_matrix} for details). Subsequently, we use the sublinear algorithm we developed for PSD matrices to obtain a low-rank approximation for $\AA$.

\paragraph{Ridge Regression.} Our techniques also naturally extend to ridge regression, when the design matrix is PSD. This connection was originally outlined by Musco and Woodruff and they obtain sublinear time algorithms for solving ridge regression, parametrized by the statistical dimension $s_{\lambda}$. 
At a high level, we compute a rank-$(s_{\lambda}/\epsilon^2)$ 
spectral approximation to the input and solve ridge regression on the resulting matrix, 
i.e., given a PSD matrix $\AA$, we compute a low-rank matrix $\BB$ such that $\|\AA - \BB \|^2_2 \leq O(\epsilon/k)\|\AA - \AA_k \|^2_F$.
Further, we observe that the low-rank matrix is in fact a coreset for the input as it simultaneously preserves the cost of all $x$ and $y$. 

We then obtain a matching query lower bound for constructing coresets for ridge regression. Our lower bound proceeds by showing that a coreset can output a low-rank approximation on the instance of Musco and Woodruff with a stronger quadratic, rather than a linear dependence on $\epsilon$. Intuitively, the hard instance has multiple principle submatrices of all $1$s placed randomly over the matrix. Since a coreset simultaneously preserves the ridge regression cost for all $x, y$, it suffices to query the coreset on tuples of (scaled) eigenvectors and learn the positions of the blocks. However, a priori, we do not know what the eingenvectors of $\AA$ are. Instead, we query the coreset on every vector with a bounded support, and pick all vectors with small regression cost. We show that our resulting set only contains vectors which do not overlap much on the locations of the hidden blocks and we show this suffices.

\subsection{Robust Low-Rank Approximation}
The robustness model we consider is as follows: we begin with an $n \times n$ PSD matrix $\AA$. An adversary is then allowed to arbitrarily corrupt $\AA$ by adding a perturbation matrix $\NN$ such that $\|\NN\|^2_F \leq \eta \|\AA\|^2_F$ and for all $i \in[j]$, $\| \NN_{i,*}\|^2_2\leq c \|\AA_{i,*}\|^2_2$, for a fixed constant $c$. Note, while the adversary is unrestricted in the entries of $\AA$ that it corrupts, the Frobenius norm of the corruption is bounded in terms of the Frobenius norm of $\AA$ and the corruption is well-spread. The motivation for considering such a model is that many matrices that we observe in practice might be close but not exactly PSD, for instance, small perturbations to PSD matrices. 

It is impossible to obtain a relative-error low-rank approximation in this setting, since we cannot even identify $\AA$ after querying all $n^2$ entries of $\AA+\NN$. To see this, consider the case where $\AA$ is rank-$k$, and observe that a relative-error algorithm requires identifying $\AA$ exactly. 
However, by querying all entries of $\AA+\NN$, we can determine the row norms exactly. Therefore, we can run the algorithm of Frieze-Kannan-Vempala \cite{fkv04} to obtain a rank-$k$ matrix $\XX \YY^{\top}$ (in factored form) such that with probability at least $99/100$,

\begin{equation}
\label{eqn:rob1}
    \begin{split}
        \|\AA + \NN - \XX \YY^{\top}  \|^2_F &\leq \|\AA + \NN - (\AA + \NN)_k \|^2_F + \epsilon\|\AA + \NN \|^2_F \\
        & \leq \|\AA + \NN - \AA_k \|^2_F + \epsilon\|\AA + \NN \|^2_F \\
        & \leq \|\AA - \AA_k \|^2_F  + \|\NN\|^2_F + 2\langle\AA - \AA_k, \NN \rangle + (3+\eta)\epsilon\|\AA \|^2_F \\
        & \leq \|\AA - \AA_k \|^2_F + O(\epsilon+\sqrt{\eta}) \|\AA\|^2_F
    \end{split}
\end{equation}
where the second inequality follows from $(\AA+\NN)_k$ being the best rank-$k$ approximation to $\AA+\NN$ and $\AA_k$ is any other rank-$k$ matrix. The third inequality uses $\|\AA +\NN\|^2_F \leq 2(\|\AA\|^2_F+\|\NN\|^2_F)$, which follows from $\ell^2_2$ distance satisfying triangle-inequality up to a factor of $2$. The last inequaliity uses Cauchy-Schwarz on $2|\langle\AA - \AA_k, \NN \rangle| \leq 2\|\AA \|_F\cdot \|\NN \|_F \leq 2\sqrt{\eta}\|\AA\|^2_F$, which follows from the assumption on $\NN$. Additionally 
\begin{equation}
\label{eqn:rob2}
    \begin{split}
        \|\AA + \NN - \XX \YY^{\top}  \|^2_F &= \|\AA  - \XX \YY^{\top}\|^2_F + \| \NN \|^2_F +2 \langle\AA -\XX\YY^{\top} , \NN\rangle \\
        & \geq \|\AA  - \XX \YY^{\top}\|^2_F  - 2\sqrt{\eta}\|\AA \|^2_F
    \end{split}
\end{equation}
Combining Equations \ref{eqn:rob1} and \ref{eqn:rob2}, we have $
\|\AA  - \XX \YY^{\top}\|^2_F \leq \|\AA - \AA_k \|^2_F   + O(\epsilon+ \sqrt{\eta})\|\AA \|^2_F$.
While this algorithm is far from optimal in terms of sample complexity, it indicates that relaxing our guarantees to additive-error is amenable to robust algorithms and indicates why we pick up a $\sqrt{\eta}$ term. 
The central question we focus on in this section is whether there exists a \textit{robust sublinear time and query algorithm} to obtain an additive-error low-rank approximation for PSD matrices. 

We begin by showing a sample complexity lower bound if $\AA$ is an arbitrary PSD matrix. The intuition from the relative-error setting still applies and the diagonal entries are crucial for sublinear algorithms. 
In tune with this intuition, the adversary corrupts large diagonal entries to decrease their magnitude and thus obfuscate rows that contain large off-diagonal entries. We therefore parameterize our lower bound and algorithms by the largest ratio between a diagonal entry of $\AA$ and $\AA+\NN$, denoted by $\phi_{\max} = \max_{j \in [n]} \AA_{j,j} / |(\AA+\NN)_{j,j}|$.  Recall, we obtain the following lower bound:

\vspace{0.1in} 
\noindent \textbf{Theorem \ref{thm:robust_lb}.} \textit{(Informal lower bound.) 
Let $\epsilon>\eta>0$.
Given $\AA + \NN$ such that $\AA$ is PSD and  $\NN$ is a corruption matrix as defined above, any randomized algorithm, that with probability at least $2/3$, outputs a rank-$k$ approximation up to additive error $(\epsilon+\eta)\|\AA \|^2_F$ must read  $\Omega\left( \phi_{\max}^2 nk/\epsilon \right)$ entries of $\AA+\NN$.
}
\vspace{0.1in} 

In our hard instance, we have a block matrix $\AA$, where we place a random $\epsilon/\eta \times \epsilon/\eta$, rank-$1$, non-contiguous block $\BB_1$ such that each entry in the block is $\sqrt{\eta^2 n/\epsilon}$ and the remaining matrix has $1$s on the diagonals and $0$s everywhere else. It is easy to see this matrix is PSD.
We observe that the block $\BB_1$ contributes an $\epsilon$-fraction of the Frobenius norm of $\AA$, and the $\ell^2_2$ norm of the diagonals is an $\eta$-fraction of the Frobenius norm of $\AA$. Therefore, the adversary can afford to corrupt all the diagonal entries in $\BB_1$ and set them to be $1$. Such a perturbation is feasible in our model and successfully obfuscates the large off-diagonal entries. Note, for this perturbation $\phi^2_{max} = \eta^2 n/\epsilon $.

Let the resulting matrix be denoted by $\AA + \NN$. Here, we observe any $\epsilon$-additive-error low-rank approximation cannot ignore the block $\BB_1$.  Since the diagonals of $\AA + \NN$ now provide no information about the off-diagonal entries, any algorithm that correctly outputs a low-rank approximation for both $\AA +\NN$ and $\II$ must detect at least one entry in $\BB_1$. Since $\BB_1$ has  $\epsilon^2/\eta^2$ non-zeros, any algorithm must query $\Omega(\eta^2n^2/\epsilon^2) = \Omega(\phi^2_{\max}n/\epsilon)$ entries to detect one entry. To obtain a linear dependence on $k$, we simply create $k$ independent copies of $\BB_1$. 

\vspace{0.1in}
\textbf{Robust Algorithm.}
Next, we focus on a robust, additive-error low-rank approximation algorithm, where the sample complexity is parameterized by $\phi_{\max}$. 
We begin by introducing a new sampling procedure to construct \textit{projection-cost preserving sketches}. Our construction is simple to state: we sample each column proportional to the corresponding diagonal entry. Computing these sampling probabilities \textit{exactly} requires reading only $n$ entries in $\AA + \NN$. We show that sampling $\widetilde{O}\left(\phi_{\max}^2 \sqrt{n}k^2/\epsilon^2 \right)$ columns proportional to this distribution preserves the projection of the columns of $\AA$ onto the orthogonal complement of any rank-$k$ subspace, up to additive error $(\epsilon+ \sqrt{\eta})\|\AA \|^2_F$.  

\vspace{0.1in} 
\noindent \textbf{Theorem \ref{thm:pcp_main_col}.} \textit{(Informal Robust Column PCP.)
Let $\AA+\NN$ be an $n \times n$ matrix following the assumptions of our noise model. Let $k \in[n]$ and $\epsilon , \sqrt{\eta}>0$. Let $q = \{q_1, q_2 \ldots q_n\}$ be a probability distribution over the columns of $\AA$ such that $q_j = (\AA+\NN)_{j,j}/\trace{\AA+\NN}$.  Construct $\CC$ by sampling $\widetilde{O}(\phi^2_{\max}\sqrt{n}k^2/\epsilon^2)$ columns of $\AA+\NN$ proportional to $q$ and rescaling appropriately. Then, with probability at least $1-c$, for any rank-$k$ orthogonal projection $\XX$, 
\[
\| \CC - \XX\CC \|^2_F = \| \AA - \XX\AA \|^2_F \pm (\epsilon+\sqrt{\eta})\| \AA \|^2_F
\]
}
We note that all prior PCP constructions work in the noiseless setting.
As a comparison, the construction of Cohen et. al. \cite{cohenmm17} works for arbitrary $\AA$, but requires $nnz(\AA)$ time and queries to compute the approximate ridge-leverage scores of $\AA$.  Musco and Musco  \cite{musco2017recursive} describe how to approximately compute the ridge leverage scores of $\AA^{1/2}$ (if $\AA$ is PSD) using the Nystrom approximation,  where $\AA = \AA^{1/2}\cdot \AA^{1/2}$. Musco and Woodruff \cite{mw17} use this method to compute the ridge leverage scores of $\AA^{1/2}$ with $\Theta(nk)$ queries and show that the this provides a $(\sqrt{n/k})$-approximation to the ridge leverage scores of $\AA$. We note that the guarantees obtained by \cite{cohenmm17,mw17} are relative error, as opposed to the additive error guarantee in the theorem above. Finally, Bakshi and Woodruff \cite{bakshi2018sublinear} provide an additive-error sublinear time construction for \textit{distance matrices} by sampling proportional to column norms. In all the aforementioned constructions, computing the sampling distribution is a non-trivial task, whereas we simply sample proportional to the diagonal entries.     

We observe that we sample columns of $\AA+\NN$, to obtain $\CC$ which is an unbiased estimator for $\|\AA+\NN\|^2_F$. 
The main technical challenge in our construction is to relate the cost of rank-$k$ projections for the column space of $\AA$ to that of $\CC$, while obtaining an optimal dependence on $n$ and $k$. Note, while we do not obtain the correct dependence on $\epsilon$, we do not have to explicitly compute all of $\CC$, only a subset of it.   

We then extend the diagonal sampling algorithm to construct a robust row 
PCP for the matrix $\CC$. We note that the 
construction for $\AA$ does not immediately give a row PCP for $\CC$ since $\CC$ is no longer a corrupted PSD matrix or even a square matrix, and thus there is no notion of a diagonal. Here, all previous approaches to construct a PCP with a 
sublinear number of queries hit a roadblock, since the matrix $\CC$ need not have any well-defined structure apart from being a scaled subset of the columns of the 
original corrupted PSD matrix $\AA+\NN$. However, we show that sampling rows of $\CC$
proportional to the diagonal entries of $\AA+\NN$ results in a row PCP for $\CC$. 

\vspace{0.1in} 
\noindent \textbf{Theorem \ref{thm:pcp_main_row}.} \textit{(Informal Robust Row PCP.)
Let $\AA+\NN$ be an $n \times n$ matrix corresponding to our noise model and let $\CC$ be a column PCP for $\AA$ as defined above. Let $p = \{p_1, p_2 \ldots p_n\}$ be a probability distribution over the rows of $\CC$ such that $p_j = (\AA+\NN)_{j,j}/\trace{\AA+\NN}$. Construct $\RR$ by sampling $\widetilde{O}(\phi_{\max}\sqrt{n}k^2/\epsilon^2)$ rows of $\CC$ proportional to $p$ and scaling appropriately. With probability at least $1-c$, for any rank-$k$ orthogonal projection $\XX$,
\[
\| \RR - \RR\XX\|^2_F = \| \CC - \CC\XX\|^2_F \pm (\epsilon +\sqrt{\eta} )\| \AA \|^2_F
\]
}

For our algorithm, we begin by constructing column and row PCPs of $\AA+\NN$, to obtain a $t \times t$ matrix $\RR$, where $t = \widetilde{O}(\phi_{\max} \sqrt{n}k^2/\epsilon^2)$. Instead of reading the entire matrix, we uniformly sample $\epsilon^3 t/k^3$ entries in each row of $\RR$, and query these entries. Note, this corresponds to reading $\epsilon^3 t^2/k^3 = \widetilde{O}(\phi^2_{\max}nk/\epsilon)$ entries in $\AA+\NN$. Ideally we would want to estimate the $\ell^2_2$ norms of each row of $\RR$ to then use a result of Frieze-Kannan-Vempala to obtain a low-rank approximation for $\RR$ \cite{fkv04}. 
It is well known that to recover a low-rank approximation for $\RR$, one can sample rows of $\RR$ proportional to row norm estimates, denoted by $\mathcal{Y}_i$ \cite{fkv04}. As shown in \cite{indyk2019sample} the following two conditions are a relaxation of those required in \cite{fkv04}, and suffice to obtain an additive error low-rank approximation :   
\begin{enumerate}
	\item For all $i \in [t]$, the corresponding estimate over-estimates the row norm of $\RR_{i,*}$, i.e., $\mathcal{Y}_i \geq \|\RR_{i,*} \|^2_2$.
	\item The sum of the over-estimates is not too much larger than the Frobenius norm of the matrix, i.e., $\sum_{i \in [t]} \mathcal{Y}_i \leq \phi^2_{\max} n/t \| \RR \|^2_F$
\end{enumerate}
If the two conditions are satisfied, Frieze-Kannan-Vempala implies sampling $s$ rows of $\RR$ proportional to $\mathcal{Y}_i$ results in an $s \times t$ matrix $\SS$ such that the row space of $\SS$ contains a good rank-$k$ approximation, where $s = O(\phi^2_{\max}nk/\epsilon t)$, which matches our desired sample complexity. 
We show that, unfortunately, such a guarantee is not possible even in the uncorrupted case, where $\NN=0$ and $\phi_{\max}=1$. Intuitively, $\RR$ may be a sparse matrix such that many rows have large norm, and uniform sampling cannot obtain concentration for all such rows, as required by the aforementioned conditions. 

Instead, we settle for a weaker statement, where we show that the estimator obtained by uniform sampling in each row is accurate with $o(1)$ probability. At a high level, we show that we can design a sampling process that is statistically close to $\ell^2_2$ sampling described by Frieze-Kannan-Vempala \cite{fkv04}. We then open up the analyzes of Frieze-Kannan-Vempala and show that our sampling process suffices to recover the low-rank approximation guarantee. Given the flurry of recent work in quantum computing \cite{kerenidis2016quantum, chia2018quantum, tang2019quantum, rebentrost2018quantum, gilyen2018quantum} that uses Frieze-Kannan-Vempala $\ell_2^2$ sampling as a key algorithmic primitive, our analysis may be of independent interest. 

\vspace{0.1in} 
\noindent \textbf{Lemma \ref{lem:large_norm}.} \textit{(Informal Estimation of Row Norms.)
Let $\RR \in \mathbb{R}^{t \times t }$ be the row PCP as defined above. For all $i \in [t]$ let $\X_i = \sum_{j \in [\epsilon^3 t/k^3]} \X_{i,j}$ such that $\X_{i,j} = k^3\RR^2_{i,j'}/\epsilon^3$ with probability $1/t$, for all $j' \in [t]$. Then, for all $i \in [t]$, $\X_i = \left(1\pm 0.1\right)\|\RR_{i,*}\|^2_2$ with probability at least $\min(\|\RR_{i,*} \|^2_2 k/\epsilon n,1)$.
}
\vspace{0.1in}

We now face two major challenges: first, the probability with which 
the estimators are accurate is too small to even detect all rows with norm larger than $\phi_{\max}^2 n\| \RR \|^2_F/t^2$, and second, there is no small query certificate for when an estimator is accurate in 
estimating the row norms. Therefore, we cannot even identify the rows
where we obtain an accurate estimate of norm.   

To address the first issue, we make the crucial observation that while we cannot estimate the norm of each row accurately, we can hope to sample the row with the same probability as Frieze-Kannan-Vempala \cite{fkv04}. Recall, their algorithm requires sampling row $\RR_{i,*}$ with probability at least $\|\RR_{i,*}\|^2_2/\|\RR\|^2_F$, which matches the probability in Lemma \ref{lem:large_norm}. Therefore, we can focus on designing a weaker notion of identifiability, that may potentially include extra rows.  

We begin by partitioning the rows of $\RR$ into two sets. Let $\mathcal{H} = \left\{ i \textrm{ }\big|\textrm{ } \|\RR_{i,*} \|^2_2 \geq \phi_{\max}^2n/t^2 \|\RR \|^2_F \right\}$ be the set of heavy rows and $[t] \setminus \mathcal{H}$ be the remaining rows. Note, $|\mathcal{H}| = O\hspace{-0.1cm}\left(t^2/\phi_{\max}^2n\right) =O\hspace{-0.1cm}\left(k^4 \log^4(n)/\epsilon^4 \right)$. 
We then condition on our estimator having norm at least $\phi_{\max}^2n\|\RR \|^2_F/t^2$. Conditioned on this event, we sample the corresponding row of $\RR$ with probability $1$. As before, we want to prevent sampling too many spurious rows, but we show only a subset of the rows in $\mathcal{H}$ satisfy this condition. 
This ensures we identify rows in $\mathcal{H}$ with the right probability. 
For all the remaining rows, we know the norm is at most 
$\phi_{\max}^2n/t^2\|\RR\|^2_F$. We show that uniformly sampling $\phi_{\max}^2n/t$ such rows suffices to simulate row norm sampling. 

We then open up the analysis of 
Frieze-Kannan-Vemapala to show that the above sampling procedure suffices to bound the overall variance, resulting in a relaxation of the conditions required to obtain an additive error low-rank approximation to $\RR$. 
Once we compute a good low-rank approximation for $\RR$ we can follow the approach of \cite{cohenmm17,mw17,bakshi2018sublinear}, where we set up two regression problems, and use the sketch and solve paradigm to compute a low-rank approximation for $\AA$, culminating in Theorem \ref{thm:main_thm}.

For corrupted correlation matrices, we observe that the true uncorrupted matrix has all diagonal entries equal to $1$. Therefore, we can discard the diagonal entries of $\AA + \NN$ and assume they are $1$. 
In this case, no matter what the adversary does to the diagonal,  $\phi_{\max}=1$ and we obtain an  $\widetilde{O}(nk/\epsilon)$ query algorithm that satisfies the above guarantee. Further, we show a matching sample complexity lower bound of $\Omega(nk/\epsilon)$, to obtain $\epsilon$-additive-error, even in the presence of no noise. 


\section{Preliminaries and Notation}
Given an $m \times n$ matrix $\AA$ with rank $r$, we can compute its 
singular value decomposition, denoted by $\texttt{SVD}(\AA) = \UU 
\mathbf{\Sigma} \VV^{\top}$, such that $\UU$ is an $m \times r$ matrix with 
orthonormal columns, $\VV^{\top}$ is an $r \times n$ matrix with orthonormal 
rows and $\mathbf{\Sigma}$ is an $r \times r$ diagonal matrix. The entries 
along the diagonal are the singular values of $\AA$, denoted by 
$\sigma_1, \sigma_2 \ldots \sigma_r$. Given an integer $k \leq r$, we 
define the truncated singular value decomposition of $\AA$ that zeros out 
all but the top $k$ singular values of $\AA$, i.e.,  $\AA_k = \UU 
\mathbf{\Sigma}_k \VV^{\top}$, where $\mathbf{\Sigma}_k$ has only $k$ non-zero 
entries along the diagonal. It is well known that the truncated SVD 
computes the best rank-$k$ approximation to $\AA$ under the Frobenius 
norm, i.e., $\AA_k = \min_{\textrm{rank}(\XX)\leq k} \| \AA - \XX\|_F$.
More generally, for any matrix $\MM$, we use the notation $\MM_k$ and 
$\MM_{\setminus k}$ to denote the first $k$ components and all but the 
first $k$ components respectively. We use $\MM_{i,*}$ and $\MM_{*,j}$ to 
refer to the $i^{th}$ row and $j^{th}$ column of $\MM$ respectively. 
For an $n \times n$ PSD matrix $\AA$, we denote the singular (eigenvalue) decomposition by $\UU \Sig \UU^{\top}$. Further, since $\Sig_{i,i} \geq 0$, let $\AA^{1/2} = \UU \Sig^{\frac{1}{2}}\UU^{\top}$ be the square root of $\AA$. Note that $\AA_{i,j} = \langle\AA^{1/2}_{i,*}, \AA^{1/2}_{j,*} \rangle$. By Cauchy-Schwarz, for all $i,j \in[n]$, $\AA^2_{i,j} =\langle\AA^{1/2}_{i,*}, \AA^{1/2}_{j,*} \rangle^2 \leq \|\AA^{1/2}_{i,*}\|^2_2\cdot \|\AA^{1/2}_{j,*}\|^2_2 =\AA_{i,i} \cdot \AA_{j,j}$. We use $\textsf{nnz}(\AA)$ to denote the number of non-zero entries (sparsity) of $\AA$.  We use operator and spectral norm interchangeably to denote  $\|\MM \|_2 = \max_{\|y\|_2=1}\|\MM y\|_2$.  We also use the notation $\MM^{\dagger}$ to denote the Moore-Penrose pseudoinverse.

\section{Relative Error PSD Low-Rank Approximation}\label{sec:rel}

In this section, we describe our main algorithm for \textit{relative-error} PSD Low-Rank Approximation, where we query only $\widetilde{O}(nk/\epsilon)$ of the input matrix $\AA$. This improves the best known algorithm by Musco and Woodruff that queries $\widetilde{O}(nk/\epsilon^{2.5})$ and matches their query lower bound of $\Omega(nk/\epsilon)$ up to polylogarithmic factors \cite{mw17}. Formally, we prove the following: 

\begin{theorem}(Sample-Optimal PSD Low-Rank Approximation.)
\label{thm:sample_opt_psd_lra}
Given an $n \times n$ PSD matrix $\AA$, an integer $k$, and $1>\epsilon >0$, Algorithm \ref{alg:sample_opt_rel} samples $\widetilde{O}(nk/\epsilon)$ entries in $\AA$ and outputs matrices $\MM, \NN^{\top} \in \mathbb{R}^{n \times k}$ such that with probability at least $9/10$, 
\begin{equation*}
    \|\AA - \MM \NN\|^2_F \leq (1+\epsilon) \|\AA - \AA_k\|^2_F
\end{equation*}
Further, the algorithm runs in $\widetilde{O}(n(k/\epsilon)^{\omega-1} +  (k/\epsilon^3)^\omega)$ time. 
\end{theorem}

We begin by defining various statistical quantities associated with a given matrix, such as the leverage and ridge-leverage scores. The leverage score of a given row measures the importance of this row in composing the row span. Leverage scores have found numerous applications in regression, preconditioning, linear programming and graph sparsification \cite{sarlos2006improved, spielman2011graph, lee2015efficient, cohen2015uniform}. In the special case of graphs, they are referred to as \textit{effective resistances}. 

\begin{definition}(\textit{Leverage Scores.})
Given a matrix $\MM \in \mathbb{R}^{n \times m}$, let $\textbf{m}_i = \MM_{i,*}$ be the $i$-th row of $\MM$. Then, for all $i \in [n]$ the $i$-th row leverage score of $\MM$ is given by
\begin{equation*}
    \tau_i(\MM) = \textbf{m}_i (\MM^{\top} \MM)^{\dagger} \textbf{m}_i^{\top}
\end{equation*}

\end{definition}

The column leverage scores can be defined analogously. 
Note, in the special case where $\MM$ has orthonormal columns, the row leverage scores of $\MM$ are simply the $\ell^2_2$ norms of the rows i.e., $\tau_i(\MM)=\| \textbf{m}_{i} \|^2_2$.  It is well-known that sampling rows of a matrix proportional to the leverage scores satisfies the subspace embedding property (Spectral Sparsification for Graphs) and leads to faster algorithms for $\ell_2$-norm Regression. Recall, for an $n \times m$ matrix $\AA$, a leverage score sampling matrix $\SS = \DD \mathbf{\Omega}^{\top}$, where $\DD$ is a $t \times t$ diagonal matrix and $\mathbf{\Omega}$ is an $n \times t$ sampling matrix. For all $j \in [t]$, select row index $i \in [n]$ with probability $p_i = \tau_i(\AA)/ \sum_i \tau_i(\AA)$ and set $\mathbf{\Omega}_{i,j}=1$ and $\DD_{j,j} = 1/\sqrt{t p_i}$.  

\begin{lemma}(Subspace Embedding.)
\label{lem:subspace_embedding}
Given a matrix $\AA \in \mathbb{R}^{n \times m}$, $\epsilon>0$, and a leverage score sampling matrix $\SS$ with $t =O(mlog(m)/\epsilon^2)$ rows, with probability at least $99/100$, for all $x \in \mathbb{R}^{m}$
\begin{equation*}
    \|\SS \AA x \|^2_2 = (1\pm \epsilon) \|\AA x\|^2_2
\end{equation*}
\end{lemma}
This simply follows from an application of the Matrix Chernoff bound. Observe that the sketch preserves all the singular values of $\AA$ up to a factor of $1\pm\epsilon$. We refer the reader to a recent survey for more details \cite{w14}. 
Next, we recall that leverage score sampling results in a fast algorithm for regression. 

\begin{lemma}(Fast Regression, Theorem 38 \cite{clarkson2013low}.) 
\label{lem:fast_regression}
Given matrices $\AA \in \mathbb{R}^{n \times m}, \BB \in \mathbb{R}^{n \times d}$ such that rank$(\AA) \leq r$ and $\epsilon>0$, sample $O(r\log(r) + r/\epsilon)$ rows of $\AA, \BB$ proportional to the leverage scores of $\AA$ to obtain a sketch $\SS$ such that $\YY^* = \arg\min_{\YY} \| \SS \AA \YY - \SS \BB \|^2_F$. Then, with probability at least $1-c$,
\begin{equation*}
    \| \AA \YY^* - \BB \|^2_F \leq (1+\epsilon)\min_{\YY}\|\AA\YY -\BB \|^2_F
\end{equation*}
for a fixed small constant $c$.
Further, the time to compute $\YY^*$ is $O(\textsf{nnz}(\AA)\log(r/\epsilon) + (n+d)(r/\epsilon)^{\omega-1} + \textrm{poly}(r/\epsilon))$.
\end{lemma}

Note, the terms in the running time follow from using Cohen's construction for OSNAP \cite{cohen2016nearly}.
Leverage score sampling matrices also approximately preserve norms in affine spaces, which leads to faster algorithms for multi-response regression, i.e., $\min_{\XX} \|\AA\XX - \BB\|^2_F$, where $\BB$ now has a large number of columns. 

\begin{lemma}(Affine Embeddings, Theorem 39 \cite{clarkson2013low}.)
\label{lem:affine_embedding}
Given matrices $\AA\in \mathbb{R}^{n \times m}$ , such that $\textrm{rank}(\AA) =r$, and $\BB \in \mathbb{R}^{n \times d}$, let $\SS$ be a leverage score sampling matrix with $t = O(r/\epsilon^2)$ rows. Further, let $\XX^*$ be the optimizer for $\min_{\XX}\|\AA\XX - \BB \|^2_F$ and let $\BB^* = \AA\XX^* - \BB$. Then, with probability at least $1-c$, for all $\XX \in \mathbb{R}^{m \times d}$
\begin{equation*}
    \|\SS\AA\XX - \SS \BB \|^2_F - \|\SS\BB^*\|^2_F = (1\pm \epsilon) \|\AA \XX - \BB \|^2_F - \|\BB^* \|^2_F
\end{equation*}
for a fixed small constant $c$. 
\end{lemma}
An important application of the above lemma (which we use extensively) is to sketch constrained regression problems, for example, when the matrix $\XX$ has a fixed small rank. Since affine embeddings approximately preserve the cost of all affine spaces up to a fixed shift, this guarantee in particular holds for $\XX$ with small rank. Recall, an important caveat here is that the cost of the sketched problem is not a relative-error approximation to the cost of the original problem since we cannot estimate $\| \BB^*\|^2_F$ in general. However, the upshot here is that the aforementioned guarantee still suffices for optimization since the fixed shift does not change the optimizer. 

The next tool we use is input-sparsity time low-rank approximation. This was achieved by Clarkson and Woodruff \cite{clarkson2013low} and the exact dependence on $k,\epsilon$ was improved in subsequent works \cite{mm13,NN13,BDN15,cohen2016nearly}. While the standard low-rank approximation guarantee achieves relative-error under Frobenius norm, here we will require a spectral norm bound, which follows from results of \cite{cohen2015dimensionality, cohenmm17}.

\begin{lemma}(Input-Sparsity Spectral LRA \cite{cohen2015dimensionality, cohenmm17}.)
\label{lem:input_sparsity_spectral_lra}
Given a matrix $\AA \in \mathbb{R}^{n \times d}$, $\epsilon, \delta >0$ and $k \in \mathbb{N}$, let $k' =O( k/\epsilon)$. Then, there exists an algorithm that outputs a matrix $\ZZ^{\top} \in \mathbb{R}^{k' \times n}$ such that with probability at least $1-\delta$,
\begin{equation*}
    \|\AA - \AA \ZZ \ZZ^{\top} \|^2_2 \leq O\left(\frac{\epsilon}{k}\right) \| \AA - \AA_{k/\epsilon}\|^2_F
\end{equation*}
in time and query complexity $\widetilde{O}\left(\textsf{nnz}(\AA) + (n+d)\textrm{poly}(k/\epsilon \delta) \right)$.
\end{lemma}
\begin{proof}
By Lemma 18 from \cite{cohen2015dimensionality} it suffices to use any obvious subspace embedding matrix with $\epsilon=O(1)$ and $k=k/\epsilon$. Here, we use OSNAP in the regime that requires $\widetilde{O}(k/\epsilon^2)$ rows and sparsity $\textrm{polylog}(k)/\epsilon$ \cite{NN13}. Instantiating this OSNAP construction with $\epsilon=O(1)$ and $k=k/\epsilon$ results in $\ZZ^{\top}$ with $k/\epsilon$ rows in the desired running time.  
\end{proof}

Next we define the ridge leverage scores of a matrix. The ridge leverage scores were used as sampling probabilities in the context of linear regression and spectral approximation \cite{li2013iterative, kapralov2017single, alaoui2015fast}, and low-rank approximation \cite{cohenmm17, mw17}.  Intuitively, the ridge leverage scores can be thought of as adding a regularization term that attenuates the smaller singular directions such that they are sampled with proportionately lower probability. 

\begin{definition}\textit{(Ridge Leverage Scores.)}
\label{def:ridge_leverage_scores}
Given a matrix $\MM \in \mathbb{R}^{n \times m}$ and an integer $k$, let $\textbf{m}_i = \MM_{i,*}$ be the $i$-th row of $\MM$. Then, for all $i \in[n]$, the $i$-th rank-$k$ ridge leverage score of $\MM$ is 
\begin{equation*}
    \rho^{k}_{i}(\MM) = \textbf{m}_i\left(\MM^{\top} \MM + \frac{\|\MM - \MM_k \|^2_F}{k}\II\right)^{\dagger} \textbf{m}^{\top}_i
\end{equation*}
\end{definition}
Since we typically use the row ridge leverage scores to define a probability distribution over the rows and sample according to this distribution, it is crucial that their sum is small as this controls the number of rows we would need to sample. This follows from a straightforward calculation: 
\begin{lemma}(Lemma 4 from \cite{cohenmm17}.) Let $\rho^k_i(\MM)$ be the $i$-th ridge leverage score of $\MM$. Then, $$\sum_{i \in [n]} \rho^k_i(\MM) \leq 2k$$
\end{lemma}

Cohen et. al. \cite{cohenmm17} show that the ridge leverage scores of a matrix can be approximated up to a small constant in $O(\textsf{nnz}(\AA))$ time, however this involves reading the entire matrix $\AA$. For the special case of $\AA$ being PSD, Musco and Musco \cite{musco2017recursive} show that the ridge leverage scores of $\AA^{\frac{1}{2}}$ can be approximated up to a small constant using a so-called Nystr{\"o}m approximation. 

\begin{lemma}(Lemma 4 of \cite{mw17}.)
\label{lem:approx_ridge_leverage}
Given a PSD matrix $\AA \in \mathbb{R}^{n \times n}$ and integer $k$, there exists an algorithm that accesses $O(nk\log(k/\delta))$ entries in $\AA$ and computes $\hat{\rho}^k_i(\AA^{\frac{1}{2}})$ for all $i \in [n]$, such that with probability $1-\delta$,
\begin{equation*}
    \rho^k_i(\AA^{1/2}) \leq \hat{\rho}^k_i(\AA^{1/2}) \leq 3 \rho^k_i(\AA^{1/2})
\end{equation*}
and runs in time $O(n (k\log(k/\delta))^{\omega-1})$, where $\omega$ is the matrix multiplication exponent. 
\end{lemma}

Note, while it is not known how to compute ridge leverage scores of a PSD matrix in sublinear time, Musco and Woodruff \cite{mw17} show that the ridge leverage scores of $\AA^{\frac{1}{2}}$ are a coarse approximation to the ridge leverage scores of $\AA$. 

\begin{lemma}(Lemma 5 in \cite{mw17}.)
\label{lem:a_half_to_a}
Given a PSD matrix $\AA \in \mathbb{R}^{n \times n}$, for all $i \in [n]$,
\begin{equation*}
    \rho^{k}_{i}(\AA) \leq 2 \sqrt{\frac{n}{k}} \rho^{k}_{i}(\AA^{\frac{1}{2}}) 
\end{equation*}
\end{lemma}

Musco and Woodruff then show that sampling columns of $\AA$, according to the corresponding ridge leverage scores of $\AA^{\frac{1}{2}}$, suffices to obtain a column projection-cost preserving sketch (PCP), if we oversample by a $\sqrt{n/k}$ factor. Projection-cost preserving sketches were introduced by Feldmen et. al. \cite{feldman2013turning} and Cohen et. al. \cite{cohen2015dimensionality} and studied in the context of low-rank approximation in \cite{cohenmm17,  mw17, bakshi2018sublinear}. 

\begin{lemma}(Column PCP from \cite{mw17}.)
\label{lem:rel_col_pcp}
Given a PSD matrix $\AA \in \mathbb{R}^{n \times n}$, integer $k$ and $\epsilon>0$, for all $j \in [n]$ let $\bar{\rho}^k_j(\AA^{\frac{1}{2}})$ be a constant approximation to the column-ridge leverage scores of $\AA^{\frac{1}{2}}$. Let $q_j =  \bar{\rho}^k_j(\AA^{\frac{1}{2}})/\sum_j \bar{\rho}^k_j(\AA^{\frac{1}{2}})$ and let $t = O\left(\sqrt{\frac{n}{k}}\sum_j \bar{\rho}^k_j(\AA^{\frac{1}{2}})\log(k/\delta)/\epsilon^2\right) = O\left(\sqrt{nk}\log(k/\delta)/\epsilon^2\right)$. Construct $\CC \in \mathbb{R}^{n \times t}$ by sampling $t$ columns of $\AA$ and setting each one to be $\frac{1}{\sqrt{t q_j}} \AA_{*,j}$ with probability $q_j$. Then, with probability $1-\delta$, for any rank-$k$ projection matrix $\XX \in \mathbb{R}^{n \times n}$,
\begin{equation*}
    (1-\epsilon)\|\AA - \XX \AA \|^2_F \leq \|\CC - \XX\CC \|^2_F \leq (1+\epsilon)\|\AA - \XX \AA \|^2_F 
\end{equation*}
Further, such a $\CC$ can be computed by accessing $\widetilde{O}(nk)$ entries in $\AA$ and in time $O(nk^{\omega-1})$.
\end{lemma}

This result also implies that the resulting matrix $\CC$ is a \textit{Spectral-Frobenius PCP} for $\AA$ (Lemma 24 in \cite{mw17}), i.e., for any rank-$k$ projection matrix $\XX$,
\begin{equation}
    (1-\epsilon)\|\AA - \XX \AA \|^2_2 -\frac{\epsilon}{k}\|\AA -\AA_k\|^2_F \leq \|\CC - \XX\CC\|^2_2 \leq (1+\epsilon)\|\AA - \XX \AA \|^2_2 +\frac{\epsilon}{k}\|\AA -\AA_k\|^2_F
\end{equation}
As noted by Musco and Woodruff, the resulting matrix $\CC$ is not even square 
and thus is it unclear how to sample rows of $\CC$ to obtain a row-PCP in 
sublinear time and queries. In particular, the ridge-leverage scores of rows of $\CC$ can be an $n/k$-factor larger than the corresponding  
ridge-leverage scores of $\AA^{\frac{1}{2}}$. Instead, Musco and Woodruff sample 
rows of $\CC$ proportional to the rank-$k/\epsilon^2$ ridge leverage scores of $\AA^{\frac{1}{2}}$. In addition, they show the stronger guarantee that a \textit{Spectral-Frobenius PCP} holds (by Lemma 8 of \cite{mw17}) for PSD Matrices.    

\begin{lemma}(Spectral-Frobenius PCP.)
\label{lem:rel_row_pcp_mixed}
Given a PSD matrix $\AA \in \mathbb{R}^{n\times n}$, an integer $k$ and $\epsilon>0$, let $\CC \in \mathbb{R}^{n\times t}$ be a column PCP for $\AA$, following Lemma \ref{lem:rel_col_pcp}. Let $k' = k/\epsilon^2$. For all $i \in [n]$, let $\bar{\rho}^{k'}_i(\AA^{\frac{1}{2}})$ be a constant approximation to the rank-$k'$ row ridge leverage scores of $\AA^{\frac{1}{2}}$. Let $p_i = \bar{\rho}^{k'}_i(\AA^{\frac{1}{2}})/ \sum_i \bar{\rho}^{k'}_i(\AA^{\frac{1}{2}})$ and let $t = O\left(\sqrt{\frac{n}{k}}\sum_i \bar{\rho}^{k'}_i(\AA^{\frac{1}{2}}) \log(n)/\epsilon \right) = O\left(\sqrt{nk}\log(n)/\epsilon^3\right)$. Then, with probability $1-c$, for all rank-$k'$ projection matrices $\XX$,
\begin{equation*}
    (1-\epsilon)\|\CC - \CC \XX \|^2_2 -\frac{\epsilon}{k}\|\AA -\AA_k\|^2_F \leq \|\RR - \RR\XX\|^2_2 \leq (1+\epsilon)\|\CC - \CC \XX \|^2_2 +\frac{\epsilon}{k}\|\AA -\AA_k\|^2_F
\end{equation*}
\end{lemma}

We observe that we could compute a low-rank approximation to $\RR$ in input sparsity time, which already requires querying $\Omega(\textsf{nnz}(\RR)) = \Omega(nk/\epsilon^4)$ entries in $\AA$  and is far from optimal in terms of the dependence on $\epsilon$. It is here that we digress from the approach of Musco and Woodruff. We observe that the dependence on $n$ and $k$ is optimal and thus we instantiate the aforementioned column and row PCPs with $\epsilon=O(1)$ and $k = k/\epsilon$. While this results in weaker PCP guaranteees, the resulting matrix $\RR$ is a $\sqrt{nk/\epsilon} \times \sqrt{nk/\epsilon}$ matrix and we can now afford to read all of it and thus we can compute a rank-$k$ low-rank approximation to $\RR$ using the input sparsity time algorithm of Clarkson and Woodruff \cite{clarkson2013low}. 

However, the main technical challenge here is that we can no longer use the approach of \cite{cohenmm17,mw17,bakshi2018sublinear} to use the low-rank approximation for $\RR$ and solve regression problems to recover an $\epsilon$-approximate low-rank matrix for $\AA$. In particular, we can now only hope for an $O(1)$ approximation if we use the standard technique of iteratively solving regression problems. Our first insight is that computing a Spectral Low-Rank Approximation to $\RR$ results in a \textit{structured} projection matrix for $\CC$, from which we can compute a \textit{structured} projection matrix for $\AA$. Further, this structured projection can be computed with only $\widetilde{O}(nk/\epsilon)$ queries. We first describe how this structured projection matrix for $\AA$ results in an efficient low-rank approximation algorithm. 

\subsection{Structured Projections to Low-Rank Approximation}

\label{subsec:structured_projection_to_lra}
Our starting point is a structural result based on the Spectral-Frobenius projection (SF) property introduced by Clarkson and Woodruff in the context of approximating arbitrary matrices with low-rank PSD matrices \cite{clarkson2017low}. In this subsection, we show that if we are given a projection matrix that satisfies the SF property, we can obtain a query-optimal algorithm for Low-Rank Approximation. We begin by defining this property:

\begin{definition}(\textit{$(\epsilon, k)$-\textsf{SF} Projection.})
\label{def:sf_projection}
Given any matrix $\AA \in \mathbb{R}^{n \times n}$, integer $k$, and $\epsilon>0$, a projection matrix $\PP \in  \mathbb{R}^{n \times n}$ is $(\epsilon,k)$-\textsf{SF} w.r.t. $\AA$ if
\begin{equation*}
    \|\AA - \AA \PP \|^2_2 \leq \frac{\epsilon}{k}\|\AA -\AA_k\|^2_F
\end{equation*}
or 
\begin{equation*}
    \|\AA -  \PP \AA \|^2_2 \leq \frac{\epsilon}{k}\|\AA -\AA_k\|^2_F
\end{equation*}
\end{definition}

Intuitively, the following structural result of Clarkson and Woodruff relates an $(\epsilon,k)$-\textsf{SF} projection to a relative-error low-rank approximation. We leverage this connection heavily in subsequent sections.

\begin{Frame}[\textbf{Algorithm \ref{alg:proj_to_rank}} : Structured Projection to Low-Rank Approximation]
\label{alg:proj_to_rank}
\textbf{Input}: A PSD Matrix $\AA \in \mathbb{R}^{n \times n}$, integer $k$, $\epsilon>0$, an orthonormal matrix $\QQ \in \mathbb{R}^{n \times k'}$ such that the projection matrix $\PP = \QQ \QQ^{\top}$ satisfies $\|\AA - \PP\AA \|^2_2 \leq \frac{\epsilon}{k}\|\AA -\AA_k\|^2_F $
\begin{enumerate}
    \item  Consider the optimization problem: \begin{equation*}\min_{\textrm{rank}(\XX)\leq k}\|\AA - \QQ \XX\QQ^{\top} \|^2_F  \vspace{-0.1in}\end{equation*}
    \item For all $i \in [n]$, compute the leverage scores, $\tau_i(\QQ)$. Since $\QQ$ has orthonormal columns, $\tau_i(\QQ)= \|\QQ_{i,*}\|^2_2$ and can be computed exactly. Let $p = \{ p_1, p_2 \ldots p_n \} $ denote a distribution over rows of $\AA$ for which $p_i = \tau_i(\QQ)/\sum_{i'}\tau_{i'}(\QQ)$. 
    \item Let $t = k'/\epsilon^2$. Construct a \textit{leverage score sampling} matrix  $\SS$ by sampling $t$ rows of $\AA$, such that $\SS = \DD \mathbf{\Omega}^{\top}$, where $\DD$ is a $t \times t$ diagonal matrix and $\mathbf{\Omega}$ is an $n \times t$ sampling matrix. For all $j \in [t]$, select row index $i \in [n]$ with probability $p_i$ and set $\mathbf{\Omega}_{i,j}=1$ and $\DD_{j,j} = 1/\sqrt{t p_i}$.  Repeat this sampling process to construct another \textit{leverage score sampling} matrix $\TT$.  
    \item Consider the sketched optimization problem : \begin{equation*}\min_{\textrm{rank}(\XX)\leq k}\|\SS\AA\TT - \SS\QQ \XX\QQ^{\top}\TT \|^2_F\vspace{-0.1in}\end{equation*}
    Compute $\SS\AA\TT$, $\PP_{\SS \QQ}$, $\PP_{ \QQ^{\top}\TT}$,$(\SS \QQ)^{\dagger}$ and $(\QQ^{\top}\TT)^{\dagger}$, where $\PP_{\SS \QQ}$ and $\PP_{ \QQ^{\top}\TT}$ are the projections onto $\SS \QQ$ and $\QQ^{\top} \TT$ respectively. Compute $\textsf{SVD}(\PP_{\SS \QQ} \SS\AA\TT \PP_{\QQ^{\top}\TT})$. By Theorem \ref{thm:gen_lra} the sketched problem is minimized by  $\XX^* = (\SS \QQ)^{\dagger} [\PP_{\SS \QQ} \SS\AA\TT \PP_{\QQ^{\top}\TT}]_k(\QQ^{\top}\TT)^{\dagger}$.
    \item Let $\UU^* \in \mathbb{R}^{k' \times k}$ be an orthonormal basis for the columns of $\XX^*$. Compute an orthonormal basis $\MM$ for $\QQ\UU^*$. Consider the following regression problem: 
    $\min_{\YY \in \mathbb{R}^{k \times n}}\|\AA - \MM \YY \|^2_F$.
    For all $i \in [n]$, compute $\tau_i(\MM) = \|\MM_{i,*}\|^2_2$. Let $q = \{q_1, q_2, \ldots, q_n \}$ be a distribution over the rows of $\AA$ such that $q_i = \tau_i(\MM)/\sum_{i'\in[n]} \tau_{i'}(\MM)$. Let $\WW$ be a \textit{leverage score sampling} matrix with $k/\epsilon$ rows sampled proportional to $q$. 
    \item Consider the sketched regression problem: $\min_{\YY \in \mathbb{R}^{k \times n}} \|\WW\AA - \WW \MM \YY \|^2_F$.
    Let $\NN$ be the minimizer to this regression problem computed using the algorithm from Lemma \ref{lem:fast_regression}. 
\end{enumerate}
\textbf{Output:} $\MM , \NN^{\top} \in \mathbb{R}^{n \times k}$ such that $\|\AA - \MM \NN\|^2_F \leq (1+\epsilon)\|\AA - \AA_k\|^2_F$   
\end{Frame}

\begin{lemma}(Structured Projections and Low-Rank Approximation \cite{clarkson2017low}.)
\label{lem:structured_projection}
Let $\PP \in \mathbb{R}^{n \times n}$ be an $(\epsilon,k)$-\textsf{SF} projection w.r.t $\AA$, then
\begin{equation*}
    \|\AA - \PP \AA_k \PP \|^2_F \leq (1+\epsilon) \|\AA - \AA_k\|^2_F
\end{equation*}
\end{lemma}

Ignoring computational and query complexity constraints, suppose we were given a matrix $\QQ \in \mathbb{R}^{n \times k'}$ with orthonormal columns such that $\PP = \QQ \QQ^{\top}$ is an $(\epsilon,k)$-\textsf{SF} Projection, where $k'$ is the dimension of the space $\PP$ projects onto. Note, for now it suffices to set $k' =\textrm{poly}(k/\epsilon)$. As a consequence of Lemma \ref{lem:structured_projection}, we observe that solving the following constrained regression problem suffices to obtain a $(1+\epsilon)$-relative error solution to the Low-Rank Approximation problem: 
\begin{equation}
\label{eqn:approx_regression}
    \min_{\textrm{rank}(\XX) \leq k} \|\AA - \QQ \XX \QQ^{\top} \|^2_F
\end{equation}
\noindent 
However, there are several challenges pertaining to this approach. As noted above, it is not immediately clear how to obtain such a $\QQ$ with $nk/\epsilon$ queries to $\AA$. Further, it is not immediately clear how to solve Equation \ref{eqn:approx_regression} efficiently. While we have reduced to optimizing over $k' \times k'$ sized matrices $\XX$ with rank at most $k$, the problem still seems intractable in sublinear time and queries.

We begin by describing how to solve the optimization problem in Equation \ref{eqn:approx_regression} with $\widetilde{O}(nk/\epsilon)$ queries given that we have access to $\QQ $ and  $\QQ\QQ^{\top}$ is an $(\epsilon,k)$-\textsf{SF} Projection. At a high level, our approach is to sketch the problem by sampling rows and columns proportional to the row leverage scores of $\QQ$. We observe that since $\QQ$ has orthonormal columns, the row leverage scores of $\QQ$ are simply the $\ell^2_2$ norms of corresponding rows. Therefore, we create sampling matrices $\SS$ and $\TT$ that sample $\textrm{poly}(k')$ rows proportional to the leverage scores of $\QQ$ and consider the resulting optimization problem:
\begin{equation}
\label{eqn:sketched_approx_regression}
    \min_{\textrm{rank}(X)\leq k}\|\SS\AA\TT - \SS\QQ\XX\QQ^{\top}\TT\|^2_F
\end{equation}

We then show that the minimizer for Equation \ref{eqn:sketched_approx_regression} is an approximate minimizer for Equation \ref{eqn:approx_regression}. Further, the optimization problem in Equation \ref{eqn:sketched_approx_regression} is referred to as Generalized Low-Rank Approximation and admits a closed form solution:

\begin{theorem}(Generalized Low-Rank Approximation \cite{friedland2007generalized}.)
\label{thm:gen_lra}
Let $\AA \in \mathbb{R}^{n \times n}$, $\BB \in \mathbb{R}^{n \times k'}$ and $\CC \in \mathbb{R}^{k' \times n}$ and $k\in\mathbb{N}$. Then, the Generalized Low-Rank Approximation problem 
\begin{equation*}
    \min_{\textrm{rank}(\XX)\leq k} \|\AA - \BB \XX \CC \|^2_F
\end{equation*}
is minimized by $\XX = \BB^{\dagger}[\PP_B \AA \PP_C]_k \CC^{\dagger}$, where $\PP_{B},\PP_C$ are the projection matrices onto $\BB$ and $\CC$ respectively.   
\end{theorem}

We apply the above theorem to Equation \ref{eqn:sketched_approx_regression}. 
Both the query complexity and running time here contribute a lower-order term and we can afford to compute the $\textsf{SVD}$ for each term. Let $\XX^*$ be the solution to the sketched optimization problem in Equation \ref{eqn:sketched_approx_regression}. Then, we can compute $\UU^*$, an orthonormal column basis for $\XX^*$ and consider $\MM $, an orthonormal basis for $\QQ \UU^* \in \mathbb{R}^{n \times k}$ to be one of the low-rank factors for $\AA$. To find the second factor, we set up the following regression problem: 
\begin{equation}
\label{eqn:final_regression}
    \min_{\NN \in \mathbb{R}^{k \times n}}\|\AA - \MM \NN \|^2_F
\end{equation}

Again, $\MM$ has orthonormal columns and thus we can efficiently compute the corresponding row leverage scores and sample $k/\epsilon$ rows. By Lemma \ref{lem:affine_embedding} this achieves a $(1+\epsilon)$-approximation to the optimal cost in Equation \ref{eqn:final_regression} and obtains an $\NN^*$ with $\widetilde{O}(nk/\epsilon)$ queries to $\AA$. At this stage, we have obtained a $(1+\epsilon)$-approximate rank-$k$ solution to Equation \ref{eqn:approx_regression} and Lemma \ref{lem:structured_projection} implies that we are done. We now formalize this argument:

\begin{theorem}(Structured Projection to Low-Rank Approximation.)
\label{thm:struc_proj_to_lra}
Given a rank-$k'$ projection matrix $\PP = \QQ \QQ^{\top}$, such that $\PP$ is an $(\epsilon, k)$-\textsf{SF} projection, Algorithm \ref{alg:proj_to_rank} queries $\widetilde{O}(nk/\epsilon + k'^2/\epsilon^4)$
entries in $\AA$ and with probability $99/100$ outputs $\MM, \NN^{\top} \in \mathbb{R}^{n \times k}$ such that
\begin{equation*}
    \|\AA - \MM \NN \|^2_F \leq (1+\epsilon)\|\AA - \AA_k\|^2_F
\end{equation*}
Further, Algorithm \ref{alg:proj_to_rank} runs in time $\widetilde{O}( n(k/\epsilon)^{\omega-1} + nk'^{\omega-1} + (k'/\epsilon^2)^\omega)$.
\end{theorem}

\begin{proof}
Since $\PP$ is an $(\epsilon, k)$-\textsf{SF} projection and $\AA_k$ is a feasible solution to $\min_{\textrm{rank}(\YY)\leq k} \|\AA - \PP \YY \PP\|^2_F$, from Lemma \ref{lem:structured_projection} we have
\begin{equation}
\label{eqn:approx_opt}
    \min_{\textrm{rank}(\YY)\leq k}\|\AA - \PP \YY \PP\|^2_F \leq (1+\epsilon)\|\AA - \AA_k\|^2_F
\end{equation}
Since $\PP = \QQ \QQ^{\top}$, we can substitute it in Equation \ref{eqn:approx_opt} to get $\min_{\textrm{rank}(\YY)\leq k} \|\AA - \QQ \QQ^{\top} \YY \QQ \QQ^{\top}\|^2_F$. We further relax this by optimizing over all rank-$k$ matrices $\XX \in \mathbb{R}^{k'\times k'}$ instead of matrices of the form $\QQ \YY \QQ^{\top}$. Therefore,
\begin{equation}
\label{eqn:approx_opt_smaller}
    \min_{\textrm{rank}(\XX)\leq k}\|\AA - \QQ \XX \QQ^{\top}\|^2_F \leq (1+\epsilon)\|\AA - \AA_k\|^2_F
\end{equation}
where we are now optimizing over a $k' \times k'$ matrix $\XX$, which is considerably smaller than $\YY$. Let $\SS,\TT^{\top} \in \mathbb{R}^{k'/\epsilon^2 \times n}$ be the \textit{leverage score sampling} matrices as defined in Algorithm \ref{alg:proj_to_rank}. Observe, from Lemma \ref{lem:affine_embedding} we know $\SS$ has a sufficient number of rows to be an affine embedding for Equation \ref{eqn:approx_opt_smaller}. However, we cannot directly apply the affine embedding guarantee since $\AA - \QQ\XX\QQ^{\top}$ is not an affine space. Let $\HH$ be $k' \times n$ matrix, let $\HH^* =\arg\min_{\HH}\|\AA - \QQ \HH \|^2_F$ and let $\AA^*$ be $\AA -\QQ \HH^*$. Then, with probability at least $1- c_1$, for all $\HH$, 
\begin{equation}
    \label{eqn:real_affine_sketch}
    \|\SS \AA - \SS\QQ\HH \|^2_F - \|\SS \AA^*\|^2_F = (1\pm\epsilon) \|\AA - \QQ\HH \|^2_F +\|\AA^*\|^2_F
\end{equation}
Since Equation  \ref{eqn:real_affine_sketch} holds for all $\HH$, in particular it holds for all rank-$k$ matrices $\XX$ such that $\HH = \XX \QQ^{\top}$. Therefore, with probability at least $1-c_1$, for all rank $k$ matrices $\XX$,
\begin{equation}
\label{eqn:approx_opt_smaller_sketch1}
    \|\SS\AA - \SS\QQ\XX\QQ^{\top}\|^2_F -\|\SS\AA^* \|^2_F = (1 \pm \epsilon)\|\AA - \QQ\XX\QQ^{\top}\|^2_F + \| \AA^*\|^2_F
\end{equation}
Here, we observe that while we cannot estimate $\|\AA^*\|^2_F$ accurately, it is a fixed matrix independent of $\XX$ and thus we can still approximately optimize. Let $\zeta_1$ be the event that Equation \ref{eqn:approx_opt_smaller_sketch1} holds. We now use the sampling matrix $\TT$ to sketch $\|\SS \AA - \ZZ \QQ^{\top} \|^2_F$. Let $\ZZ' = \arg\min_{\ZZ}\|\SS \AA - \ZZ \QQ^{\top} \|^2_F $ and let $\SS\AA' = \SS \AA - \ZZ' \QQ^{\top}$. Then, with probability at least $1-c_2$, for all $\ZZ$,
\begin{equation}
\label{eqn:approx_opt_smaller_sketch2}
    \|\SS \AA\TT - \ZZ \QQ^{\top}\TT\|^2_F - \|\SS\AA' \TT\|^2_F = (1\pm \epsilon)\|\SS \AA - \ZZ \QQ^{\top}\|^2_F + \|\SS\AA'\|^2_F
\end{equation}
In particular, the above equation holds for all rank-$k$ matrices $\XX$ such that $\ZZ = \SS \QQ^{\top}\XX$. Let $\zeta_2$ be the event that the aforementioned equation holds. Combining equations \ref{eqn:approx_opt_smaller_sketch1} and \ref{eqn:approx_opt_smaller_sketch2} and conditioning on $\zeta_1$ and $\zeta_2$, for all rank-$k$ matrices $\XX$, 
\begin{equation}
    \|\SS \AA\TT - \SS\QQ\XX\QQ^{\top}\TT\|^2_F - \|\SS\AA' \TT\|^2_F = (1\pm\epsilon)^2 \left( \|\AA - \QQ \XX\QQ^{\top} \|^2_F + \|\SS\AA^* \|^2_F + \|\AA^*\|^2_F \right) + \|\SS\AA' \|^2_F
\end{equation}

Here, we observe that while the sketch does not preserve the cost of all $\XX$ up to relative error $(1+\epsilon)$, the additive error $\Delta \leq (1+\epsilon)\left(\|\SS\AA^* \|^2_F + \|\AA^*\|^2_F + \|\SS\AA'\|^2_F + \|\SS\AA'\TT \|^2_F \right) $ is fixed and is independent of $\XX$. 
Let $\XX^* = \arg\min_{\textrm{rank}(\XX)\leq k} \|\SS\AA\TT - \SS\QQ\XX\QQ^{\top}\TT \|^2_F$. Then, union bounding over $\zeta_1$ and $\zeta_2$, with probability $1-c_1 - c_2$, 
\begin{equation}
\label{eqn:sat_to_a_k}
    \|\AA - \QQ\XX^*\QQ^{\top}\|^2_F \leq (1 +\epsilon) \min_{\textrm{rank}(\XX)\leq k} \|\AA - \SS\QQ\XX\QQ^{\top}\|^2_F
\end{equation}
Therefore, it suffices to efficiently compute $\XX^*$. By Theorem \ref{thm:gen_lra}, we know that the sketched optimization problem above is minimized by $\XX^* = (\SS \QQ)^{\dagger} [\PP_{\SS \QQ} \SS\AA\TT \PP_{\QQ^{\top}\TT}]_k(\QQ^{\top}\TT)^{\dagger}$, which can be computed exactly as shown in Step 4 of Algorithm \ref{alg:proj_to_rank}.
We note that we can now explicitly compute $\SS \AA \TT$ by querying the relevant entries in $\AA$. Further, we can compute $\SS \QQ$ and $\QQ^{\top}\TT$ without querying $\AA$ at all. 
Recalling equation \ref{eqn:sat_to_a_k} we can approximate the optimal low rank approximation cost: $$\|\AA - \QQ\XX^*\QQ^{\top}\|^2_F \leq (1+O(\epsilon))\|\AA - \AA_k\|^2_F$$ 
While we have now approximately minimized the optimization problem from Equation \ref{eqn:approx_regression}, recall our goal was to obtain a rank-$k$ approximation to $\AA$ in factored form i.e., outputting $n \times k$ matrices $\MM, \NN^{\top}$ such that the low rank approximation is given by $\MM\NN$. 
Towards this end, we compute $\UU^*$, an orthonormal column basis for $\XX^*$  such that $\XX^* = \UU^* \VV^*$. Substituting this in the above equation we have 
\begin{equation}
\label{eqn:factorized_struc_proj}
    \|\AA -\QQ \UU^* \VV^* \QQ^{\top} \|^2_F\leq (1+\epsilon)\|\AA - \AA_k \|^2_F
\end{equation}
Let $\MM = \QQ \UU^* \in \mathbb{R}^{n \times k}$ 
be one of the low-rank factors for $\AA$. To find the second factor, we observe : 
\begin{equation}
\label{eqn:pre_final_regression}
    \min_{\YY \in \mathbb{R}^{k \times n}}\|\AA - \MM \YY \|^2_F \leq \|\AA - \MM \VV^* \QQ^{\top} \|^2_F
\end{equation}
\noindent and therefore, approximately optimizing $\min_{\YY \in \mathbb{R}^{k \times n}}\|\AA - \MM \YY \|^2_F$ suffices. Again, $\MM$ has orthonormal columns and thus we can efficiently compute the corresponding leverage scores to create a sketch $\WW$ with $O(k/\epsilon)$ rows. From Lemma \ref{lem:fast_regression},  with probability at least $1-c_3$ for all $\YY$, 
\begin{equation*}
    \|\WW\AA - \WW\MM \YY \|^2_F = (1\pm\epsilon)\| \AA - \MM \YY \|^2_F
\end{equation*}
Let $\NN$ be the optimal solution for the sketched problem as defined in Algorithm \ref{alg:proj_to_rank}. Then, with probability at least $1- c_3$,
\begin{equation}
    \|\AA - \MM \NN \|^2_F \leq \left(\frac{1+\epsilon}{1-\epsilon}\right)\min_{\YY\in \mathbb{R}^{k \times n} }\| \AA - \MM \YY \|^2_F
\end{equation}
We conclude correctness by union bounding over the failure probabilities of all the sketches and observing that with probability at least $99/100$,
\begin{equation*}
    \|\AA - \MM \NN \|^2_F \leq (1+O(\epsilon))\|\AA - \MM\VV^* \QQ^{\top} \|^2_F \leq (1+O(\epsilon))\|\AA - \AA_k\|^2_F
\end{equation*}
where the inequalities follow from Equations \ref{eqn:final_regression}, \ref{eqn:pre_final_regression} and \ref{eqn:factorized_struc_proj}.

Finally, we analyze the query complexity and running time of our algorithm.  Since Algorithm \ref{alg:proj_to_rank} is given $\QQ$ as input, computing the leverage scores in Step 2 requires no queries to $\AA$ and requires $O(nk')$ time. Next, observe we do not have to explicitly compute $\SS \AA$ or $\AA \TT$, since $\SS \AA \TT$ is simply a submatrix of $\AA$ with $(k'/\epsilon^2)^2$ entries appropriately scaled, it suffices to query them. $\SS \AA \TT$ can be computed in $O(k'^2/\epsilon^4)$ time. Next, we compute $\textsf{SVD}(\SS\QQ)$ and $\textsf{SVD}(\QQ^{\top}\TT)$, which requires no queries to $\AA$ and time $O(k'^{\omega}/\epsilon^2)$. We can then compute $(\SS\QQ)^{\dagger}, (\QQ^{\top}\TT)^{\dagger},\PP_{\SS\QQ}$ and $\PP_{\QQ^{\top}\TT}$ from the aforementioned \textsf{SVD}s. Next, we compute the matrix $\PP_{\SS\QQ}\SS\AA\TT\PP_{\QQ^{\top}\TT}$, which requires no extra queries to $\AA$ and time $O( (k'/\epsilon^2)^\omega)$, which is also the time required to compute its $\textsf{SVD}$. We can then compute $\XX^*$ in Step 4 with a total of  $O(k'^2/\epsilon^4)$ queries to $\AA$ in time $O(nk' + (k'/\epsilon^2)^{\omega})$. 

In Step 5, we can compute $\UU^*$ by computing the $\textsf{SVD}$ of $\XX^*$ and compute $\MM$ in time $O(nk'^{\omega-1} + k'^\omega)$ and do not require any queries to $\AA$. In Step 6, computing $\WW\AA$ requires $O(nk/\epsilon)$ queries to $\AA$, since $\WW$ has $\widetilde{O}(k/\epsilon)$ rows. Note, this step contributes the leading term to the query complexity and it is crucial $\WW$ does not have more rows. By Lemma \ref{lem:fast_regression}, $\NN$ can be computed in time $\widetilde{O}(nk/\epsilon + n(k/\epsilon)^{\omega-1} + k^3/\epsilon)$.  
Overall, Algorithm \ref{alg:proj_to_rank} requires $\widetilde{O}(nk/\epsilon + k'^2/\epsilon^4)$ queries to $\AA$ and runs in time $\widetilde{O}( n(k/\epsilon)^{\omega-1} + nk'^{\omega-1} + (k'/\epsilon^2)^\omega)$.
\end{proof}

In light of Theorem \ref{thm:struc_proj_to_lra}, to obtain a low rank approximation for $\AA$, it suffices to obtain an \textsf{SF} Projection. In particular, it suffices to obtain a matrix $\QQ \in \mathbb{R}^{n \times k'}$, for $k' = \textrm{poly}(k,1/\epsilon)$ such that $\PP = \QQ\QQ^{\top}$ is an $(\epsilon,k)$-\textsf{SF} projection, by querying $\widetilde{O}(nk/\epsilon)$ entries in $\AA$. 
One possible approach to computing such a $\QQ$ is to use the following result by Musco and Woodruff:


\begin{theorem}(Theorem 25, \cite{mw17}.)
Given a PSD matrix $\AA$, integer $k$, $\epsilon>0$, there exists an algorithm that reads $\widetilde{O}(nk/\epsilon^6 + nk^2/\epsilon^2)$ entries of $\AA$ and with probability at least $99/100$, outputs $\MM,\NN^{\top} \in \mathbb{R}^{n \times k}$ such that 
\begin{equation*}
    \|\AA - \MM \NN \|^2_2 \leq (1+\epsilon)\| \AA - \AA_k\|^2_2 + \frac{\epsilon}{k}\|\AA - \AA_k \|^2_F
\end{equation*}
\end{theorem}

\noindent Instantiating this theorem with $\epsilon = O(1)$ and $k = k/\epsilon$, we obtain a matrix $\MM,\NN^{\top} \in \mathbb{R}^{n \times k/\epsilon}$ such that
\begin{equation*}
\begin{split}
    \|\AA -  \MM \NN^{\top}\|^2_2 & \leq O(1)\|\AA - \AA_{k/\epsilon} \|^2_2 + O\left(\frac{\epsilon}{k}\right) \|\AA - \AA_{k/\epsilon}\|^2_F\\
    & \leq O\left(\frac{\epsilon}{k}\right) \|\AA - \AA_k\|^2_F
\end{split}
\end{equation*}
\noindent where the last inequality follows from observing $\|\AA -\AA_{k/\epsilon}\|^2_F\leq \|\AA -\AA_k\|^2_F$ and
\begin{equation*}
    \|\AA - \AA_k\|^2_F = \sum^n_{j = k+1} \sigma^2_j(\AA) \geq \left(\frac{k}{\epsilon}- k\right)\sigma^2_{k/\epsilon}\geq \left(\frac{k}{\epsilon}- k\right)\|\AA -\AA_{k/\epsilon}\|^2_2
\end{equation*}
We can then compute an orthonormal basis for $\MM$ and denote it by $\QQ$. Here, we observe $\PP = \QQ\QQ^{\top}$ is an $(\epsilon,k)$-\textsf{SF} projection matrix. 
Further, the algorithm of Musco and Woodruff instantiated with the above parameters queries $\widetilde{O}(nk^2/\epsilon^2)$ entries in $\AA$.  
As a corollary of Theorem \ref{thm:struc_proj_to_lra}, providing the rank-$k/\epsilon$ projection matrix $\QQ$ as input to Algorithm \ref{alg:proj_to_rank}, implies an algorithm for low rank approximation which queries $\widetilde{O}(nk^2/\epsilon^2)$ entries in $\AA$. This already improves the $\epsilon$-dependence in the query complexity of best known algorithm for PSD low-rank approximation, since the algorithm of Musco and Woodruff requires $\widetilde{O}(nk/\epsilon^{2.5})$ queries \cite{mw17}. Note, this algorithm has worse dependence on $k$. However, our goal is to obtain linear dependence on both $k$ and $1/\epsilon$. Towards this end, we focus on obtaining an \textsf{SF} projection with fewer queries to $\AA$.  


\subsection{Spectral Regression}
\label{sec:spectral_regression}
In this subsection, we consider the \textit{Spectral Regression} problem. This problem is a natural generalization of least-squares regression, when the response variable is a matrix. \textit{Spectral Regression} arises in the context of Regularized Least Squares Classification, for instance \cite{chen2010multiclass}.  Given matrices $\AA \in \mathbb{R}^{n \times d}$, $\XX \in \mathbb{R}^{d \times m}$ and $\BB \in \mathbb{R}^{n \times m}$, the \textit{Spectral Regression} problem considers the following optimization problem: 
\begin{equation*}
    \min_{\XX} \|\AA\XX - \BB \|_2
\end{equation*}
We note that this is natural variant of multi-response regression where we minimize the difference between $\AA\XX$ and $\BB$ in spectral norm as opposed to the extensively studied and well-understood Frobenius norm. To the best of our knowledge the only relevant related work on \textit{Spectral Regression} is by Clarkson and Woodruff \cite{clarkson2009numerical} and Cohen et. al. \cite{cohen2015optimal}. Both these works provide oblivious sketches to reduce the dimension of the problem, which unfortunately do not suffice for our application. Instead of \textit{Spectral Regression} in its full generality we focus on the following special case:  

Given an $n \times n$ PSD matrix $\AA$, a rank parameter $k$, and an accuracy parameter $\epsilon$, let $\CC$ be a $n \times \sqrt{nk/\epsilon}$ matrix such that it is a \textit{column PCP} for $\AA$, satisfying the guarantees of Lemma \ref{lem:rel_col_pcp}, instantiated with $k = k/\epsilon$, and $\epsilon =O(1)$. Let $\ZZ^{\top}$ be a 
$k/\epsilon \times \sqrt{nk/\epsilon}$ matrix with orthonormal rows such that the corresponding projection matrix  $\ZZ \ZZ^{\top} $ is an $(O(1),k/\epsilon)$-\textsf{SF} Projection for $\CC$. Then, we consider the following \textit{Spectral Regression} problem:
\begin{equation}
\label{eqn:spec_regression}
    \min_{\WW \in \mathbb{R}^{n \times k/\epsilon}}\|\CC -\WW\ZZ^{\top}\|_2
\end{equation}


Our main technical contribution here is to obtain a new algorithm to solve this optimization problem. We subsequently show how understanding this special case is crucial to obtaining \textit{optimal} algorithms for low rank approximation of PSD matrices. The techniques we develop here may be of independent interest and find applications to other problems. Formally, we prove the following:

\begin{theorem}(Approximate Spectral Regression.)
\label{thm:spectral_regression}
Let $\CC \in \mathbb{R}^{n \times \sqrt{nk/\epsilon}}$ be a column PCP for $\AA$ satisfying the guarantees of Lemma \ref{lem:rel_col_pcp} instantiated with $k = k/\epsilon$ and $\epsilon=O(1)$. Let $\ZZ \in \mathbb{R}^{\sqrt{nk/\epsilon}\times k/\epsilon}$ be an orthonormal matrix such that $\ZZ \ZZ^{\top}$ is an $(O(1),k/\epsilon)$-\textsf{SF} projection for $\CC$. Then, Algorithm \ref{alg:approx_spectral_regression} queries $\widetilde{O}(nk/\epsilon)$ entries in $\AA$ and with probability $99/100$ computes $\widehat{\WW}$ such that 
\begin{equation*}
    \|\CC - \widehat{\WW} \ZZ^{\top} \|^2_2 \leq \widetilde{O}(1)\left(\min_{\WW \in \mathbb{R}^{n \times k/\epsilon}} \|\CC - \WW\ZZ^{\top}\|^2_2 + \frac{\epsilon}{k}\|\CC- \CC_{k/\epsilon} \|^2_F\right)
\end{equation*}
Further, the algorithm runs in time $\widetilde{O}(nk/\epsilon + (k/\epsilon)^{\omega})$.  
\end{theorem}

\begin{Frame}[\textbf{Algorithm \ref{alg:approx_spectral_regression}} : Approximate Spectral Regression]
\label{alg:approx_spectral_regression}
\textbf{Input}: A PSD Matrix $\AA \in \mathbb{R}^{n \times n}$, integer $k$, and $\epsilon>0$. $\CC \in \mathbb{R}^{n \times \sqrt{nk/\epsilon}}$, a \textit{column PCP} for $\AA$ satisfying the guarantees of Lemma \ref{lem:rel_col_pcp} instantiated with $k = k/\epsilon$ and $\epsilon=O(1)$. $\ZZ \in \mathbb{R}^{\sqrt{nk/\epsilon}\times k/\epsilon}$ be an orthonormal matrix such that $\ZZ \ZZ^{\top}$ is an $(O(1),k/\epsilon)$-\textsf{SF} projection for $\CC$. 
\begin{enumerate}
    \item Consider the \textit{Spectral Regression} problem:
    $$\min_{\WW}\|\CC - \WW \ZZ^{\top}\|^2_2 $$
     Let $t = \sqrt{nk/\epsilon}$. For all $j \in [t]$, compute $\tau_j(\ZZ^{\top}) = \|\ZZ_{j,*}\|^2_2$. Let $q =\{q_1, q_2, \ldots, q_t\}$ be a distribution of columns of $\CC$ such that for all $j \in [t]$, $q_j = \min( \tau_{j}(\ZZ^{\top}), 1)$.
    \item Construct a sampling matrix $\SS$ such that $\CC \SS$ selects each column $\CC_{*,j}$ independently with probability $q_j$ and scales it by $1/\sqrt{q_j}$. Similarly, construct $\ZZ^{\top} \SS$. Consider the sketched optimization problem :
    $$\min_{\WW}\| \CC \SS - \WW\ZZ^{\top}\SS \|^2_2$$
    \item Compute $(\ZZ^{\top}\SS)^{\dagger} = \SS^{\top}\ZZ(\ZZ^{\top}\SS \SS^{\top}\ZZ)^{-1}$. Let $\widehat{\WW} = \CC\SS (\ZZ^{\top}\SS)^{\dagger}$ be the solution to the sketched optimization problem. 
\end{enumerate}
\textbf{Output:} $\widehat{\WW} \in \mathbb{R}^{n \times k/\epsilon}$ such that $\|\CC - \widehat{\WW}\ZZ^{\top}\|^2_2 \leq \widetilde{O}(1)\min_{\WW}\|\CC - \WW\ZZ^{\top}\|^2_F +\widetilde{O}(\epsilon/k)\|\CC - \CC_{k/\epsilon}\|^2_F$   
\end{Frame}

We begin by characterizing the optimal solution and optimal cost for the \textit{Spectral Regression} problem. We prove a structural result that shows the optimal solution for \textit{Spectral Regression} is given by projecting $\CC$ away from the span of $\ZZ^{\top}$. This matches the characterization of the optimal solution to regression under the Frobenius norm, given by the well-known normal equations. Recall, by definition of the Moore-Penrose pseudoinverse, this projection matrix is $(\ZZ^{\top})^{\dagger}\ZZ^{\top}$. 

Then, the optimal cost for Equation \ref{eqn:spec_regression} is $\|\CC - \CC (\ZZ^{\top})^{\dagger}\ZZ^{\top} \|^2_2$ and is achieved by $\WW^* = \CC(\ZZ^{\top})^{\dagger}$. Intuitively, we show that any feasible $\WW$ must incur the above cost by analyzing $\|y^{\top}(\CC - \CC (\ZZ^{\top})^{\dagger}\ZZ^{\top})\|^2_2$ for a fixed vector $y$. This enables us to exploit the geometry of Euclidean space and instantiate $y$ as needed to relate it back to the spectral norm.

\begin{lemma}(Characterizing Opt for Spectral Regression.)
\label{lem:characterizing_opt}
Let $\CC$ and $\ZZ$ be matrices as defined in Theorem \ref{thm:spectral_regression}. 
Let $\WW^* = \CC(\ZZ^{\top})^{\dagger} = \CC \ZZ (\ZZ^{\top} \ZZ)^{-1}$, such that $\WW^*\ZZ^{\top}$ is the projection of $\CC$ on the colspan$(\ZZ)$. Let $\CC^* = \CC - \WW^*\ZZ^{\top}$ be the projection of $\CC$ orthogonal to colspan$(\ZZ)$. Then, 
\begin{equation*}
    \|\CC^*\|^2_2 = \min_{\WW} \|\CC - \WW \ZZ^{\top} \|^2_2
\end{equation*}
and the corresponding minimizer is $\WW^*$.
\end{lemma}
\begin{proof}
Note, by definition $\|\CC - \WW^* \ZZ^{\top} \|^2_2 = \|\CC^*\|^2_2$ and since $\WW^*$ is feasible, $$\min_{\WW}\|\CC - \WW\ZZ^{\top} \|^2_2\leq \| \CC^* \|^2_2.$$ 
Therefore, it suffices to show any $\WW$ must incur cost at least $\|\CC^*\|^2_2$.  By definition, we have
\begin{equation*}
    \begin{split}
        \CC = \CC (\ZZ^{\top})^{\dagger} \ZZ^{\top} + \CC (\II -(\ZZ^{\top})^{\dagger} \ZZ^{\top}) = \CC(\ZZ^{\top})^{\dagger} \ZZ^{\top} + \CC^* \\
    \end{split}
\end{equation*}
By definition of spectral norm, $\|\CC - \WW \ZZ^{\top}\|^2_2 \geq   \|y^{\top}\CC - y^{\top} \WW \ZZ^{\top}\|^2_2$, for all $y$ such that $\| y\|_2=1$. Next, for any unit vector $y \in \mathbb{R}^{n}$, \begin{equation}
\label{eqn:pythagorean}
\begin{split}
    \|y^{\top}\CC - y^{\top}\WW\ZZ^{\top} \|^2_2 & = \| y^{\top}(\CC(\ZZ^{\top})^{\dagger} \ZZ^{\top} + \CC^*) - y^{\top}\WW\ZZ^{\top} \|^2_2 \\
    & = \|y^{\top} \CC^* - y^{\top}(\WW - \WW^*)\ZZ^{\top}  \|^2_2 \\
    & = \|y^{\top} \CC^* \|^2_2 + \|y^{\top}(\WW -\WW^*)\ZZ^{\top} \|^2_2 + 2\langle y^{\top} \CC^* , y^{\top} (\WW - \WW^*)\ZZ^{\top} \rangle 
\end{split}
\end{equation}
We observe that $\WW \ZZ^{\top} = \CC (\ZZ^{\top})^{\dagger} \ZZ^{\top}$ is the projection of $\CC$ on the rowspan of $\ZZ^{\top}$ and $\CC^*$ is the projection of $\CC$ on the orthogonal complement of rowspan of $\ZZ^{\top}$. Therefore, $\langle \CC (\ZZ^{\top})^{\dagger} \ZZ^{\top} , \CC^* \rangle = 0$. Further, for any $y$, $y^{\top} (\WW - \WW^*)\ZZ^{\top}$ is in the row span of $\ZZ^{\top}$ and is thus perpendicular to $y^{\top} \CC^*$. Plugging this back in to Equation \ref{eqn:pythagorean}, we have 
\begin{equation}
\label{eqn:pytha_2}
    \|y^{\top}\CC - y^{\top}\WW\ZZ^{\top} \|^2_2  = \|y^{\top} \CC^* \|^2_2 + \|y^{\top}(\WW -\WW^*)\ZZ^{\top} \|^2_2 \geq \|y^{\top} \CC^*\|^2_2 
\end{equation}
where the inequality follows from non-negativity of norms. 
Since Equation \ref{eqn:pytha_2} holds for all $y$, we can pick $y$ such that $\| y^{\top} \CC^*\|^2_2 = \|\CC^*\|^2_2$. Therefore, $\|\CC - \WW\ZZ^{\top}\|^2_2 \geq \|y^{\top}\CC - y^{\top} \WW \ZZ^{\top}\|^2_2 \geq \|\CC^*\|^2_2$. This completes the proof. 
\end{proof}

Next, we  sketch the \textit{Spectral Regression} problem from Equation \ref{eqn:spec_regression} such that we approximately preserve the spectral norm cost of all $\WW \in \mathbb{R}^{n \times k/\epsilon}$. A natural approach here would be to follow the Affine Embedding idea for Frobenius norm and hope a similar guarantee holds for spectral norm as well. 
However, since $\ZZ$ could have rank as large as $k/\epsilon$ and we can no longer obtain a relative-error $(1+\epsilon)$-approximate Affine Embedding even for Frobenius norm without incurring a larger dependence on $\epsilon$. 
Instead, we relax the notion of approximation for our sketch. We note that it suffices to construct a sketch $\SS$ such that if 
\begin{equation*}
    \widehat{\WW} = \arg\min_{\WW}\|\CC\SS- \WW\ZZ^{\top} \SS\|^2_2
\end{equation*}
then 
\begin{equation*}
    \|\CC -\widehat{\WW}\ZZ^{\top} \|^2_2 \leq \widetilde{O}(1)\left(\min_{\WW \in \mathbb{R}^{n \times k/\epsilon}} \|\CC - \WW\ZZ^{\top}\|^2_2 + \frac{\epsilon}{k}\|\CC- \CC_{k/\epsilon} \|^2_F\right)
\end{equation*}
as stated in Theorem \ref{thm:spectral_regression}. Note, here we only need to weakly preserve the cost of the optimal $\WW$ for the sketched problem as opposed to preserving the cost of all matrices $\WW$. At a high level, this comes down to analyzing the spectrum of $\| \CC^* \SS \SS^{\top}  \|_2$
We begin with the definition of the Approximate Matrix Multiplication (AMM) guarantee and discuss its application in approximately minimizing \textit{Spectral Regression}. \textcolor{red}{} 

\begin{definition}\textit{($(\epsilon,  k)$-Spectral AMM.)}
Given matrices $\AA \in \mathbb{R}^{n \times m}$ and $\BB \in \mathbb{R}^{m \times d}$, a sketch $\PI \in \mathbb{R}^{m \times t}$ satisfies $(\epsilon, k)$-Spectral AMM if with probability at least $1-\delta$,
\begin{equation*}
    \|\AA \PI\PI^{\top} \BB - \AA \BB \|_2 \leq \epsilon \sqrt{\left(\| \AA \|^2_2 + \frac{\| \AA\|^2_F}{k}\right) \cdot \left(\|\BB \|^2_2 + \frac{\|\BB \|^2_F}{k}  \right)}
\end{equation*}
\end{definition}

\textit{Approximate Matrix Multiplication} was introduced by  Drineas et al. \cite{drineas2006fast} with respect to the Frobenius norm, as opposed to the spectral norm above. Subsequent work by Cohen et al. \cite{cohen2015optimal} studied the spectral norm bound and showed that any sketch $\PI$ that is an oblivious subspace embedding (i.e., satisfies Lemma \ref{lem:subspace_embedding} with $\PI$ being an oblivious sketch) implies an AMM guarantee, as long as $\PI$ has $\Theta(k +\log(1/\delta)/\epsilon^2)$ columns. The \textit{Spectral AMM} property combined with an $O(1)$-Subspace Embedding suffice to approximately minimize the Spectral Regression problem :

\begin{theorem}(Theorem 3, \cite{cohen2015optimal})
\label{thm:gen_regression}
Let $\AA$, $\BB$ and $\PI$ be as defined above. If $\PI$ is an $(\sqrt{\epsilon}, \textrm{rank}(\AA))$-\textit{Spectral AMM} for $\UU_A$ and $(\II -\PP_{\AA})\BB$, and an $O(1)$-Subspace Embedding for $\AA$,  and  $\widehat{\XX} =  \arg\min_{\XX}\|\PI\AA\XX- \PI\BB \|^2_2$, then with probability $99/100$,
\begin{equation*}
    \|\AA \widehat{\XX} - \BB \|^2_2 \leq (1+\epsilon)\|\PP_{\AA}\BB - \BB \|^2_2 + \frac{\epsilon}{k}\|\PP_{\AA}\BB - \BB\|^2_F
\end{equation*}
where $\UU_{\AA}$ is an orthonormal basis for $\AA$ and $\PP_{\AA}$ is the projection onto the span of $\AA$.
\end{theorem}

However, all the constructions presented for the sketch in \cite{cohen2015dimensionality} are either oblivious sketches or require sampling proportional to both $\AA$ and $\BB$. Applying an oblivious sketch $\SS$ in our problem requires computing $\CC \SS$ which would query $\Omega(\textsf{nnz}(\CC)) = \Omega(n^{1.5}\sqrt{k/\epsilon})$ entries in $\AA$. Therefore, the main challenge here is to construct a sampling matrix $\SS$ while reading $\widetilde{O}(nk/\epsilon)$ entries in $\AA$  such that $\SS$ is an $(\widetilde{O}(1), k/\epsilon)$-\textit{Spectral AMM} and an $O(1)$-Subspace Embedding. We construct $\SS$ by sampling $\widetilde{O}(k/\epsilon)$ columns of $\ZZ^{\top}$ proportional to the leverage scores of $\ZZ^{\top}$. While it is easy to show $\SS$ is a Subspace Embedding, observe that our sampling probabilities are computed without reading $\CC$.   




\textbf{Proof of Theorem \ref{thm:spectral_regression}.}
As a starting point, we observe that yet again, since $\ZZ^{\top}$ has orthonormal rows, the leverage scores are simply the $\ell^2_2$ norms of the columns of $\ZZ^{\top}$. 
Therefore, one possible approach is to construct a \textit{leverage score sampling} sketch $\SS$ for $\CC$, by sampling columns proportional to the leverage 
scores of $\ZZ^{\top}$. We note we can afford to sample at most $\widetilde{O}(k/\epsilon)$ columns, since our algorithm queries all entries in the resulting sketched matrix $\CC \SS$.  

Further, for reasons to be discussed later, it is crucial that we sample columns of $\CC$ independently, as opposed to the standard way of sampling with replacement we have used thus far. The independent sampling process can be described as follows:  for all $j \in [\sqrt{nk/\epsilon}]$, we sample $\CC^*_{*,j}$ with probability $\min(\|\ZZ^{\top}_{*,j}\|^2_2,1)$. We use the following lemma from \cite{cohenmm17} to show that independently sampling columns satisfies some desirable properties.

\begin{lemma}(Lemma 21, \cite{cohenmm17}.)
\label{lem:independent_sampling}
Given a matrix $\MM \in \mathbb{R}^{n \times m}$, for all $j \in [m]$ let $\bar{\rho}^k_j(\MM) = \Theta(\rho^k_j(\MM)$ be estimates of the rank-$k$ column ridge-leverage scores of $\MM$ and let $q_j= \min(\bar{\rho}^k_j(\MM)\log(k/\delta)/\epsilon^2,1)$. Then, construct $\MM\SS$ by selecting each column $\MM_{*,j}$ with probability $q_j$ and scale it by $1/\sqrt{q_j}$. Then, with probability at least $1-\delta$, $\MM \SS $  has $\sum_{j\in [m]}\bar{\rho}^k_j\cdot \log(k/\delta)/\epsilon^2$ columns and 
\begin{equation*}
    (1-\epsilon)\MM\SS \SS^{\top}\MM^{\top} - \frac{\epsilon}{k}\| \MM  - \MM_k\|^2_F\II \preceq \MM \MM^{\top} \preceq (1+\epsilon)\MM\SS \SS^{\top}\MM^{\top} + \frac{\epsilon}{k}\| \MM  - \MM_k\|^2_F\II
\end{equation*}
\end{lemma}

The above lemma independently samples columns proportional to the \textit{ridge leverage scores}.  
In our setting, we can set the ridge parameter $\lambda =0$, and sample according to the exact leverage scores of $\ZZ^{\top}$. Formally, let $q = \{q_1, q_2,\ldots, q_m\}$ be the corresponding distribution over columns of $\ZZ^{\top}$ such that $q_j = \min(\|\ZZ^{\top}\|^2_2 \log(k),1)$. Since $\ZZ^{\top}$ has $k/\epsilon$ orthonormal rows, the leverage scores sum up to $\textrm{rank}(\ZZ^{\top}) \leq k/\epsilon$. We then use Lemma \ref{lem:independent_sampling} by setting $\epsilon =1/10, \delta = 0.01$ and thus with probability at least $99/100$, $\ZZ^{\top} \SS$ has $\sum_{j \in [\sqrt{nk/\epsilon}]}\tau_j(\ZZ^{\top})\log(n) = \widetilde{O}(k/\epsilon)$ rows and 
\begin{equation}
\label{eqn:constant_se}
    \frac{9}{10}\ZZ^{\top} \SS \SS^{\top} \ZZ \preceq \ZZ^{\top} \ZZ \preceq \frac{11}{10} \ZZ^{\top} \SS \SS^{\top} \ZZ
\end{equation}
If the guarantee in Equation \ref{eqn:constant_se} holds for a sketch $\SS$, we refer to $\SS$ as satisfying an $O(1)$-Subspace Embedding property. Observe, this is equivalent to $\SS$ preserving all singular values of $\ZZ^{\top}$ up to a constant.  


We can now obtain a closed form solution for the  \textit{Spectral Regression} problem in the sketched space. 
By Lemma \ref{lem:characterizing_opt}, the optimal solution to the optimization problem in Step 2 of Algorithm \ref{alg:approx_spectral_regression} is given by $\widehat{\WW} = \CC\SS(\ZZ^{\top} \SS)^{\dagger}$. Since $\SS$ satisfies the $O(1)$-Subspace Embedding property in Equation \ref{eqn:constant_se}, it preserves the rank of $\ZZ^{\top}$. Therefore, $\ZZ^{\top} \SS$ has full row rank and $(\ZZ^{\top} \SS)^{\dagger} = \SS^{\top} \ZZ (\ZZ^{\top}\SS \SS^{\top} \ZZ)^{-1}$ and thus $\widehat{\WW} = \CC \SS \SS^{\top} \ZZ (\ZZ^{\top}\SS \SS^{\top} \ZZ)^{-1}$ is the optimal solution. Next, we bound the cost of $\widehat{\WW}$ in the original problem. Let $\PP_{\ZZ^{\top}} = \ZZ^{\dagger}\ZZ^{\top}$ be the orthogonal projection matrix onto $\ZZ^{\top}$. Using the fact that $\|\MM \|^2_2 = \max_{\| y\|_2=1}\|y^{\top}\MM\|^2_2$ and the Pythagorean Theorem for Euclidean space we have
\begin{equation}
\label{eqn:sketched_solution}
    \begin{split}
        \|\CC - \widehat{\WW}\ZZ^{\top}\|^2_2 & = \|\CC - \CC \SS \SS^{\top} \ZZ (\ZZ^{\top}\SS \SS^{\top} \ZZ)^{-1} \ZZ^{\top} \|^2_2\\
        & = \max_{\|y\|^2_2=1} \|y^{\top} \CC\PP_{\ZZ^{\top}}  - y^{\top}\CC  \SS \SS^{\top} \ZZ (\ZZ^{\top}\SS \SS^{\top} \ZZ)^{-1} \ZZ^{\top}\PP_{\ZZ^{\top}}\|^2_2 +\\
        & \hspace{0.58in}  \| y^{\top} \CC(\II-\PP_{\ZZ^{\top}}) - y^{\top}\CC  \SS \SS^{\top} \ZZ (\ZZ^{\top}\SS \SS^{\top} \ZZ)^{-1} \ZZ^{\top}(\II - \PP_{\ZZ^{\top}}) \|^2_2 \\
    \end{split}
\end{equation}
Here, we observe $y^{\top}\CC  \SS \SS^{\top} \ZZ (\ZZ^{\top}\SS \SS^{\top} \ZZ)^{-1} \ZZ^{\top}$ is a vector in the row space of $\ZZ^{\top}$ and$(\II - \PP_{\ZZ^{\top}})$ is the projection on the orthogonal complement of rowspan$(\ZZ^{\top})$, thus this evaluates to $0$. Since $\CC(\II - \PP_{\ZZ^{\top}}) = \CC^*$, we can upper bound $\|y^{\top}\CC(\II  - \PP_{\ZZ^{\top}})  \|_2 $ by $\|\CC^*\|^2_2$. Similarly, we can upper bound the first term by its spectral norm. Therefore, plugging this back into Equation \ref{eqn:sketched_solution},
\begin{equation}
\label{eqn:sketched_solution_2}
    \begin{split}
        \|\CC - \widehat{\WW}\ZZ^{\top}\|^2_2 & \leq \| \CC(\ZZ^{\top})^{\dagger}\ZZ^{\top}  - \CC \SS \SS^{\top} \ZZ (\ZZ^{\top}\SS \SS^{\top} \ZZ)^{-1} \ZZ^{\top}\|^2_2 + \| \CC^*\|^2_2 \\
        & = \| \left(\CC(\ZZ^{\top})^{\dagger}\ZZ^{\top}\SS \SS^{\top} \ZZ   - \CC \SS \SS^{\top} \ZZ \right)(\ZZ^{\top}\SS \SS^{\top} \ZZ)^{-1} \|^2_2 + \| \CC^*\|^2_2 \\
        & \leq \| \CC(\ZZ^{\top})^{\dagger}\ZZ^{\top}\SS \SS^{\top} \ZZ   - \CC \SS \SS^{\top} \ZZ\|^2_2 \|(\ZZ^{\top}\SS \SS^{\top} \ZZ)^{-1} \|^2_2 + \| \CC^*\|^2_2 \\
    \end{split}
\end{equation}
where we use that $\ZZ^{\top}$ has orthonormal columns and the sub-multiplicativity of the spectral norm. From Equation \ref{eqn:constant_se}, it follows that for all $i \in [k/\epsilon]$,  $\sigma_i^2(\ZZ^{\top}\SS\SS^{\top}\ZZ) = (1 \pm 0.1)^2\sigma^2_{i}(\ZZ^{\top}\ZZ) = (1 \pm 0.1)^2$. Therefore, $\|(\ZZ^{\top}\SS\SS^{\top}\ZZ)^{-1}\|^2_2 = 1/\sigma^2_{\min}(\ZZ^{\top}\SS\SS^{\top}\ZZ) \leq 100/81$. 
Substituting this back into Equation \ref{eqn:sketched_solution_2}, we have
\begin{equation}
    \begin{split}
        \|\CC - \widehat{\WW}\ZZ^{\top}\|^2_2 & \leq O(1) \|(\CC(\ZZ^{\top})^{\dagger} \ZZ^{\top} -\CC )\SS\SS^{\top}\ZZ\|^2_2 + \| \CC^*\|^2_2\\
        & \leq O(1) \|\CC^*\SS\SS^{\top}\ZZ \|^2_2 + \| \CC^*\|^2_2\\
    \end{split}
\end{equation}
where the last inequality follows from the definition of $\CC^*$. 
In order to bound the cost above, we focus on analyzing $\| \CC^*\SS\SS^{\top}\ZZ \|^2_2$. Since we want to compare $\|\CC^*\SS\SS^{\top}\ZZ\|^2_2$ to $\|\CC^* \ZZ\|^2_2$, a natural way to proceed would be to interpret this term as an instance of \textit{Approximate Matrix Product}. Therefore, we next show that the leverage score sampling matrix $\SS$ satisfies the \textit{Spectral AMM} property for $\CC^*$ and $\ZZ^{\top}$. Here, we want to analyze how sampling columns of $\CC^*$ proportional to the \textit{leverage scores} of $\ZZ^{\top}$ affects the spectrum of $\CC^*$. 
An important tool in this analysis is the following result by Rudelson and Vershynin on how the spectral norm of a matrix degrades when we sample a uniformly random subset of rows of a matrix: 

\begin{theorem}( Theorem 1.8 in \cite{rudelson2007sampling})
\label{thm:rudelson_vershynin_original}
Given a matrix $\AA \in \mathbb{R}^{n \times n}$, let $\mathcal{Q}$ be a uniformly random subset of $[n]$ s.t. $\expecf{}{\mathcal{Q}} = q$. Let $\AA_{|\mathcal{Q}}$ denote the submatrix restricted to the rows indexed by $\mathcal{Q}$. Then,
\begin{equation*}
     \expecf{}{\big\|\AA_{|\mathcal{Q}}\big\|_2} = O\left( \sqrt{\frac{q}{n}}\|\AA \|_2 + \sqrt{\log(q)}\|\AA \|_{(n/q)}\right)
\end{equation*}
where $\|\AA \|_{(n/q)}$ is the average of the largest $n/q$ $\ell_2$-norms of columns of $\AA$. 
\end{theorem}

We extend the above statement to rectangular matrices: 

\begin{corollary}[Spectral Decay for Rectangular Matrices]
\label{thm:rudelson_vershynin}
Given a matrix $\AA \in \mathbb{R}^{n \times m}$, s.t. for all $j,j' \in [m]$, $\| \AA_{*,j}\|^2_2 =\Theta(\| \AA_{*,j'}\|^2_2)$.
Let $\mathcal{Q}$ be a uniformly random subset of $[n]$ s.t. $\expecf{}{\mathcal{Q}} = q$. Let $b= \ceil{n/m}$ and $\AA_{|\mathcal{Q}}$ denote the submatrix restricted to the rows indexed by $\mathcal{Q}$. Then,
\begin{equation*}
     \expecf{}{\big\|\AA_{|\mathcal{Q}}\big\|_2} =  O\left( \sqrt{\frac{q}{n}}\|\AA \|_2 + \sqrt{\log(q)/b}\|\AA \|_{(n/q)}\right)
\end{equation*}
\end{corollary}
\begin{proof}
First, consider the case when $m \geq n$.  To see this, let $\textsf{SVD}(\AA) = \UU \Sig \VV^{\top}$ where $\UU \Sig$ is an $n \times n$ matrix. Now, $\|\AA \|_2 = \|\UU \Sig\|_2$ and applying Theorem \ref{thm:rudelson_vershynin_original} to $\UU\Sig$, we have
\begin{equation}
\label{eqn:more_cols}
\begin{split}
    \expecf{}{\big\|\AA_{|\mathcal{Q}}\big\|_2} =\expecf{}{\big\|(\UU\Sig)_{|\mathcal{Q}}\big\|_2} & = O\left( \sqrt{\frac{q}{n}}\|\UU \Sig \|_2 + \sqrt{\log(q)}\|\UU\Sig \|_{(n/q)}\right)\\
    & = O\left( \sqrt{\frac{q}{n}}\|\AA \|_2 + \sqrt{\log(q)}\|\AA \|_{(n/q)}\right)
\end{split}
\end{equation}
where we repeatedly use that $\VV^T$ has orthonormal rows. Here, we note that since the columns of $\AA$ have the same squared norm up to a constant,$\|\AA \|_{(n/q)} = \Theta(\|\AA \|_{1\to 2})$, i.e. the max column norm of $\AA$. 

Next, consider the case where $m< n$. Let $b = \ceil{n/m}$. In order to analyze the spectral norm of $\AA_{\mathcal{Q}}$, we create $b$ copies of $\AA$ and concatenate them such that the resulting matrix $\AA^*$ has more columns than rows. Applying Equation \eqref{eqn:more_cols} to $\AA^*$ and substituting the average with max, we have 
\begin{equation}
    \expecf{}{\big\|\AA^*_{|\mathcal{Q}}\big\|_2} = O\left( \sqrt{\frac{q}{n}}\|\AA^* \|_2 + \sqrt{\log(q)}\|\AA^* \|_{1\to 2}\right) 
\end{equation}
Observe, $\AA^*_{\mid \mathcal{Q}}$ selects uniformly random rows of $\AA^*$ and $\big\|\AA^*_{|\mathcal{Q}}\big\|_2 = \max_{\|x\|_2=1} \|x^{\top}\AA^*_{|\mathcal{Q}}\|_2$ and for any vector $x$, $\|x^{\top}\AA^*_{|\mathcal{Q}}\|_2 = \sqrt{b} \|x^{\top}\AA_{|\mathcal{Q}}\|_2 $. Therefore, $\expecf{}{\big\|\AA^*_{|\mathcal{Q}}\big\|_2} = \sqrt{b}\cdot \expecf{}{\big\|\AA_{|\mathcal{Q}}\big\|_2}$ and $\|\AA^* \|_2 = \sqrt{b} \cdot \|\AA^* \|_2$. Finally, it is easy to see that since the columns of $\AA^*$ are copies of columns of $\AA$, the max column norm does not change. Therefore, 
\eqref{eqn:more_cols} to $\AA^*$, we have 
\begin{equation}
    \expecf{}{\big\|\AA_{|\mathcal{Q}}\big\|_2} = O\left( \sqrt{\frac{q}{n}}\|\AA \|_2 + \sqrt{\log(q)/b}\|\AA \|_{1\to 2}\right)
    = O\left( \sqrt{\frac{q}{n}}\|\AA^* \|_2 + \sqrt{\log(q)}\|\AA^* \|_{(n/q)}\right) 
\end{equation}
and the claim follows.
\end{proof}

Intuitively, there are two technical challenges in applying Corollary \ref{thm:rudelson_vershynin}. First, a \textit{leverage score sampling} matrix need not sample columns uniformly at random, since we have no control over the column norms of $\ZZ^{\top}$. Second, the $\|\cdot\|_{(n/q)}$ norm only shrinks when all columns of $\AA$ have roughly the same squared norm. We overcome these challenges by partitioning the matrix, first according to row norms, such that each partition does indeed have the same row norm, up to a factor of $2$. Next, we further partition each such matrix according to the sampling probabilities, such that within each partition, the sampling process is \textit{close} to uniform sampling. Formally, 


\begin{lemma}(Weak Spectral Approximate Matrix Product.)
\label{lem:weak_amm}
Let $\ZZ, \CC^*$ and $\SS$ be as defined in Lemma \ref{lem:characterizing_opt}. Then, with probability at least $99/100$, $\SS$ satisfies $(\widetilde{O}(1), k/\epsilon)$-\textit{Spectral AMM}, i.e., 
\begin{equation*}
    \|\CC^* \SS \SS^{\top} \ZZ \|^2_2 \leq \widetilde{O}(1)\left( \frac{\epsilon}{k}\| \CC^*\|^2_F + \|\CC^*\|^2_2\right)
\end{equation*}
\end{lemma}
\begin{proof}
By sub-multiplicativity of the spectral norm and $\SS$ being an $O(1)$-subspace embedding for $\ZZ^{\top}$, we have 
\begin{equation}
    \begin{split}
        \|\CC^* \SS \SS^{\top} \ZZ \|^2_2 & \leq \|\CC^* \SS \|^2_2\cdot \|\SS^{\top}\ZZ \|^2_2 \\
        & \leq O(1) \|\CC^* \SS \|^2_2
    \end{split}
\end{equation}
where the second inequality follows from $\ZZ^{\top}$ having orthonormal rows. 

We begin by observing that Corollary \ref{thm:rudelson_vershynin} requires the squared row norms of $\CC^*$ to be roughly the same, which need not be the case in general. Note, here the sampling matrix subsamples columns of $\CC^*$, as opposed to rows in Corollary \ref{thm:rudelson_vershynin}.
Thus, we partition the rows of $\CC^*$ into $O(\log(n))$ blocks such that either the squared column norms are the same up to a factor of $2$ or they are at most $\|\CC^* \|^2_F/\poly(n)$. Formally, 
for all $\ell \in [c\log(n)]$, let  $$\mathcal{B}_\ell = \left\{ i \in [n] : \frac{\|\CC^*\|^2_F}{2^{\ell+1}}  \leq \| \CC^*_{i,*}\|^2_2 \leq \frac{\|\CC^*\|^2_F}{2^{\ell}}\right\}$$ 
represent the blocks for rows with large squared norm. Let $\mathcal{B}_r = [n] \setminus \cup_{\ell \in [\log(n)]} \mathcal{B}_{\ell}$ be the remaining rows, which have norm at most $\|\CC^*\|^2_F/\textrm{poly}(n)$. Since the set of indices in the blocks form a partition of the rows of $\CC^*$, we can write $\|\CC^*\|^2_F = \sum_{\ell \in [\log(n)]}\| \CC^*_{\mathcal{B}_\ell}\|^2_F + \|\CC^*_{\mathcal{B}_r}\|^2_F$. Similarly, we can bound the spectral norm as follows: 
\begin{equation}
\label{eqn:bounding_op_norm}
\begin{split}
    \|\CC^* \SS \|^2_2 = \max_{\|y\|^2_2=1} \|\CC^*\SS y\|^2_2 & \leq O\left(\sum_{\ell \in [\log(n)]} \|\CC^*_{\mathcal{B}_{\ell}}\SS y\|^2_2 + \|\CC^*_{\mathcal{B}_r}\SS y\|^2_2\right)\\
    & \leq O\left(\sum_{\ell \in [\log(n)]} \|\CC^*_{\mathcal{B}_{\ell}}\SS \|^2_2 + \|\CC^*_{\mathcal{B}_r}\SS \|^2_2\right)
\end{split}
\end{equation}
We now handle the two separately. Since $\SS $ is an unbiased estimator of the squared Frobenius norm of $\ZZ^{\top}$, it is an unbiased estimator of the squared Frobenius norm of $\CC^*$. Therefore, with probability at least $99/100$, 
\begin{equation}
\label{eqn:small_row_norm}
    \|\CC^*_{| \mathcal{B}_r}\SS\|^2 _2 \leq \|\CC^*_{| \mathcal{B}_r}\SS\|^2 _F = O( \|\CC^*_{| \mathcal{B}_r}\|^2_F) \leq \frac{\|\CC^* \|_F^2}{\textrm{poly}(n)} << \frac{\epsilon}{k} \|\CC^*\|^2_F
\end{equation}
For the remaining terms, we cannot use this n{\"a}ive analysis as this would only leave us with an upper bound of $\| \CC^*\|^2_F$, which is too large. 

If instead of a leverage score sampling matrix, $\SS$ were a uniform sampling sketch that samples $k/\epsilon$ columns of $\CC^*$ in expectation, we could apply Corollary \ref{thm:rudelson_vershynin} for each $\ell$, with $q = k/\epsilon$ and $n = \sqrt{nk/\epsilon}$ and $b =\ceil{ \sqrt{nk}/(\sqrt{\epsilon}|\mathcal{B}_{\ell}|)}$, to obtain
\begin{equation}
\label{eqn:uni_large_row_norm}
\begin{split}
    \expecf{}{\big\|\CC^*_{|\mathcal{B}_\ell}\SS \big\|^2_2}  =  \sqrt{\frac{n \epsilon}{k}} \expecf{}{\big\|(\CC^*_{|\mathcal{B}_\ell})_{Q} \big\|^2_2}
    & \leq O\left( \big\|\CC^*_{|\mathcal{B}_\ell}\big\|^2_2 + \frac{\epsilon \log(k/\epsilon)|\mathcal{B}_{\ell}|}{k} \big\|(\CC^*)^{\top}_{|\mathcal{B}_\ell}\big\|^2_{\sqrt{\epsilon n/k}}\right) \\
    & \leq O\left(  \big\|\CC^*_{|\mathcal{B}_\ell}\big\|^2_2 + \frac{\epsilon \log(k/\epsilon)}{k} \big\|\CC^*_{|\mathcal{B}_\ell}\big\|^2_{F}\right)
\end{split}
\end{equation}
where $Q$ is the subset of columns selected by $\SS$ and the second inequality follows from observing that the all the row norms of $(\CC^*)_{\mid \mathcal{B}_\ell}$ are within a factor of $2$ of each other and thus the max squared row norm times the size of the set is the squared Frobenius norm. 

Using Equations \ref{eqn:small_row_norm} and \ref{eqn:uni_large_row_norm} to upper bound the two terms in Equation \ref{eqn:bounding_op_norm} suffices to finish the proof. Unfortunately, a similar analysis does not immediately go through when we replace a uniform sampling matrix with a leverage score sketch. 
Instead, we partition the sketch $\SS$ into buckets such that 
each bucket corresponds to rows in $\SS$ that scale columns of $\CC^*$ within a factor of $2$. For notational convenience, 
let $m = \sqrt{nk/\epsilon}$ and $t = k/\epsilon$. Recall, we 
construct $\SS$ by sampling the $j$-th column of $\ZZ^{\top} $ 
independently with probability $q_j = \min(\|\ZZ_{j,*}\|^2_2 \log(k), 1)$ and scale this column by $1/\sqrt{q_j}$.  We 
group the scaling factors into buckets. Note, if for some $j$, $q_j < 1/n^3$, we can ignore the corresponding column.

Let $\zeta_j$ be the indicator for a column of $\ZZ^{\top}$ to be 
sampled by $\SS$. Then, $\Pr[\zeta_j =1] = q_j = \min(\|\ZZ^{\top}_{*,j}\|^2_2\log(k), 1)$. Since $q_j \leq 1/n^3 $,
we can union bound over at most $m$ such events and conclude 
with probability at least $1- 1/n^2$, for all $j \in m$, no 
column $\ZZ^{\top}_{*,j}$ is sampled such that $q_j \leq 1/n^3$.  
Further, since $q_j \leq 1$,  $1/\sqrt{q_j} \in [1, n^{1.5}]$. Therefore, it suffices to bucket values in the range $[1, n^{1.5}]$. 
For all $h \in [c\log(n)]$, let $\S$ denote the set of column 
indices from $\ZZ^{\top}$ that were sampled by the sketch $\SS$. Then,  
\begin{equation*}
    \T_{h} = \left\{ j \in \S :2^{h} \leq \frac{1}{\sqrt{q_j}} \leq 2^{h+1} \right\}
\end{equation*}
Let $\SS_{\T_h}$ be the subset of rows of $\SS$ which are indexed by the set $\T_h$. Since $t$ is fixed and the scaling factors in $\T_h$ differ by at most a factor of $2$, the corresponding sampling probabilities in $\DD$ differ by at most $\sqrt{2}$, which is still not uniform. To fix this, we change the sampling process and independently sample each column indexed by $j \in \T_h$ with probability $2^{h+1}$, while still scaling it by $1/\sqrt{tq_j}$. Let this new distribution be denoted by $q'$. Under the new sampling process, we now sample rows independently and therefore, we are at least as likely to see all the rows sampled by $\SS_{\T_h}$ in our new sampling process. Therefore, it now holds that 
\begin{equation*}
    \expecf{q}{\|\CC^*_{\mid \mathcal{B}_\ell}\SS_{\T_h}\|^2_2} \leq \expecf{q'}{\|\CC^*_{\mid \mathcal{B}_\ell}\SS_{\T_h}\|^2_2}
\end{equation*}
Further, in the new sampling process, each row restricted to the set $\T_h$ is uniformly sampled with probability $1/2^{(h+1)/2}$ and thus we can apply Corollary \ref{thm:rudelson_vershynin} to $\CC^*_{\mid \mathcal{B}_\ell}\SS_{\T_h}$.  
\begin{equation}
    \begin{split}
        \expecf{q'}{\|\CC^*_{\mid \mathcal{B}_\ell}\SS_{\T_h}\|^2_2} & \leq O\left( \|\CC^*_{\mid\mathcal{B}_\ell} \|^2_2 + \frac{\epsilon \log(k/\epsilon) |\mathcal{B}_{\ell}| }{k }  \big\|(\CC^*_{\mid \mathcal{B}_\ell})^{\top} \big\|_{(\sqrt{\epsilon n/k})} \right)  \\
        & \leq O\left( \|\CC^*_{\mid \mathcal{B}_\ell} \|^2_2 + \frac{\epsilon\log(k/\epsilon)}{k} \big\|\CC^*_{\mid\mathcal{B}_\ell} \big\|^2_F \right) 
    \end{split}
\end{equation}
where the second inequality follows from squared row norms in $\CC^*_{\mid\mathcal{B}_\ell}$ being equal up to a factor of $2$. 
Therefore, with probability at least $1- 1/c'\log(n)$, 
\begin{equation}
\label{eqn:spec_bound_bucket}
    \big\|\CC^*_{\mid \mathcal{B}_\ell}\SS_{\T_h} \big\|^2_2 \leq \widetilde{O}\left(\|\CC^*_{\mid \mathcal{B}_\ell} \|^2_2 + \frac{\epsilon\log(k/\epsilon)}{k} \|\CC^*_{\mid \mathcal{B}_\ell} \|^2_F \right) 
\end{equation}
Let $\eta_h$ be the event that the above bound holds. Then, union bounding over all $c\log(n)$ such events, with probability at least $99/100$, simultaneously for all $h$,  
\begin{equation}
\label{eqn:union_bound_bucket}
\begin{split}
    \|\CC^*_{\mid \mathcal{B}_\ell}\SS \|^2_2 & \leq O\left(\sum_{h \in [c\log(n)]} \|\CC^*_{\mid \mathcal{B}_\ell}\SS_{\T_h} \|^2_2\right)\\
    & \leq \widetilde{O}\left(  \|\CC^*_{\mid \mathcal{B}_\ell} \|^2_2 + \frac{\epsilon}{k} \|\CC^*_{\mid \mathcal{B}_\ell} \|^2_F \right)
\end{split}
\end{equation}
which follows from Equation \ref{eqn:spec_bound_bucket}.
Substituting this back into Equation \ref{eqn:bounding_op_norm}, 
\begin{equation}
\begin{split}
    \|\CC^* \SS \|^2_2  & \leq O\left(\sum_{\ell \in [\log(n)]} \|\CC^*_{\mid\mathcal{B}_{\ell}}\SS \|^2_2 + \|\CC^*_{\mid \mathcal{B}_r}\SS \|^2_2\right)\\
    & \leq \widetilde{O}\left( \|\CC^*\|^2_2 + \frac{\epsilon}{k} \|\CC^* \|^2_F \right) \\
\end{split}
\end{equation}
where the second inequality follows from Equation \ref{eqn:union_bound_bucket} and observing that $\|\CC^*_{\mathcal{B}_\ell} \|^2_2 \leq \|\CC^*\|^2_2 $ and $\sum_{\ell}\|\CC^*_{\mathcal{B}_\ell} \|^2_F = \|\CC^*\|^2_F$, which completes the proof. 
\end{proof}

Combining the above lemma with \ref{eqn:sketched_solution}, and observing that $\|\CC^*\|^2_F \leq \|\CC - \CC_{k/\epsilon}\|^2_F$, we can bound the cost of $\widehat{\WW}$ 
\begin{equation*}
    \|\CC - \widehat{\WW}\ZZ^{\top}\|^2_2 \leq \widetilde{O}\left(\min_{\WW} \|\CC - \WW\ZZ^{\top}\|^2_2 + \frac{\epsilon}{k}\|\CC- \CC_{k/\epsilon} \|^2_F\right)
\end{equation*}
which completes the correctness proof of Theorem \ref{thm:spectral_regression}. Next, we analyze the running time. In Step 1 of Algorithm \ref{alg:approx_spectral_regression}, we compute a distribution over the columns of $\ZZ^{\top}$, which does not require reading any entries in $\AA$ and takes time $\sqrt{nk/\epsilon} \cdot k/\epsilon = \sqrt{n}(k/\epsilon)^{1.5}$. Step 2 requires computing $\CC \SS$ and $\ZZ^{\top}\SS$. Note since $\SS$ samples $\widetilde{O}(k/\epsilon)$ columns in $\CC$, we have to query $n\cdot \widetilde{O}(k/\epsilon)$ entries in $\AA$ to explicitly compute $\CC\SS$ and can be computed in as much time. Since $\ZZ^{\top}$ has fewer rows the running time is dominated by computing $\CC\SS$. For Step 3, we compute $(\ZZ^{\top}\SS\SS^{\top}\ZZ)^{-1}$, which requires no queries to $\AA$ and runs in time $\widetilde{O}((k/\epsilon)^{\omega})$ and thus $(\ZZ^{\top} \SS)^{\dagger}$ can be computed in the same time. Therefore, the total query complexity of Algorithm \ref{alg:approx_spectral_regression} is $\widetilde{O}(nk/\epsilon)$ and the running time is $\widetilde{O}(nk/\epsilon + (k/\epsilon)^{\omega})$, which concludes the proof.  


\subsection{Sample-Optimal Algorithm}

In this subsection, we describe our main algorithm for PSD Low-Rank Approximation. Given a PSD matrix $\AA$, our algorithm queries $\widetilde{O}(nk/\epsilon)$ entries in $\AA$ and runs in time $\widetilde{O}(n(k/\epsilon)^{\omega-1} +  (k/\epsilon^3)^\omega)$. This resolves an open question on the $\epsilon$-dependence of the query complexity and matches the lower bound of $\Omega(nk/\epsilon)$ up to $\textrm{polylog}$ factors from \cite{mw17}. 
At a high level, our algorithm consists of two stages: first, we use the existing machinery developed by Musco and Woodruff to obtain weak PCPs by setting $\epsilon$ to be a constant. By observing that their algorithms have linear dependence on the rank, we can afford to rank-$(k/\epsilon)$ PCPs instead. This enables us to find a structured subspace that contains a spectral low-rank approximation for the PCP.

Since our PCPs are accurate only up to $O(1)$-error, we cannot directly extract a $(1+\epsilon)$ relative error approximation for $\AA$. However, we show that the PCPs have enough structure to obtain a structured subspace that spans a $(1+\epsilon)$-approximate solution for $\AA$. A key ingredient to recover this structured subspace is an efficient algorithm for \textit{Spectral Regression}.

Following the approach of Musco and Woodruff we use the ridge leverage scores of $\AA^{1/2}$ to compute $\CC$, a column PCP for $\AA$ and $\RR$ a row PCP for $\CC$, with a minor tweak: we instantiate their theorems (Lemmas \ref{lem:rel_col_pcp} and \ref{lem:rel_row_pcp_mixed}) with $k=k/\epsilon$ and $\epsilon=O(1)$.  
While the precise guarantees satisfied by our PCPs are weaker than the PCPs used by Musco and Woodruff, the dimensions of our PCPs are smaller.

\begin{Frame}[\textbf{Algorithm \ref{alg:sample_opt_rel}} : Sample Optimal PSD  Low-Rank Approximation]
\label{alg:sample_opt_rel}
\textbf{Input}: A PSD Matrix $\AA \in \mathbb{R}^{n \times n}$, integer $k$, and $\epsilon>0$.
\begin{enumerate}
    \item  Let $t= c\sqrt{\frac{nk}{\epsilon}}\log(n)$, for some constant $c$ and let $k' = \widetilde{O}(k/\epsilon)$. For all $j \in [n]$, let $\bar{\rho}^{k'}_j(\AA^{1/2})$ be the approximate column ridge-leverage scores that satisfy Lemma \ref{lem:approx_ridge_leverage}. Let $q = \{ q_1, q_2 \ldots q_n \} $ denote a distribution over columns of $\AA$ such that $q_j = \rho_j^{k'}(\AA^{1/2})/\sum_{j}\rho_j^{k'}(\AA^{1/2})$. 
    \item Construct a \textit{column PCP} for $\AA$ by sampling $t$ columns of $\AA$ such that each column is set to $\frac{\AA_{*,j}}{\sqrt{t q_j}}$ with probability $q_j$, for all $j \in [n]$. Let $\CC$ be the resulting $n \times t$ matrix that satisfies the guarantee of Lemma \ref{lem:rel_col_pcp} instantiated with $k =k'$ and $\epsilon=O(1)$.
    \item  Construct a \textit{row PCP} for $\CC$ by sampling $t$ rows of $\CC$ such that each row is set to $\frac{\CC_{i,*}}{\sqrt{t q_i}}$ with probability $q_i$, for all $i \in [n]$. Let $\RR$ be the resulting $t \times t$ matrix that satisfies the guarantee of Lemma \ref{lem:rel_row_pcp_mixed} instantiated with $k = k/\epsilon$ and $\epsilon=O(1)$.
    \item Run the \textit{input-sparsity} algorithm from Lemma \ref{lem:input_sparsity_spectral_lra} to compute a rank-$k/\epsilon$  matrix $\ZZ$ with orthonormal columns such that $\|\RR - \RR\ZZ\ZZ^{\top} \|^2_2\leq O\left(\frac{\epsilon}{k}\right)\|\RR - \RR_{k/\epsilon}\|^2_F $.
    \item Run Algorithm \ref{alg:approx_spectral_regression} with parameters $k$, $\epsilon$ on the \textit{Spectral Regression} problem 
    $$\min_{\WW}\|\CC - \WW \ZZ^{\top}\|_2$$
    Let $\widehat{\WW}$ be the output of  Algorithm \ref{alg:approx_spectral_regression}. Compute an orthonormal basis $\QQ$ for $\WW$. Note, $\QQ\QQ^{\top}$ is an $(O(1),k/\epsilon)$-\textsf{SF} projection for $\AA$.  
    \item Run Algorithm \ref{alg:proj_to_rank} with input $\AA$, $\QQ$, $k$ and $\epsilon$ to approximately minimize $\|\AA  - \QQ \XX\QQ^{\top} \|^2_F$ over rank $k$ matrices $\XX$. Let $\MM, \NN$ be the output of Algorithm \ref{alg:proj_to_rank}.
\end{enumerate}
\textbf{Output:} $\MM , \NN^{\top} \in \mathbb{R}^{n \times k}$ such that $\|\AA - \MM \NN\|^2_F \leq (1+\epsilon)\|\AA - \AA_k\|^2_F$   
\end{Frame}

 In particular, we obtain a row PCP $\RR$, which is a $\sqrt{nk/\epsilon}\times \sqrt{nk/\epsilon}$ matrix (ignoring polylogarithmic factors) and we can afford to read all of it. The input-sparsity time algorithm from Lemma \ref{lem:input_sparsity_spectral_lra}  queries $\textsf{nnz}(\RR) = \widetilde{O}(nk/\epsilon)$ entries to obtain a  rank-$(k/\epsilon)$ matrix $\ZZ$ with orthonormal columns such that 
\begin{equation}
\label{eqn:sf_projection}
    \|\RR - \RR \ZZ \ZZ^{\top} \|^2_2 \leq \frac{\epsilon}{k} \|\RR  - \RR_{k}\|^2_F
\end{equation}
Since $\RR$ is a \textit{Spectral-Frobenius PCP} for $\CC$, $\ZZ\ZZ^{\top}$ satisfies $\|\CC - \CC\ZZ\ZZ^{\top} \|^2_2 \leq \frac{\epsilon}{k}\|\CC -\CC_{k} \|^2_F$. Since $\CC$ is a \textit{Spectral-Frobenius PCP} for $\AA$, it suffices to obtain a projection for the column space of $\CC$ that also satisfies the above guarantee. Therefore, we solve the following \textit{Spectral Regression} problem: $\min_{\WW}\|\CC - \WW\ZZ^{\top}\|^2_2$. Recall, we can approximately optimize this using Algorithm \ref{alg:approx_spectral_regression}. Let $\widehat{\WW}$ be the resulting solution. 
We can then compute an orthonormal basis for $\widetilde{\WW}$ (denoted by $\QQ$) and show that $\QQ \QQ^{\top}$ is an $(O(1),k/\epsilon)$-\textsf{SF} projection for $\AA$. Then, we can obtain a low rank approximation for $\AA$ by simply running Algorithm \ref{alg:proj_to_rank}.

\textbf{Proof of Theorem \ref{thm:sample_opt_psd_lra}.}
Let $k' = \widetilde{O}(k/\epsilon)$.
It follows from Lemma \ref{lem:approx_ridge_leverage} that we can compute the rank-$k'$ ridge leverage scores of $\AA^{1/2}$, up to a constant factor using the algorithm of Musco and Musco \cite{musco2017recursive}. By Lemma \ref{lem:a_half_to_a}, the ridge leverage scores of $\AA^{1/2}$ are a $\sqrt{\epsilon n/k'}$-approximation to the ridge leverage scores of $\AA$. Let $q$ be a distribution over rows and columns of $\AA$ as defined in Algorithm \ref{alg:sample_opt_rel}. Since we sample $t=O(\sqrt{nk/\epsilon}\log(n))$ columns of $\AA$ proportional to $q$,  instantiating Lemma \ref{lem:rel_row_pcp_mixed} with $k=k'$ and $\epsilon=0.1$, we obtain a mixed Spectral-Frobenius column PCP $\CC$ such that with probability at least $1-c_1$, for all rank-$k'$ projections $\XX$, 
\begin{equation}
\label{eqn:col_pcp}
\frac{9}{10}\|\AA - \XX\AA \|^2_2 - \frac{1}{10k'} \|\AA - \AA_{k'} \|^2_F \leq \|\CC - \XX \CC \|^2_2 \leq \frac{11}{10}\|\AA - \XX\AA \|^2_2 + \frac{1}{10k'} \|\AA - \AA_{k'} \|^2_F
\end{equation}
Let $\zeta_1$ be the indicator for $\CC$ satisfying the above guarantee. 
Similarly, sampling $t$ rows of $\CC$ proportional to $q$, results in a mixed Spectral-Frobenius row PCP for $\RR$ such that with probability at least $1-c_2$, for all rank-$k'$ projection matrices $\XX$,
\begin{equation}
\label{eqn:row_pcp}
\frac{9}{10}\|\CC - \CC\XX \|^2_2 - \frac{1}{10k'} \|\CC - \CC_{k'} \|^2_F \leq \|\RR - \RR \XX \|^2_2 \leq \frac{11}{10}\|\CC - \CC\XX \|^2_2 + \frac{1}{10k'} \|\CC - \CC_{k'} \|^2_F
\end{equation}
Further, it is well-known that with the same probability $\| \RR - \RR_{k'}\|^2_F = \|\CC - \CC_{k'}\|^2_F$. Let $\zeta_2$ be the event that $\RR$ satisfies the above guarantee. 
Next, we compute a Spectral Low-Rank Approximation for $\RR$, using the algorithm from Lemma \ref{lem:input_sparsity_spectral_lra}, with $k = k'$ and $\epsilon =0.1$. As a result, with probability at least $1- c_3$, we obtain a rank-$k'$ matrix $\ZZ\in \mathbb{R}^{t \times k}$, such that $\ZZ \ZZ^{\top}$ is a $(0.1, k')$-\textsf{SF} projection for $\RR$, i.e., 
\begin{equation}
\label{eqn:sf_projection_1}
    \|\RR - \RR \ZZ \ZZ^{\top} \|^2_2 \leq \frac{1}{10k'} \|\RR  - \RR_{k'}\|^2_F
\end{equation}
Let $\zeta_3$ be the event that $\ZZ$ satisfies the above guarantee. Union bounding over $\zeta_1, \zeta_2, \zeta_3$, we know that all of them hold with probability at least $1- (c_1 + c_2 + c_3)$. 
Since $\RR$ is a Spectral-Frobenius row PCP for $\CC$ and $\ZZ \ZZ^{\top}$ is a rank-$k'$ projection matrix,  it follows from Equation \ref{eqn:row_pcp}
\begin{equation}
\label{eqn:sf_proj_c}
\begin{split}
    \|\CC - \CC \ZZ \ZZ^{\top} \|^2_2 & \leq \frac{10}{9} \|\RR - \RR\ZZ\ZZ^{\top} \|^2_2 + \frac{1}{9k'}\|\RR - \RR_{k'}\|^2_F \\
    & \leq \frac{1}{10k'}\|\RR - \RR_k \|^2_F + \frac{1}{9k'}\|\RR - \RR_{k'}\|^2_F\\
    & \leq \widetilde{O}\left(\frac{\epsilon}{k}\right) \|\CC - \CC_{k'} \|^2_F
\end{split}
\end{equation}
where the second inequality follows from Equation \ref{eqn:sf_projection} and the third follows from the fact that PCPs preserve Frobenius norm up to a constant factor. While conditioning on $\zeta_3$, it follows from Equation \ref{eqn:sf_proj_c} that $\ZZ \ZZ^{\top}$ is an $(\widetilde{O}(1), k/\epsilon)$-\textsf{SF} projection for $\CC$, our goal is to compute an \textsf{SF} projection for $\AA$. Since $\ZZ \ZZ^{\top}$ is a $t \times t$ matrix, it does not even match the dimensions of $\AA$. Therefore, we set up the following \textit{Spectral Regression} problem: 
\begin{equation}
\label{eqn:sr_2}
    \min_{\WW \in \mathbb{R}^{n \times k'}} \|\CC - \WW \ZZ^{\top} \|^2_2
\end{equation}
Let $\widehat{\WW}$ be the approximate minimizer of the above problem obtained by running Algorithm \ref{alg:approx_spectral_regression}. Then, it follows from Theorem \ref{thm:spectral_regression} that with probability at least $99/100$, 
\begin{equation}
\label{eqn:spec_reg_output}
\begin{split}
\|\CC - \widehat{\WW}\ZZ^{\top} \|^2_2 & \leq \widetilde{O}(1)\left( \min_{\WW} \|\CC - \WW \ZZ^{\top} \|^2_2 + \frac{\epsilon}{k}\|\CC- \CC_{k'} \|^2_F \right) \\
& \leq \widetilde{O}(1)\left( \|\CC - \CC_{k'} \|^2_2 + \frac{\epsilon}{k}\|\CC- \CC_{k'} \|^2_F \right) \\
& \leq \widetilde{O}(1)\left( \frac{\epsilon}{k}\|\CC- \CC_{k'} \|^2_F \right) 
\end{split}
\end{equation}
where the second inequality follows from $\|\CC - \CC\ZZ\ZZ^{\top}\|^2_2 \leq \widetilde{O}\left(\frac{\epsilon}{k}\right)\|\CC - \CC_{k'} \|^2_F$ (by definition of an \textsf{SF} projection) and observing that $\WW =\CC\ZZ^{\top}$ is a feasible solution to Equation \ref{eqn:sr_2}. Let $\zeta_4$ be the event that Equation \ref{eqn:spec_reg_output} holds. 
Next, let $\QQ$ be an orthonormal basis for $\WW$. We observe that $\QQ\QQ^{\top}\CC$ is the orthogonal projection of $\CC$ onto the subspace spanned by $\QQ$ and the matrix $\widehat{\WW}\ZZ^{\top}$ also lies in the subspace. Therefore, by the Pythagorean Theorem, for any fixed unit vector $y$, 
\begin{equation*}
    \|\CC y - \QQ\QQ^{\top}\CC y \|^2_2 \leq \|\CC y - \widehat{\WW}\ZZ^{\top}y \|^2_2 \leq \|\CC  - \widehat{\WW}\ZZ^{\top} \|^2_2
\end{equation*}
Picking $y$ such that $\| \CC y - \QQ\QQ^{\top}\CC y \|^2_2 = \|\CC - \QQ\QQ^{\top}\CC \|^2_2$, and combining it with Equation \ref{eqn:spec_reg_output} we have
\begin{equation}
\label{eqn:spec_bound_projective}
\|\CC - \QQ\QQ^{\top}\CC \|^2_2 \leq \widetilde{O}(1)\left( \frac{\epsilon}{k}\|\CC- \CC_{k'} \|^2_F \right) 
\end{equation}
Conditioning on event $\zeta_1$, we know that $\|\CC - \CC_{k'} \|^2_F = \|\AA - \AA_{k'}\|^2_F$. 
Since $\QQ\QQ^{\top}$ is a rank-$k'$ projection matrix and $\CC$ is a mixed Spectral-Frobenius column PCP for $\AA$, it follows from Equation \ref{eqn:col_pcp},
\begin{equation}
\begin{split}
     \|\AA - \QQ \QQ^{\top} \AA \|^2_2 & \leq \frac{10}{9}\|\CC - \QQ\QQ^{\top}\CC \|^2_2 + \frac{1}{9k'}\|\AA - \AA_{k'} \|^2_F \\
     & \leq \widetilde{O}(1)\left( \frac{\epsilon}{k}\|\AA- \AA_{k'} \|^2_F \right) 
\end{split}
\end{equation}
where the last inequality follows from Equation \ref{eqn:spec_bound_projective}. Therefore, $\QQ \QQ^{\top}$ is a $(0.1, k')$-\textsf{SF} projection for $\AA$. Finally, we run Algorithm \ref{alg:proj_to_rank} on $\min_{\textrm{rank}(\XX)=k} \|\AA - \QQ\XX\QQ^{\top}\|^2_F$. Here, we note that for Algorithm \ref{alg:proj_to_rank}, a $(0.1, k')$-\textsf{SF} projection is equivalent to an $(\epsilon, k)$-\textsf{SF} projection, up to polylogarithmic factors. Therefore, Theorem \ref{thm:struc_proj_to_lra} holds as is. Then,
by Theorem \ref{thm:struc_proj_to_lra}, we know that with probability at least $99/100$ Algorithm \ref{alg:proj_to_rank} outputs matrices $\MM, \NN$ such that $\|\AA - \MM \NN^{\top} \|^2_F \leq (1+\epsilon)\| \AA - \AA_k\|^2_F$. Let $\zeta_5$ be the event that the aforementioned algorithm succeeds. Then, union bounding over $\zeta_1, \zeta_2, \zeta_3, \zeta_4$ and $\zeta_5$, with probability at least $9/10$, $\MM, \NN$ is a relative-error Low-Rank Approximation for $\AA$, which concludes correctness. 

Next, we analyze the query complexity and running time of Algorithm \ref{alg:sample_opt_rel}. Step $1$ computes the rank-$k'$ \textit{ridge leverage scores} of $\AA^{1/2}$ and by Lemma \ref{lem:approx_ridge_leverage}, requires reading  $O(nk'\log(k')) = \widetilde{O}(nk/\epsilon)$ entries in $\AA$ and runs in time $\widetilde{O}(n (k/\epsilon)^{\omega-1})$. Steps $2$ and $3$ require no queries to $\AA$ and the sampling can be performed in $O(n)$ time.  In step $4$, the \textit{input sparsity} algorithm from Lemma \ref{lem:input_sparsity_spectral_lra} queries $\textsf{nnz}(\RR) = t^2 = \widetilde{O}(nk/\epsilon)$ entries in $\AA$ and runs in $\widetilde{O}(nk/\epsilon + \sqrt{n}\textrm{poly}(k/\epsilon)) = \widetilde{O}(nk/\epsilon)$ time. We know from Theorem \ref{thm:struc_proj_to_lra} that Step $5$ requires $\widetilde{O}(nk/\epsilon)$ queries to $\AA$ and runs in $\widetilde{O}(nk/\epsilon + (k/\epsilon)^{\omega})$ time. Finally, in Step $6$, we run Algorithm \ref{alg:proj_to_rank} such that $\QQ$ is a $n \times k'$ matrix. Therefore, it follows from Theorem \ref{thm:struc_proj_to_lra},  that the total number of queries to $\AA$ is $\widetilde{O}(nk/\epsilon + k^2/\epsilon^6)$ and the running time is $\widetilde{O}(n(k/\epsilon)^{\omega-1} +  (k/\epsilon^3)^\omega)$. The final query complexity and running time follows, and this concludes the proof. 

\vspace{0.2in}
\noindent \textbf{Outputting a PSD Low-Rank Approximation.} Here, we extend our algorithm to show that we can obtain a relative-error low-rank approximation matrix $\BB$ such that $\BB$ itself is a PSD matrix, using the same sample complexity and running time as in Theorem \ref{thm:sample_opt_psd_lra}. Outputting a PSD low-rank approximation was first considered by Clarkson and Woodruff \cite{clarkson2017low}, who obtain an input-sparsity algorithm for arbitrary $\AA$. When $\AA$ is PSD, Musco and Woodruff show that this problem can be solved with $\widetilde{O}(nk/\epsilon^3+ nk^2/\epsilon^2)$ queries, in time $\widetilde{O}(n (k/\epsilon)^{\omega} + nk^{\omega-1}/\epsilon^{3(\omega-1)})$.  

We run Algorithm \ref{alg:sample_opt_rel} till Step 5, i.e., we recover $\QQ$ such that $\QQ \QQ^{\top}$ is a \textsf{SF} projection for $\AA$. We then modify Algorithm \ref{alg:proj_to_rank} by considering the following optimization problem instead: 
\begin{equation}
\label{eqn:psd_lra}
    \min_{\substack{\textrm{rank}(\XX) \leq k \\ \XX \succeq 0 } }\|\AA - \QQ\XX\QQ^{\top} \|^2_F
\end{equation}
As before, we sketch on both sides by sampling proportional to the leverage scores of $\QQ$. Let the resulting sampling matrices be denoted by $\SS, \TT$. Then, we have the following sketched optimization problem: 
\begin{equation}
\label{eqn:sketched_psd_lra}
    \min_{\substack{\textrm{rank}(\XX) \leq k \\ \XX \succeq 0 } }\|\SS\AA\TT - \SS\QQ\XX\QQ^{\top}\TT \|^2_F
\end{equation}
Following Step 4 in Algorithm \ref{alg:proj_to_rank}, we can compute $\SS \AA\TT$, $\PP_{\SS\QQ}$, $\PP_{\QQ^{\top}\TT}$. We then compute $\widehat{\XX} = (\SS \QQ)^{\dagger} \PP_{\SS \QQ} \SS\AA\TT \PP_{\QQ^{\top}\TT}(\QQ^{\top}\TT)^{\dagger}$ and  $\XX^* = [(\widehat{\XX} + \widehat{\XX}^{\top})/2]_{k+}$, where for any matrix $\MM$, $[\MM]_{k+}$ is defined by setting all but the top-$k$ positive eigenvalues to $0$. Finally, we output $\NN \NN^{\top}$ where $\NN = \QQ (\XX^*)^{1/2}$.

\begin{corollary}(Outputting a PSD Low-Rank Approximation.)
Given an $n \times n$ PSD matrix $\AA$, an integer $k$, and $1>\epsilon >0$, there exists an algorithm that samples $\widetilde{O}(nk/\epsilon)$ entries in $\AA$ and outputs a rank-$k$ $\MM \MM^{\top} $ such that with probability at least $9/10$, 
\begin{equation*}
    \|\AA - \MM \MM^{\top}\|^2_F \leq (1+\epsilon) \|\AA - \AA_k\|^2_F
\end{equation*}
Further, the algorithm runs in $\widetilde{O}(n(k/\epsilon)^{\omega-1} +  (k/\epsilon^3)^\omega)$ time.
\end{corollary}

\begin{proof}
We first note that an extension of Lemma \ref{lem:structured_projection} holds for outputting a PSD matrix as well. As a consequence of the following lemma, obtaining an approximate solution to the optimization problem in Equation \ref{eqn:psd_lra} suffices.

\begin{lemma}(Structured Projections and PSD LRA \cite{clarkson2017low}.)
\label{psd_structured_projection}
Let $\PP \in \mathbb{R}^{n \times n}$ be an $(\epsilon,k)$-\textsf{SF} projection w.r.t $\AA$, then
\begin{equation*}
    \|\AA - \PP \AA_{k+} \PP \|^2_F \leq (1+\epsilon) \|\AA - \AA_{k+}\|^2_F
\end{equation*}
\end{lemma}

\noindent We then use the analysis of Lemma 15 from \cite{clarkson2017low} to conclude that $\XX^*$ is the minimizer for Equation \ref{eqn:sketched_psd_lra}. Finally, we note that the running time and query complexity is dominated by computing $\QQ$ and thus is the same as Theorem \ref{thm:sample_opt_psd_lra}. Computing $\XX^*$ requires no additional queries to $\AA$ and only contributes a lower order term to the running time. 
\end{proof}

\subsection{Negative-Type Distances}
In this subsection, we consider the problem of computing low-rank approximation for distance matrices. Here, the input matrix $\AA$ is formed by the pairwise distances
between a set of points $\PP = \{p_1, \ldots, p_n\}$ in an underlying metric space $d$, i.e., $A_{i,j} = d(p_i, p_j)$. Low-rank approximation for distance matrices was introduced by Bakshi and Woodruff \cite{bakshi2018sublinear} who obtained sublinear time \textit{additive-error} algorithms for arbitrary metrics. Subsequently, Indyk et. al. \cite{indyk2019sample} provided sample-optimal algorithms for additive-error low-rank approximation. For arbitrary distance matrices, it is known that relative-error algorithms require $\Omega(\textrm{nnz}(\AA))$ queries \cite{bakshi2018sublinear}.

Here, we focus on the special case of negative-type (Euclidean Squared) metrics \cite{schoenberg1938metric}. Negative-type metrics have numerous applications in algorithm design since it is possible to optimize over them using a semidefinite program (SDP). One significant algorithmic application of negative-type metrics appears in the Arora-Rao-Vazirani algorithm for the Sparsest Cut problem \cite{arora2009expander}. We refer the reader to extensive subsequent work on embeddability of such metrics and the references therein \cite{arora2008euclidean,arora2007frechet,chawla2005embeddings}. It is well-known that negative-type metrics include $\ell_1$ and $\ell_2$ metrics, spherical metrics and hyper metrics \cite{deza2009geometry,terwilliger1987classification}.  Therefore, our algorithms extend to distance matrices that arise from all such metrics. 

For negative-type metrics, Bakshi and Woodruff obtain a bi-criteria relative-error low-rank approximation algorithm that queries $\widetilde{O}(nk/\epsilon^{2.5})$ entries in $\AA$ and output a rank $k+4$ matrix. In contrast, we obtain a sample-optimal algorithm that does not require a bi-criteria guarantee. As noted above, our algorithm works for any distance matrix where the distance can be realized as a negative-type metric.

\begin{theorem}[Sample-Optimal Negative-Type LRA]
\label{thm:optimal_euclidean_lra}
Let $\AA \in \mathbb{R}^{n \times n}$ be a negative-type distance matrix. Given $\epsilon>0$ and $k\in [n]$, there exists an algorithm that queries $\widetilde{O}(nk/\epsilon)$ entries in $\AA$ and outputs matrices $\MM, \NN^{\top} \in \mathbb{R}^{n \times k}$ such that with probability $99/100$,
\begin{equation*}
    \|\AA - \MM \NN \|^2_F \leq (1+\epsilon)\|\AA - \AA_k \|^2_F
\end{equation*}
Further, the algorithm runs in time $\widetilde{O}(n (k/\epsilon)^{\omega-1})$. 
\end{theorem}

To demonstrate the connection between negative-type metrics and PSD matrices, we observe that a negative-type distance matrix $\AA$ can be realized as the distances corresponding to a point set $\mathcal{P}=\{x_1, x_2, \ldots x_n\}$ such that $\AA_{i,j}= \|x_i - x_j\|^2_2 = \|x_i \|^2_2 + \| x_j\|^2_2 - 2\langle x_i,x_j\rangle$. Therefore, we can rewrite $\AA$ as $\RR_1 + \RR_2 - 2\BB$, where for all $j\in[n]$, $(\RR_1)_{i,j}=\|x_i \|^2_2$, $\RR_2 = \RR_1^{\top}$ and $\BB$ is PSD. Further, we can obtain query access to $\BB$ by simply assuming w.l.o.g. that $x_1$ is centered at the origin and the $i$-th entry in the first row corresponds to $\|x_i\|^2_2$. Therefore, we can simulate our PSD low-rank approximation algorithms on the matrix $\BB$ by only having query access to $\AA$. 

Our main contribution here is to show that if $\PP=\QQ\QQ^{\top}$ is an $(O(1),k/\epsilon)$-\textsf{SF} projection matrix for $\BB$, then adjoining $\QQ^{\top}$ with the row span of $\RR_1$ and $\RR_2$ results in an \textsf{SF}-projection matrix for $\AA$. Here, the row span of $\RR_1$ is $\mathbf{1}^{\top}/\sqrt{n}$ and $\RR_2$ is $v$ such that for all $i \in [n]$ $v_i = \|x_i\|^2_2/\sum_{i}\|x_i\|^2_2$. We note that once we obtain an \textsf{SF} projection for $\AA$, we can run Algorithm \ref{alg:proj_to_rank} to output a $(1+\epsilon)$ relative-error low-rank approximation. 

\begin{lemma}[Structured Projections for Distance Matrices]
\label{lem:struct_proj_distance_matrix}
Let $\AA$ be a negative-type matrix such that $\AA = \RR_1 + \RR_2 - 2\BB$, as defined above and let $\epsilon>0$. Given an $(O(1), k/\epsilon)$-\textsf{SF} projection $\PP=\QQ \QQ^{\top}$ for $\BB$, let $\Om^{\top}$ be a basis for $\QQ^{\top}$ appended with the basis vectors for rowspan$(\RR_1)$ and rowspan$(\RR_2)$. Then, with probability at least $99/100$, 
\begin{equation*}
    \min_{\textrm{rank}(\XX) \leq k}\|\AA - \Om \XX \Om^{\top} \|^2_2 \leq (1+\epsilon)\|\AA - \AA_k \|^2_F
\end{equation*}
\end{lemma}
\begin{proof}
By Lemma 7 in \cite{clarkson2017low}, for any symmetric matrices $\YY,\ZZ$ such that $(\YY-\ZZ)\ZZ=0$ and projection matrix $\PP$, the following holds:
\begin{equation}
\label{eqn:projection_equality}
    \|\YY - \PP \ZZ\PP \|^2_F = \|\YY -\ZZ \|^2_F + \|\ZZ - \PP\ZZ\PP \|^2_F + 2 \trace{(\YY-\ZZ)(\II-\PP)\ZZ\PP}
\end{equation}
Applying Equation \ref{eqn:projection_equality} with $\YY =\AA$ and $\ZZ=\AA_k$, for any projection matrix $\PP$, we have 
\begin{equation}
\label{eqn:decomposition}
    \|\AA - \PP \AA_k\PP \|^2_F = \|\AA -\AA_k \|^2_F + \|\AA_k - \PP\AA_k\PP \|^2_F + 2 \trace{(\AA-\AA_k)(\II-\PP)\AA_k\PP}
\end{equation}
Next, we bound the $\|\AA_k - \PP\AA_k\PP \|^2_F$  as follows: 
\begin{equation}
\label{eqn:middle_term}
    \begin{split}
        \|\AA_k - \PP\AA_k\PP \|^2_F & \leq 2 \|\AA_k (\II - \PP) \|^2_F \\
        & \leq 2k \|\AA_k (\II - \PP) \|^2_2\\
        & \leq 2k \|\AA(\II - \PP)  \|^2_2\\
    \end{split}
\end{equation}
To bound the trace, we use the Von Neuman trace inequality,
\begin{equation}
\label{eqn:last_term}
\begin{split}
    2\trace{(\AA-\AA_k)(\II-\PP)\AA_k\PP} & =  2 \trace{(\AA-\AA_k)(\II-\PP)^2 \AA_k\PP} \\
    & \leq 2 \sum_{i \in [n]} \sigma_i((\AA-\AA_k)(\II-\PP)) \sigma_i((\II -\PP)\AA_k \PP)\\
    & \leq 2k \|(\AA-\AA_k)(\II - \PP) \|_2 \|(\II -\PP)\AA_k \PP \|_2\\
    & \leq 2k \|\AA(\II -\PP) \|^2_2
\end{split}
\end{equation}
It suffices to bound $\|\AA(\II-\PP)\|^2_2$ for $\PP = \Om \Om^{\top}$. Since $\AA = \RR_1 + \RR_2 -2\BB$ and $(\RR_1 + \RR_2)(\II - \PP) = 0$, we have 
\begin{equation*}
\begin{split}
    \|\AA(\II-\PP)\|^2_2 &\leq 2\|\BB (\II - \PP)\|^2_2 \\
    & \leq 2\|\BB (\II - \QQ\QQ^{\top})\|^2_2\\
    & \leq O\left(\frac{\epsilon}{k}\right)\|\BB - \BB_{k+2} \|^2_F \\
\end{split}
\end{equation*}
To relate $\|\BB - \BB_{k}\|^2_F$ back to $\AA$, observe 
\begin{equation*}
    \begin{split}
        \|\AA - \AA_k \|^2_F = \|\RR_1 + \RR_2 - 2\BB - \AA_k \|^2_F &= 4\| \BB - (\RR_1+\RR_2 - \AA_k)/2\|^2_F \\
        & \geq 4\|\BB - \BB_{k+2} \|^2_F
    \end{split}
\end{equation*}
Therefore, $\|\AA(\II- \PP)\|^2_2 \leq O(\epsilon/k)\|\AA - \AA_k\|^2_F$. We can thus bound Equations \ref{eqn:last_term} and \ref{eqn:middle_term} with $O(\epsilon)\|\AA - \AA_k \|^2_F$. Substituting this into Equation \ref{eqn:decomposition}, we conclude that $\|\AA - \PP \AA_k \PP \|^2_F \leq (1+O(\epsilon))\| \AA-\AA_k\|^2_F$, for $\PP = \Om \Om^{\top}$ and the claim follows. 
\end{proof}

Recall, we can compute an \textsf{SF} projection for the PSD matrix $\BB$ efficiently using Algorithm \ref{alg:sample_opt_rel} and then solve the optimization problem in Lemma \ref{lem:struct_proj_distance_matrix} using Algorithm \ref{alg:proj_to_rank}.  We can therefore reduce low-rank approximation of negative-type matrices to PSD low-rank approximation with only $O(n)$ additional queries and Theorem \ref{thm:optimal_euclidean_lra} follows. 

\subsection{Ridge Regression}
We consider the following regression problem: given a PSD matrix $\AA$, a vector $y$ and a ridge parameter $\lambda$, 
$$\min_{x}\|\AA x - y \|^2_2 + \lambda \| x\|^2_2.$$
As a corollary of Theorem \ref{thm:sample_opt_psd_lra}, we obtain a faster algorithm for the aforementioned problem. We begin with the following simple lemma from \cite{mw17}:

\begin{lemma}[Lemma 26 in \cite{mw17}]
Given a PSD matrix $\AA$, vector $y$, and $\lambda>0$, let $\BB$ be a matrix such that $\|\AA - \BB \|^2_2 \leq \epsilon^2 \lambda$. Then, for any vector $\tilde{x}$ such that 
\[
\|\BB \tilde{x} - y \|^2_2 + \lambda\|\tilde{x} \|^2_2 \leq (1+\epsilon')\left(\min_{x}\|\BB x - y \|^2_2 + \lambda\| x \|^2_2  \right)
\]
we have 
\[
\|\AA \tilde{x} - y \|^2_2 +  \lambda\|\tilde{x} \|^2_2 \leq (1+\epsilon')(1+5\epsilon)\left(\min_{x}\|\AA x - y \|^2_2 + \lambda\| x \|^2_2 \right) 
\]
\end{lemma}

Therefore, it suffices to find a rank-$k$ matrix $\BB$ such that $\|\AA - \BB \|^2_2 \leq \epsilon^2 \lambda$. Let $\widetilde{s}_{\lambda}$ be an upper bound on the statistical dimension $s_{\lambda} = \trace{(\AA^2+ \lambda \II)^{-1} \AA^2}$ 
Setting $k = \widetilde{s}_{\lambda}/\epsilon^2$, we can bound $\|\AA - \AA_k \|^2_F$ as follows: 
\begin{equation*}
    \begin{split}
        \frac{\epsilon^2}{\widetilde{s}_{\lambda}}\|\AA - \AA_k \|^2_F & \leq \epsilon^2 \frac{\sum^n_{i=k+1} \lambda_i^2(\AA)}{\sum^n_{i=1} \lambda_i^2(\AA)/(\lambda^2_i(\AA)+\lambda)}\\
        &\leq \epsilon^2 \frac{\sum^n_{i=k+1} \lambda_i^2(\AA)}{\sum^n_{i=k+1} \lambda_i^2(\AA)/(\lambda^2_i(\AA)+\lambda)} \\
        &\leq c \epsilon^2 \lambda
    \end{split}
\end{equation*}

\noindent We can then solve the regression problem $\min_{x}\|\BB x - y \|^2_2+ \lambda \|x\|^2_2$ exactly in time $O(n k^{\omega-1})$ and obtain the following result:

\begin{theorem}[Ridge Regression]
\label{thm:ridge_regression}
Given a PSD matrix $\AA$, a regularization parameter $\lambda$ and an upper bound $\widetilde{s}_{\lambda}$ on the statistical dimension $s_{\lambda} = \trace{ (\AA^2 +\lambda \II)^{-1} \AA^2}$, there exists an algorithm that queries $\widetilde{O}(n \widetilde{s}_{\lambda}/\epsilon^2)$ entries of $\AA$ and with probability $99/100$ outputs $\widehat{x}$ such that for all $y \in \mathbb{R}^d$,
\begin{equation*}
    \|\AA \widehat{x} - y \|^2_2 + \lambda \|\widehat{x}\|^2_2 \leq (1+\epsilon)\left(\min_{x} \|\AA x - y \|^2_2 + \lambda \|x\|^2_2\right)
\end{equation*}
Further, the algorithm runs in $\widetilde{O}(n (\tilde{s}_{\lambda}/\epsilon^2)^{\omega-1})$ time. 
\end{theorem}

\begin{remark}
Observe that we can derive a data structure from our algorithm that preserves the objective cost (up to $1+\epsilon$) for all $x$ and $y$ simultaneously and thus we obtain a coreset for Ridge Regression. 
\end{remark}

To complement the above algorithmic result, we present a new lower bound for coreset constructions for ridge regression, which matches our upper bound in all parameters. At a high level, our hard instance for constant $s_{\lambda}$ consists of $1/\epsilon^2$ blocks of all $1$s, each of size $\epsilon \sqrt{n} \times \epsilon \sqrt{n}$, placed randomly across the matrix. Since any coreset construction must preserve the cost of all $x, y$, we pick pairs $(x,y)$ to be the eigenvectors of $\AA$ (scaled appropriately) and show that in order to preserve the cost of all pairs, the coreset algorithm must find all the blocks, which requires $\Omega(n/\epsilon^2)$ queries to $\AA$. Repeating the above construction $s_{\lambda}$-times suffices to obtain a linear lower bound in terms of $s_{\lambda}$. Formally, 

\begin{theorem}[Coreset Lower Bound for Ridge Regression]
\label{thm:coreset_lb_for_ridge}
Given a PSD matrix $\AA$ and $\epsilon, \lambda>0$ let $s_{\lambda} = \trace{ (\AA^2 +\lambda \II)^{-1} \AA^2}$ denote the statistical dimension of $\AA$. Then, any coreset construction $\mathcal{C}$ that with constant probability, preserves the ridge regression cost up to $(1+\epsilon)$ simultaneously for all $x, y$, must read $\Omega(n s_{\lambda}/\epsilon^2)$ entries in $\AA$.
\end{theorem}

We recall the lower bound instance for low rank approximation of PSD matrices shown by Musco and Woodruff : 

\begin{theorem}[Lower Bound for PSD LRA (\cite{mw17})]
Given an $n \times n $ PSD matrix $\AA$, $\epsilon_0 >0$ and $k_0 \in [n]$, any randomized algorithm that outputs a rank $k_0$ matrix $\BB$ such that with probability at least $9/10$, 
\[
\|\AA - \BB \|^2_F \leq (1+\epsilon_0) \|\AA - \AA_{k_0} \|^2_F
\]
must query $\Omega(nk_0/\epsilon_0)$ entries in $\AA$. 
\end{theorem}

We consider the hard distribution defined by Musco and Woodruff, and show that we can obtain a low rank approximation to this instance with strengthened parameters by using a coreset for ridge regression. 

\begin{definition}[Hard Input Distribution for LRA (\cite{mw17})]
\label{def:hard_dist}
Let $\MM$ be an $n \times n$ matrix and let  $\epsilon_0 >0$, $k_0 \in [n]$. Let $\gamma(n, \epsilon_0 , k_0)$ be a distribution over $\MM$ such that $\S \subset [n]$ is a uniformly random subset of size $n/2$, which is further partitioned into subsets $\S_1, \S_2, \ldots \S_{k_0}$ such that for all $\ell \in [k]$, $\S_{\ell}$ is picked uniformly at random and $|\S_{\ell}| = n/(2k_0)$. For each subset $\S_{\ell}$, let $\AA_{\S_{\ell}}$ denote the principle submatrix of $A$ indexed by the set $\S_{\ell}$. Then, with probability $1/2$, $\AA_{\S_\ell}$ is such that all the diagonal entries are set to $1$ and a uniformly random principle submatrix of $\AA_{\S_\ell}$, indexed by the set $\T_{\ell}$, such that $|\T_{\ell}| = c\sqrt{\epsilon_0 |\S_{\ell}|}$ is set to all $1$s. With the remaining probability $\AA_{\S_\ell}$ is set to the $\II$.
\end{definition}

We show that we can derive a low-rank matrix $\BB$ that satisfies the relative-error guarantee above from a coreset for ridge regression.

\begin{proof}[Proof of Theorem \ref{thm:coreset_lb_for_ridge}]
We show a proof by contradiction, where the high level idea is that a coreset for ridge regression can be used to derive a low-rank approximation to $\AA$, when $\AA$ is picked from $\gamma(n, O(1), s_{\lambda}/\epsilon^2)$ (the hard distribution defined in \ref{def:hard_dist}). First, we observe that with probability at least $99/100$, the input distribution has $\Omega(s_{\lambda}/\epsilon^2)$ blocks that contain a principle submatrix with all $1$s. To see this let $X_{1}, \ldots X_{s_{\lambda}/\epsilon^2 }$ be indicators for the corresponding blocks $\AA_{\S_{\ell}}$ having a principle submatrix of all $1$s. Then, 
\begin{equation}
    \expecf{}{\sum_{ \ell \in [s_{\lambda}/\epsilon^2]} X_{\ell} } = \frac{\epsilon^2 n}{2 s_{\lambda}}
\end{equation}
Since the $X_{\ell}$'s are independent, by a Chernoff bound we have 
\begin{equation}
    \prob{}{ \sum_{ \ell \in [s_{\lambda}/\epsilon^2]} X_{\ell} \leq (1-\delta) \frac{\epsilon^2 n}{2 s_{\lambda}} } \leq \exp\left(-\frac{c\delta \epsilon^2 n }{s_{\lambda}}\right)
\end{equation}
For $n \geq \Omega(s_{\lambda}/\epsilon^2)$, we can bound the above probability by $1/100$. 
We begin by showing that for our input instance, $s_{\lambda} = \Theta( n/\lambda)$ and thus the aforementioned equations differ by $O(\epsilon n/s_{\lambda})$. To see this observe
\begin{equation}
    s_{\lambda} = \sum_{i\in [n]} \frac{ \sigma_i^2(\AA)}{\sigma_i^2(\AA) + \lambda }
\end{equation}
Then, there are $s_\lambda/2\epsilon^2$ large eigenvalues, each of magnitude $\epsilon \sqrt{n/s_{\lambda}}$ and thus the total contribution is 
\[
\frac{s_{\lambda}}{2\epsilon^2} \cdot \frac{ (\epsilon^2 n/s_{\lambda}) }{(\epsilon^2 n/s_{\lambda}) + \lambda } = \frac{n}{\epsilon^2n/s_{\lambda} + \lambda}
\]
The remaining eigenvalues simply contribution $1/(1+\lambda)$ to the sum and since there are at most $n$ of them, the total contribution is $n/(1+\lambda)$. Therefore, we can conclude $s_{\lambda} = \Theta(n/\lambda)$.

For a block in $\AA$ indexed by $\ell$, let $\hat{x}_{\ell}$ be the eigenvector supported on indices in $\S_{\ell}$ and let $\hat{y}_{\ell} = \sqrt{n /s_{\lambda}} \hat{x}_{\ell}$. For non-identity blocks, $\AA \hat{x}_{\ell} =  |\T_{\ell} |\hat{x}_{\ell} = \sqrt{\epsilon^2 n /s_{\lambda}} \hat{x}_{\ell}$ and the regression cost is 
\begin{equation}
    \| (1 - \epsilon)\sqrt{n/s_{\lambda}} \hat{x}_{\ell} \|^2_2 + \lambda = (1-2\epsilon) n/s_{\lambda} + cn/s_{\lambda}
\end{equation}
When the block indexed by $\ell$ is the identity block, we get $\AA \hat{x}_{\ell} = \hat{x}$ and the regression cost is 
\begin{equation}
    \| (\sqrt{n/s_{\lambda}} -1) \hat{x}_{\ell} \|^2_2 + \lambda = (n/s_{\lambda} + 1 - 2\sqrt{n/s_{\lambda}}) + cn/s_{\lambda}
\end{equation}
Instead consider a vector that intersects an eigenvector $\tilde{x}_{\ell}$ on a $(1-\gamma)$-fraction of the support and the rest is arbitrary. Then, when an an all $1$s block exists, $\AA \tilde{x}_{\ell} \geq  (1-\gamma)^2 |\T_{\ell} |\hat{x}_{\ell} =  (1-\gamma)^2 \sqrt{\epsilon^2 n /s_{\lambda}} \hat{x}_{\ell}$ and thus the regression cost is at most 
\[
(1 - \epsilon(1-2\gamma))^2 n/s_{\lambda} + cn/s_\lambda 
\]
Further, when the block is simply the identity, a similar calculation shows that the regression cost is at least $(1- 2\epsilon\gamma)^2n/s_{\lambda} + cn/s_{\lambda} $. 
Therefore, the ridge regression cost determines the existence of a $(1-2\gamma)$-fraction of an all $1$s principle submatrix even when $\hat{x}_{\ell}$ intersects with an eigenvector on a $(1-\gamma)$-fraction of coordinates. 


Consider a coreset $\mathcal{C}$ for the above instance. Since this coreset preserves the ridge regression objective upto a $(1+\epsilon/1000)$ factor for all $x, y$, as per our above discussion we can query the coreset on the tuples $(\hat{x}_{\ell}, \hat{y}_{\ell})$, which represent the eigenvectors of each block, to determine if a block contains a principle submatrix with all $1$s. However, a priori we do not know the support of the eigenvector within each block $\AA_{\S_{\ell}}$. 

Instead we query the coreset on all possible supports and show we can determine the right one as follows: let $\tilde{x}$ be supported on a set that intersects with a principle submatrix of all $1$s on at most a $\gamma$ fraction. Observe that $\AA \tilde{x} \leq \frac{1}{\gamma}\left( \gamma^2 \epsilon \tilde{x} \right)$ and thus the ridge regression cost can be lower bounded as follows: 
\begin{equation}
(1 - \epsilon\gamma)^2n/s_{\lambda} + cn/s_{\lambda}
\end{equation}
We therefore take the set of all vectors on which the coreset cost is less than the above cost and let the resulting list be $\mathcal{L}$. Note, this list must include the eigenvectors and further, only includes vectors which intersect an all $1$'s submatrix on a $1-\gamma$-fraction. Therefore, picking a set of $\epsilon^2 n/s_{\lambda}$ vectors that have maximum support suffices. 

Since we detect a $(1-\gamma)$-fraction of all principle submatrices in $\AA$, it follows that we can output a $1+c'$-approximate low-rank approximation for $\AA$, for a fixed small constant $c'$. To see this, observe that the optimal low-rank approximation to $\AA$ is given by the matrix that selects all the principle submatrices with all $1$s and thus $\|\AA - \AA_{k_0}\|^2_F = n - k_0 = n - s_{\lambda}/\epsilon^2$. Further, our approximation to $\AA_k$, denoted by $\BB$, matches $\AA_k$ on a $(1-\gamma)$-fraction of each principle submatrix of all $1$s and thus we match $\AA_k$ on these entries. Subsequently we bound the additional cost that $\BB$ incurs which $\AA_k$ does not. This includes entries that are $1$ in $\AA_k$ and $0$ in $\BB$ and vice versa. 

To bound the cost of the entries that exist in $\AA_k$ but do not exist in $\BB$, observe on each principle submatrix, $\BB$ and $\AA$ intersect in at least $(1-\gamma)^2$-fraction of the entries and thus the remaining entries are at most $(1-(1-\gamma)^2)   \cdot\frac{\epsilon^2 n}{c s_{\lambda}}\cdot \frac{s_{\lambda}}{\epsilon^2} = 4\gamma n/c$, since the size of each block is $\frac{\epsilon^2 n}{c s_{\lambda}}$ and the number of blocks are at most $\frac{s_{\lambda}}{\epsilon^2}$.
Finally, observe that since we do not pick exact eigenvectors we can have non-zero off diagonal entries in $\BB$ that do not exist in $\AA_k$. However, we have at most $\gamma$-fraction of the support on each indicator vector remaining and contributing to two rectangular blocks of $1$s, each of size $\gamma \cdot\frac{\epsilon^2 n}{c s_{\lambda}}\cdot \frac{s_{\lambda}}{\epsilon^2} = \gamma n/c$. Therefore, the additional non-zero entries in $\BB$ that do not appear in $\AA$ are $2\gamma n/c$ in number.

Therefore, the overall cost $\| \AA - \BB \|^2_F \leq   (1+6\gamma/c) n$. By setting the constants $\gamma$ and $c$ and observing that $s_{\lambda}/\epsilon^2 \ll n$, we obtain a $(1 + c')$-low-rank approximation to $\AA$, for any arbitrary small constant $c'$. 
Therefore, our reduction suffices to solve the hard instance above and a lower bound of $\Omega(n k_0/\epsilon_0) = \Omega(n s_{\lambda}/\epsilon^2)$ queries follows.
\end{proof}
\section{Robust Low-Rank Approximation}

One drawback of relative-error guarantees is that the corresponding algorithms cannot tolerate any amount of noise. Therefore, we introduce a robust model for low-rank approximation by relaxing the requirement from relative-error guarantees to additive-error guarantees.  
In the robustness model we consider, we begin with an $n \times n$ PSD matrix $\AA$. An adversary is then allowed to arbitrarily corrupt $\AA$ by adding a corruption matrix $\NN$ such that the corruption in each row is an fixed constant times the $\ell^2_2$ row norm of the row and the total corruption is an $\eta$-fraction of squared Frobenius norm of $\AA$. While the adversary may corrupt any number of entries of $\AA$, the norm of the corruption matrix is bounded and the algorithm has query access to $\AA + \NN$. We parameterize our lower bound and algorithms by the largest ratio between a diagonal entry of $\AA$ and $\AA+\NN$, denoted by $\phi_{\max} = \max_{j \in [n]} \AA_{j,j} / |(\AA+\NN)_{j,j}|$. This captures the intuition that the diagonal is crucial for sublinear time low rank approximation and the sample complexity degrades as we corrupt larger diagonals entries. 

\subsection{Lower Bound for Robust PSD Low-Rank Approximation}
In this subsection, we show a query lower bound of $\Omega(\eta^2n^2k/\epsilon^2) = \Omega(\phi^2_{\max}nk/\epsilon)$ for any algorithm that outputs a low rank approximation up to additive-error $(\epsilon+\eta)\|\AA \|^2_F$. Note, obtaining error smaller than $\eta\|\AA\|^2_F$ is information-theoretically impossible and reflected in the query lower bound.   


Our lower bound holds for randomized algorithms, and uses Yao's minimax principle \cite{yao_minimax}. The overall strategy is to demonstrate a lower bound for deterministic algorithms on a carefully chosen input distribution. We construct our input distributions as follows: let $\AA \in\mathbb{R}^{n \times n}$ be a block diagonal matrix with such that $\BB_1$ is $5\epsilon/\eta \times 5\epsilon/\eta$ randomly positioned, non-contiguous block with all entries $\sqrt{n\eta^2/5\epsilon}$ and $\BB_2$ is the identity matrix on the remaining indices. $\AA$ is clearly a PSD matrix since each principle submatrix is PSD. Observe $\|\AA \|^2_F = (25\epsilon/\eta^2)\cdot(n\eta^2/5\epsilon) + (n- 5\epsilon/\eta) =  (1+5\epsilon)n - o(n)$.  Further, the dense block $\BB_1$ contributes a total squared Frobenius norm of at least $4\epsilon n$ and the diagonal entry contributes an $\eta/\epsilon$ fraction of each row. Since $\epsilon\geq \eta$, the diagonal contributes at most the entire $\ell^2_2$ norm. 
The corresponding diagonal of $\BB_{1}$ also has $\ell^2_2$ squared norm $ 5\epsilon/\eta \cdot n\eta^2/5\epsilon = \eta n$. 

At a high level, the adversary can then corrupt the diagonal and set each diagonal entry to be $1$, making it hard for the algorithm to find rows corresponding to $\BB_1$. We show that any $\epsilon$-additive-error low-rank approximation must detect at least one entry in $\BB_1$ to adaptively sample the corresponding row and column,  but the diagonals no longer provide any useful information. Thus any algorithm must query most entries in $\AA$. Further, in our construction, note $\phi_{\max} = \sqrt{n\eta^2/5\epsilon}$. 

We first describe intuitively why a low rank approximation needs to recover many rows from the block $\BB_1$. Since $\AA$ has this block structure, the best rank-$1$ approximation satisfies $\|\AA - \AA_k\|^2_F = n - |\BB_1|$. Therefore, assuming the cardinality of $\BB_1$ is negligible, in order to obtain an overall error bound of $\epsilon\|\AA\|^2_F \geq \epsilon n$, the algorithm must find a constant fraction of off-diagonal entries in $\BB_1$. This is because $\BB_1$ contributes at least $5\epsilon n$ norm. However, since the diagonals no longer convey any information about the off-diagonals, and the block $\BB_1$ is placed on a random subset of indices,  any deterministic algorithm must read arbitrary off-diagonal entries until it finds a non-zero entry. Since there are only $25\epsilon^2/\eta^2$ non-zeros in $\BB_1$, to find one in expectation (over the input distribution) requires sampling $\epsilon^2 n^2/\eta^2$ entries. While the above serves well as intuition, a rigorous proof requires many additional steps. We begin by defining a distribution over the input matrices :

\begin{definition}
Given $n \in \mathbb{N}, \epsilon >\eta>0$, let $\mathcal{S} \subset [n]$ be a uniformly random subset of size $\lceil 5\epsilon/\eta\rceil$. Let  $\mu(n,\epsilon, \eta)$ be a distribution over matrices $\MM \in \mathbb{R}^{n \times n}$ such that $\forall i \in[n]$, $\MM_{i,i} = 1$ and $\forall i,j \in \mathcal{S}$, $\MM_{i,j} = \sqrt{\eta^2n/5\epsilon}$. All remaining entries in $\MM$ are $0$. 
\end{definition}

Next, we show that any $\MM$ sampled from $\mu(n,\epsilon, \eta)$ can be decomposed into $\AA+\NN$ such that $\AA$ is PSD and $\|\NN\|^2_F \leq \eta \|\AA \|^2_F$. To see this, let $\NN$ be a diagonal matrix such that for all $i\in \mathcal{S}$, $\NN_{i,i} = -\sqrt{\eta^2n/5\epsilon}$ and let $\AA = \MM - \NN$. We give an algebraic proof that $\AA$ is PSD, but $\AA$ can also be decomposed into a rank-$1$ block of all $\sqrt{\eta^2n/5\epsilon}$-s corresponding to all $i,j\in\mathcal{S}$ and identity on the remaining indices. For all $x \in \mathbb{R}^n$,
\begin{equation}
\begin{split}
    x^T\AA x &= \sum_{i,j \in \mathcal{S}} \AA_{i,j} x_i x_j + \sum_{i,j \notin \mathcal{S}} \AA_{i,j} x_i x_j \\
    & = \sqrt{\eta^2 n/5\epsilon} \sum_{i,j \in \mathcal{S}} x_i x_j  + \sum_{i \notin \mathcal{S}} x_i^2 \\
     & = \sqrt{\eta^2 n/5\epsilon}\left(\sum_{i \in \mathcal{S}}x_i\right)^2 + \sum_{i \notin \mathcal{S}} x_i^2 \\
     & \geq 0
\end{split}
\end{equation}
and thus $\AA$ is PSD. Further, $\|\AA \|^2_F =(25\epsilon^2/\eta^2)\cdot(n\eta^2/5\epsilon) + (n- 5\epsilon/\eta) =  (1+5\epsilon)n - 5\epsilon/\eta$.
Then, $\|\NN\|^2_F = 5\epsilon/\eta \cdot \eta^2 n/\epsilon = \eta \| \AA\|^2_F$, as desired. Intuitively, we show that if $\BB$ is a rank-$k$ matrix that is a good low-rank approximation for $\MM$ sampled from $\mu$, then it cannot be a good low-rank approximation for $\II$. 
To this end, we consider a distribution where $\MM$ is drawn from $\mu(n, \epsilon, \eta)$ with probability $1/2$ and is $\II_{n \times n}$ with probability $1/2$. 

\begin{definition}(Hard Distribution)
Given $n\in \mathbb{N}, \epsilon > \eta > 0$, let $\nu(n,\epsilon,\eta)$ be a distribution over $\MM \in \mathbb{R}^{n \times n}$ such that with probability $1/2$, $\MM$ is sampled from $\mu(n,\epsilon,\eta)$ and with probability $1/2$, $\MM = \II_{n\times n}$.
\end{definition}

We now show that a low-rank approximation to $\MM$ can be used as a certificate to separate the mixture $\nu(n,\epsilon,\eta)$ since it can distinguish between the input being identity or far from it. Thus if the distributions are close in a statistical sense, any algorithm to distinguish between the two would require querying many entries in $\MM$. Formally, 

\begin{lemma}(LRA as a Distinguisher.)
\label{lem:distinguish}
Let $\MM$ be a matrix drawn from $\mu(n,\epsilon,\eta)$ and let $\BB$ be a rank-$k$ matrix that is the candidate low-rank approximation to $\MM$ such that $\|\MM -\BB \|^2_F \leq \epsilon n$. Then, $\|\MM - \II \|^2_F > 1.1 \epsilon n$.  
\end{lemma}
\begin{proof}
Since $\|\MM -\BB \|^2_F \leq \epsilon n$, $\BB$ must have at least $4\epsilon n$ mass on the off-diagonal entries of $\MM$. So, $\BB$ must have at least $10\epsilon^2/\eta^2$ non-zero off-diagonal entries. Therefore, it must have at least $5\epsilon^2/\eta^2$ entries with squared mass $\epsilon n/2$. To see why, assume there is a subset of at least $12\epsilon^2/\eta^2$ entries, each being at most  $\sqrt{ n \eta^2 /10\epsilon}$. Restricted to only these entries, the squared Frobenius norm difference between $\MM$ and $\BB$ is already at least $1.2\epsilon n$, contradicting our assumption. 
Given that there exists a subset of $5\epsilon^2/\eta^2$ off-diagonal entries having squared mass $1.2 \epsilon n$, $\| \BB - \II\|^2_F > 1.2\epsilon n$, and thus $\BB$ is not an additive error low-rank approximation for $\II$. 
\end{proof}

\begin{theorem}(Lower bound for PSD Matrices.)
\label{thm:robust_lb}
Let $\AA$ be a PSD matrix, $k\in \mathbb{Z}$ and $\epsilon>0$ be any constant. Let $\NN$ be an arbitrary matrix such that $\|\NN \|^2_F\leq \eta \|\AA \|^2_F$. Any randomized algorithm $\mathcal{A}$ that only has query access to $\AA + \NN$, with probability at least $2/3$, computes a rank-$k$ matrix $\BB$ such that 
\begin{equation*}
    \|\AA - \BB \|^2_F \leq \| \AA - \AA_k \|^2_F + \epsilon\|\AA \|^2_F
\end{equation*}
must read $\Omega\left( \phi_{\max}^2 nk/\epsilon \right)$ entries of $\AA+\NN$ on some input, possibly adaptively, in expectation.
\end{theorem}
\begin{proof}
Let Algorithm $\mathcal{A}$ be a deterministic algorithm that outputs a rank-$k$ matrix $\BB$ such that it is an additive-error low-rank approximation $\MM$. Let $T \subset [n^2]$ be the subset of entries read by $\mathcal{A}$. Let $L(\mu)$ denote the distribution of $T$ conditioned on $\MM \sim \mu(n,\epsilon,\eta)$ and $L(i)$ be the distribution of $T$ conditioned on $\MM=\II$. By Lemma \ref{lem:distinguish}, since the output of $\mathcal{A}$ can be used to distinguish between the two distributions, it is well-known that the success probability over the randomness in $T$ is at most $1/2 + D_{TV}(L(\mu), L(i))/2$ \cite{bar2002complexity}. Since we assume $\mathcal{A}$ succeeds with probability at least
$2/3$, 
\begin{equation}
\label{eqn:tv_lb}
    D_{TV}(L(\mu), L(i)) \geq 1/3
\end{equation}
It remains to upper bound $D_{TV}$ in terms of $|T|$. Recall, $\mathcal{S}$ is the random set of indices where $\mu(n,\epsilon,\eta)$ is non-zero. Let $\widehat{\mathcal{S}}$ be the subset of $\mathcal{S}$ restricted to the off-diagonal entries of $\MM$. When $\MM \sim \mu(n,\epsilon,\eta)$, $\forall i,j \in \widetilde{\mathcal{S}}$, $\MM_{i,j}$ is non-zero and when $\MM = \II$, the same entries are $0$. Observe, for all $i,j \notin \widetilde{\mathcal{S}}$, $\MM_{i,j}$ are fixed. Further, $\mathcal{S}$ is a uniform subset of $[n]$. Therefore,
\begin{equation}
    \Pr\left[ (i,j) \in T \mid (i,j) \in \mathcal{S} \right]= \frac{|T|\epsilon^2}{\eta^2 n^2}
\end{equation}
Then, with probability at least $1- |T|\epsilon^2/\eta^2 n^2$, $\mathcal{A}$ queries the same entries for both $L(\mu)$ and $L(i)$. 
Therefore $$D_{TV}(L(\mu), L(i)) \leq |T|\epsilon^2/\eta^2 n^2.$$ 
Combined with Equation \ref{eqn:tv_lb}, if $\mathcal{A}$ succeeds with probability at least $2/3$, 
$|T|\epsilon^2/\eta^2 n^2 \geq 1/3$
and thus $|T| = \Omega(\eta^2 n^2/\epsilon^2)$. Given that any deterministic algorithm must query  $\Omega(\eta^2 n^2/\epsilon^2) = \Omega(\phi_{\max}^2 n/\epsilon)$ entries for $\nu(n,\epsilon,\eta)$, to now obtain a linear dependence on the rank $k$, we can use the standard approach of creating $k$ disjoint copies of the block $\BB_1$ in the hard distribution, as shown in \cite{mw17}.
The theorem follows from Yao's minimax principle. 
\end{proof}

\subsection{Robust Sublinear Low-Rank Approximation Algorithms}

In this subsection, we provide a robust algorithm for the model discussed above. We parameterize our algorithms and lower bound by the largest ratio between a diagonal entry of $\AA$ and $\AA+\NN$, denoted by $\phi_{\max} = \max_{j \in [n]} \AA_{j,j} / |(\AA+\NN)_{j,j}|$. In addition, we provide robust \textit{PCP} constructions, by introducing a new sampling procedure to construct projection-cost preserving sketches. Our sampling procedure is straightforward: we sample each column proportional to the diagonal entry in that column. This sampling requires $n$ queries to the matrix $\AA$ to obtain an additive-error projection cost preservation guarantee.  
Further, for the special case of correlation matrices, we can uniformly sample columns of $\AA$ to obtain a smaller matrix such that all rank $k$ projections in the column and row space are preserved.

For our algorithms, we assume we know $\phi_{\max}$. In practice, this assumption may not hold, but we can query as many entries in $\AA + \NN$ as our budget allows, given that correctness holds only when the queries are at least $\widetilde{O}(\phi^2_{\max} nk/\epsilon)$. Since we read the diagonals of $\AA+\NN$ and we know $\phi_{\max}$, we can obtain an upper bound on $\AA_{i,i}$ and $\AA_{j,j}$.  Therefore, whenever we query an off-diagonal entry in $\AA + \NN$, we can truncate it to $\phi_{\max}\sqrt{|(\AA+\NN)_{i,i}|\cdot|(\AA+\NN)_{j,j}|}$ without increasing the corruption in our input.

\vspace{0.1in}
\textbf{Robust Projection-Cost Preserving Sketches.}
Here, we show that diagonal sampling is a robust sampling procedure to create \textit{projection-cost preserving sketches}.
We begin by relating the $\ell^2_2$ row (or column) norms of a PSD matrix to it's spectral norm. Let $\AA$ be a PSD matrix and let $\UU \Sig \UU^T$ be the \textsf{SVD} for $
\AA$. 

\begin{lemma}
\label{lem:row_norm_spectral_norm}
Given an $n \times n$ PSD matrix $\AA$, for all $i \in [n]$, $\|\AA_{i,*}\|^2_2 \leq  \| \AA \|_2 \cdot \AA_{i,i}$.
\end{lemma}
\begin{proof}
Observe, $\AA_{i,*} = \UU_i \Sig \UU^T$ and $\AA_{i,i} = (\UU_{i,*} \Sig \UU^T)_i = \sum^n_{j=1} \sigma_j(\AA) \UU^2_{i,j}$. Then,
\begin{equation*}
    \begin{split}
        \|\AA_{i,*}\|^2_2 =  \AA_{i,*}\AA^T_{i,*}  = \UU_{i,*} \Sig \UU^T \UU \Sig \UU_{i,*}^T
        & = \sum_{j=1}^n \sigma^2_j(\AA) \UU_{i,j}^2 \\
        & \leq \| \AA \|_2 \sum_{j=1}^n \sigma_j \UU_{i,j}^2 \\
        &= \| \AA \|_2\cdot \AA_{i,i}
    \end{split}
\end{equation*}
\end{proof}

\noindent An immediate consequence of Lemma \ref{lem:row_norm_spectral_norm} is that the $\ell^2_2$ norm of a row or column of a PSD matrix is at most $\frac{\|\AA\|_F}{\AA_{i,i}}$. Note, this precludes matrices where most of the mass in concentrated on a small number of rows or columns.   
Recall, we observe as input the matrix $\AA + \NN$ and our goal is to obtain a PCP for this in sublinear time and queries.  

Musco and Musco \cite{musco2017recursive} describe how to approximately compute the ridge leverage scores of $\AA^{\frac{1}{2}}$ (if $\AA$ is PSD) using a Nystrom approximation. \cite{mw17} use this method to compute the ridge leverage scores of $\AA^{\frac{1}{2}}$ with $O(nk)$ queries, where $\AA = \AA^{\frac{1}{2}}\cdot \AA^{\frac{1}{2}}$. However, these approaches do not apply when we perturb the input and it may no longer be PSD. Therefore, the best known construction by Cohen et. al.  \cite{cohenmm17} would require $\Omega(nnz(\AA))$ time to compute approximate ridge-leverage scores of $\AA$. Note, this does not use the structure that $\AA$ has. 

In contrast, we show that sampling columns proporitional to the diagonal entries suffices to obtain a PCP. Note, 
we only need to query the diagonal of $\AA$ to compute the distribution over columns exactly.
The main technical challenge here is to obtain the correct dependence on $n$ and $k$ and account for the perturbation to the input, given that our sampling probabilities are straightforward to compute and do not rely on spectral properties of $\AA+\NN$.  Note, the following is a structural result and while we do not know $\AA$, we can still show the following :

\begin{theorem}(Robust Spectral Bound.)\label{thm:pcp1}
Let $\AA$ be an $n \times n$ PSD matrix and $\NN$ be an arbitrary matrix such that $\|\NN \|^2_F \leq \eta \|\AA \|^2_F$ and for all $j \in [n]$, $ \|\NN_{*,j} \|^2_2 \leq c \|\AA_{*,j}\|^2_2$, for any fixed constant $c$. Let $\phi_{\max} = \max_j \AA_{j,j}/(\AA+\NN)_{j,j} $ and let $q = \{q_1, q_2 \ldots q_n\}$ be a distribution over the columns of $\AA + \NN$ such that for all $j$, $q_j = (\AA+\NN)_{j,j}/\trace{\AA + \NN}$ and let $t =  O\left(\phi_{\max}\sqrt{n}k^2\log(n/\delta)/\epsilon^2\right)$. Then, construct a sampling matrix $\TT$ that samples $t$ columns of $\AA+\NN$ such that it samples column$(\AA+\NN)_{*,j}$ with probability $q_j$ and scales it by $1/\sqrt{t q_j}$. With probability at least $1-\delta$, for any rank-$k$ orthogonal projection $\XX$,
\[
\AA\AA^T  - \left(\frac{\epsilon}{k} \right)\| \AA\|^2_F \II  \preceq \AA\TT (\AA\TT)^T \preceq \AA\AA^T +  \left(\frac{\epsilon}{k}\right)\| \AA\|^2_F\II 
\]
\end{theorem}
\begin{proof}
First, we note that we cannot explicitly compute $\AA\TT$, but we can show that the sampling probabilities we have access to result in a PCP for $\AA$.  
Let $\Y = \AA\TT (\AA\TT)^T - \AA\AA^T$. For notational convenience let $\AA_j = \AA_{*,j}$. We can then write $\Y = \sum_{j \in [t]} \left(\CC_{*,j} \CC_{*,j}^T -\frac{1}{t}\AA\AA^T\right) = \sum_{j\in [t]} \XX_j$, where $\XX_j = \frac{1}{t}(\frac{1}{q_j}\AA_{j}\AA^T_{j} - \AA\AA^T )$ with probability $q_j$. We observe that $\expec{}{\XX_j} = \expec{}{\CC_{*,j} \CC_{*,j}^T -\frac{1}{t}\AA\AA^T}=  0$, and therefore, $\expec{}{\Y} =0$. Next, we bound the operator norm of $\Y$. To this end, we use the Matrix Bernstein inequality, which in turn requires a bound on the operator norm of $\XX_j$ and variance of $\Y$. 
Recall, 
\begin{equation}
\label{eqn:spec_x_j}
    \begin{split}
        \| \XX_j \|_2 & = \left\| \frac{1}{tq_j}\AA_{j}\AA_{j} - \frac{1}{t}\AA\AA^T \right\|_2 \\
        &\leq  \frac{\trace{\AA+\NN}}{t (\AA+\NN)_{j,j}} \|\AA_{j}\|^2_2 + \frac{1}{t} \| \AA\|^2_2\\
        & \leq \frac{2\trace{\AA}+|\trace{\NN}|}{t(\AA+\NN)_{j,j}}\left((1+\eta)\|\AA\|_2\AA_{j,j} \right) \\
        & \leq \frac{c(\trace{\AA}+|\trace{\NN}|) \|\AA\|_2 \phi_{\max}}{t}  \\
        & \leq \frac{c \phi_{\max}\sqrt{n}\|\AA \|_F \|\AA\|_2 }{t} 
    \end{split}
\end{equation}
where we use triangle inequality for operator norm to obtain the first inequality, triangle inequality up to a factor of $2$ for $\ell^2_2$ norms for the second inequality,  $\|\NN_{j}\|^2_2 \leq \eta \|\AA_{*,j}\|^2_2$ and  $\|\AA_j \|^2_2\leq \| \AA \|_2 \cdot \AA_{j,j}$ (from Lemma \ref{lem:row_norm_spectral_norm}) for the third inequality and definition of $\phi_{\max}$ and $\eta = O(1)$ for the fourth. Finally, we relate the trace of $\AA$ and $\NN$ to their respective Frobenius norm using Cauchy-Schwarz:  
\begin{equation*}
    \begin{split}
        \trace{\AA} = \sum^n_{i=1}\sigma_i(\AA) \leq \sqrt{\sum^n_{i=1} \sigma_i^2(\AA) \cdot n} = \sqrt{ n \|\AA\|^2_F } 
    \end{split}
\end{equation*}
and
\begin{equation*}
    \begin{split}
        |\trace{\NN}| = \left|\sum^n_{i=1}\sigma_i(\AA)\right| \leq \sqrt{\sum^n_{i=1} \sigma_i^2(\NN) \cdot n} = \sqrt{ n \|\NN\|^2_F } \leq \sqrt{n\eta} \|\AA\|_F
    \end{split}
\end{equation*}
where the last inequality follows from $\|\NN \|_F \leq \sqrt{\eta}\|\AA\|_F$.
Next, we bound $\var{\Y}\leq \expecf{}{\Y^2}$.

\begin{equation}
    \begin{split}
        \expecf{}{\Y^2}  & =  t\expecf{}{\left((\AA\TT)_{*,j} (\AA\TT)_{*,j}^T -\frac{1}{t}\AA\AA^T\right)^2} \\
        & = t\expecf{}{ \left((\AA\TT)_{*,j} (\AA\TT)_{*,j}^T\right)^2 +\frac{1}{t^2}\left(\AA\AA^T\right)^2 - \frac{2}{t}(\AA\TT)_{*,j} (\AA\TT)_{*,j}^T \AA\AA^T } \\
        & = \frac{1}{t}\left( \sum_{j \in [n]} \frac{(\AA_{j}\AA^T_{j})^2}{q_j} + (\AA\AA^T)^2 - \sum_{j \in [n]} 2 \AA_{j}\AA^T_{j} \AA\AA^T \right) \\
        & \preceq \frac{\trace{\AA+\NN}}{t \AA_{j,j}}\left( \sum_{j \in [n]}  (\AA_{j}\AA^T_{j})^2 \right) \\
        & \preceq \frac{c\phi_{\max}\sqrt{n}\|\AA\|_F \|\AA \|_2 }{t} \AA \AA^T \\
    \end{split}
\end{equation}
where we use linearity of expectation, $(\AA\AA^T)^2 \succeq 0$ and $\|\AA_{j,*}\|^2_2\leq \| \AA\|_2\cdot \AA_{j,j}$.
Applying the Matrix Bernstein inequality, 
\begin{equation*}
    \begin{split}
        \prob{}{\| \Y \|_2 \geq  \epsilon \| \AA \|^2_F } & \leq 2n \exp{\left(- \frac{ \epsilon^2 \| \AA \|^4_F }{  \frac{c\phi_{\max}\sqrt{n}\| \AA \|_F \|\AA \|^3_2 }{t}  +  \frac{2\phi_{\max}\sqrt{n}\|\AA\|_F \|\AA\|_2 \left( \epsilon \|\AA \|^2_F \right)}{3t} }\right)} \\
        & \leq 2n \exp{\left(-\frac{\epsilon^2t }{c\phi_{\max}\sqrt{n}\| \AA \|_F \|\AA \|_2 }\left(  \frac{  \|\AA \|^4_F }{\|\AA \|^2_2 + \epsilon \|\AA \|^2_F  }\right) \right)} \\
        & \leq 2n \exp{\left(- \frac{ \epsilon^2 t  }{c'\phi_{\max}\sqrt{n} } \right)} \\
        &\leq \delta/2
    \end{split}
\end{equation*}
where the last inequality follows from setting $t = O(\phi_{\max}\sqrt{n}\log(n/\delta)/\epsilon^2)$. To yield the claim, we set $\epsilon = \epsilon/k$. 
\end{proof}

\noindent We use the above spectral bound to show that sampling proportional to diagonal entries preserves the projection cost of the columns of $\AA$ on to any $k$-dimensional subspace up to an additive $(\epsilon +\sqrt{\eta}) \|\AA\|^2_F$.

\begin{theorem}(Column Projection-Cost Preservation.)
\label{thm:pcp_main_col}
Given $\AA + \NN$, where $\AA$ is an $n \times n$ PSD matrix and $\NN$ is an arbitrary noise matrix as defined above, $k \in \mathbb{Z}$ and $\epsilon>\eta >0$, let $q = \{q_1, q_2 \ldots q_n\}$ be a probability distribution over the columns of $\AA+\NN$ such that $q_j = \frac{(\AA+\NN)_{j,j}}{\trace{\AA+\NN}}$.  Let $t =  O\left(\phi_{\max}\sqrt{n}k^2\log(\frac{n}{\delta})/\epsilon^2\right)$. Then, construct $\CC$ using $t$ columns of $\AA+\NN$ and set each one to $\frac{(\AA+\NN)_{*,j}}{\sqrt{t q_j}}$ with probability $q_j$. With probability at least $1-c$, for any rank-$k$ orthogonal projection $\XX$,
\[
\| \CC - \XX\CC \|^2_F = \| \AA - \XX\AA \|^2_F \pm (\epsilon+ \sqrt{\eta})\| \AA \|^2_F
\]
for a fixed constant $c$. 
\end{theorem}

\begin{proof}
Here, the matrix $\CC$ is actually a matrix we can compute. Observe that we can relate $\CC$ to the sampling matrix $\TT$ as defined in Theorem \ref{thm:pcp1} as $\CC = (\AA + \NN)\TT$. 
We follow the proof strategy of the relative error guarantees in \cite{cohenmm17} and additive error guarantees in \cite{bakshi2018sublinear} but note, our spectral bounds from Theorem \ref{thm:pcp1} apply to matrices that we do not actually compute. 
Observe, $\| \AA - \XX\AA \|^2_F = \trace{ (\mathbb{I}-\XX) \AA \AA^T (\mathbb{I}-\XX )}$. Then, 
\begin{equation}
    \begin{split}
        \trace{ (\mathbb{I}-\XX) \AA \AA^T (\mathbb{I}-\XX)} & = \trace{  \AA \AA^T} + \trace{ \XX\AA \AA^T \XX} - \trace{\AA \AA^T \XX} - \trace{ \XX\AA \AA^T} \\
        & = \trace{  \AA \AA^T} + \trace{ \XX\AA \AA^T \XX} - \trace{\AA \AA^T \XX\XX} - \trace{ \XX\XX\AA \AA^T} \\
        & = \trace{  \AA \AA^T} + \trace{ \XX\AA \AA^T \XX} - \trace{\XX\AA \AA^T \XX} - \trace{ \XX\AA \AA^T \XX} \\
        & = \trace{  \AA \AA^T} - \trace{ \XX\AA \AA^T \XX}
    \end{split}
\end{equation}
where we used the fact that for any projection matrix $\XX= \XX^2$ in addition to the cyclic property of the trace.
Similarly, 
\begin{equation}
    \begin{split}
        \| \CC - \XX\CC \|^2_F  = \trace{(\mathbb{I}-\XX) \CC \CC^T (\mathbb{I}-\XX)} 
        & = \trace{  \CC \CC^T} - \trace{ \XX\CC \CC^T \XX} 
    \end{split}
\end{equation}

\noindent We first relate $\trace{  \AA \AA^T}$ and $\trace{  \CC \CC^T}$. Recall, $$\expecf{}{\trace{\CC \CC^T}} = \expecf{}{\|\CC\|^2_F} = \|\AA +\NN \|^2_F \leq \trace{\AA\AA^T} + 2\sqrt{\eta}\| \AA \|^2_F  $$
Using a scalar Chernoff bound, we show that with probability at least $1-1/\textrm{poly}(n)$, $\|\CC\|^2_F = (1\pm \epsilon)\|\AA+\NN\|^2_F$. This is equivalent to $| \|\CC \|^2_F - \|\AA +\NN\|^2_F | \leq \epsilon \| \AA +\NN\|^2_F$. Observe, for all $j \in [t]$, $\CC_{*,j} = \frac{1}{\sqrt{q_{j'} t}} (\AA+\NN)_{*,j'}$ for some $j'\in [n]$. Then, 
\begin{equation}
\begin{split}
    \|\CC_{*,j} \|^2_2  = \frac{1}{q_{j'} t} \| (\AA+\NN)_{*,j'} \|^2_2 
     & = \frac{\trace{\AA+\NN} \epsilon^2}{ \phi_{\max}\sqrt{n}k \log(n/\delta) (\AA+\NN)_{j',j'}} \| (\AA+\NN)_{*,j'} \|^2_2 \\
     & \leq \frac{c\sqrt{n}\|\AA\|_F \epsilon^2}{ \sqrt{n} \log(n/\delta)} \| \AA\|_2 \\
     & \leq \frac{c\epsilon^2}{ k \log(n/\delta)} \| \AA+\NN \|^2_F
\end{split}
\end{equation}
where we use $\trace{\AA} \leq \sqrt{n} \|\AA\|_F$, $\trace{\NN} \leq \sqrt{ \eta n }\|\AA\|_F$ and $t = O(\phi_{\max}\sqrt{n}k\log(n/\delta)/\epsilon^2)$.
Therefore, $\frac{ k \log(n/\delta)}{\epsilon^2 \|\AA \|^2_F} \|\CC_{*,j} \|^2_2 \in [0,1]$. By a Chernoff bound, 

\begin{equation}
    \begin{split}
        \prob{}{\|\CC \|^2_F \geq  (1 +2\epsilon) \|\AA+\NN \|^2_F } & = \prob{}{ \frac{k \log(n/\delta)}{\epsilon^2 \|\AA+\NN \|^2_F} \|\CC \|^2_F \geq \frac{k \log(n/\delta)}{\epsilon^2} (1 +\epsilon)} \\ 
        & \leq \exp\left(- \frac{k \epsilon^2 \log(n/\delta)}{\epsilon^2} \right)\\
        &\leq \frac{\delta}{2}
    \end{split}
\end{equation}
We can repeat the above argument to lower bound $\|\CC \|^2_F$.
Therefore, with probability $1- \delta$, we have 
\begin{equation*}
    |\|\CC \|^2_F - \|\AA +\NN\|^2_F| \leq \epsilon \|\AA+\NN\|^2_F
\end{equation*}
Here, we can upper bound this by observing $\|\AA + \NN \|^2_F \leq \| \|^2_F  + \|\NN\|^2_F + 2\langle\AA, \NN \rangle \leq \|\AA \|^2_F + 3\sqrt{\eta}\|\AA \|^2_F$. Therefore, 
\begin{equation}
\label{eqn:pcp1}
|\|\CC \|^2_F - \|\AA\|^2_F| \leq \epsilon\|\AA \|^2_F + (1+\epsilon)\sqrt{\eta}\|\AA\|^2_F \leq (\epsilon +2 \sqrt{\eta})\|\AA\|^2_F 
\end{equation}

\noindent Next, we relate $\trace{ \XX\CC \CC^T \XX} $ and $\trace{ \XX\AA \AA^T \XX}$. First, we observe 
\begin{equation}
    \CC \CC^T = (\AA\TT + \NN\TT)(\AA\TT + \NN\TT)^T = (\AA\TT)(\AA\TT)^T + (\AA\TT)(\NN\TT)^T + (\NN\TT)(\AA\TT)^T+(\NN\TT)(\NN\TT)^T
\end{equation}
We begin by first bounding $ \trace{ \XX(\AA\TT)(\AA\TT)^T \XX}$.
Observe, $\XX$ is a rank $k$ projection matrix and we can represent it as $\ZZ \ZZ^T$, where $\ZZ \in \mathbb{R}^{n \times k}$ and has orthonormal columns. By the cyclic property of the trace, we have 
$$
\trace{\ZZ \ZZ^T (\AA\TT)(\AA\TT)^T \ZZ \ZZ^T}  =\trace{\ZZ^T (\AA\TT)(\AA\TT)^T \ZZ}  =  \sum_{j \in [k]} \ZZ^T_{*,j}(\AA\TT)(\AA\TT)^T \ZZ_{*,j}
$$
Similarly, $\trace{\ZZ \ZZ^T \AA\AA^T  \ZZ \ZZ^T}  = \sum_{j \in [k]} \ZZ^T_{*,j} \AA \AA^T  \ZZ_{*,j}$. By Theorem \ref{thm:pcp1} , we have 

\begin{equation}
    \begin{split}
         \sum_{ j \in [k]} \left(  \ZZ^T_{*,j} \AA  \AA^T  \ZZ_{*,j} - \left(\frac{\epsilon}{k}  \right) \| \AA \|_F^2 \ZZ^T_{*,j} \II \ZZ_{*,j} \right)  &\leq  \sum_{j \in [k]} \left( \ZZ^T_{*,j} (\AA\TT)(\AA\TT)^T  \ZZ_{*,j}  \right) \\
         &\leq  \sum_{ j \in [k]} \left(  \ZZ^T_{*,j} \AA  \AA^T  \ZZ_{*,j} + \left(\frac{\epsilon}{k}  \right) \| \AA \|_F^2 \ZZ^T_{*,j} \II \ZZ_{*,j}\right)
     \end{split}
\end{equation}
Since $\ZZ^T_{*,j} \ZZ_{*,j} = 1$ and $\trace{\ZZ^T\AA\AA^T\ZZ} = \trace{\XX\AA\AA^T\XX}$, we have  
\begin{equation}
\label{eqn:pcp2_head}
    \trace{\XX \AA \AA^T \XX} - \epsilon  \| \AA \|_F^2 \leq  \trace{ \XX(\AA\TT)(\AA\TT)^T \XX} \leq \trace{\XX \AA \AA^T \XX} + \epsilon \| \AA \|_F^2 
\end{equation}
Next, we focus on $\trace{\XX (\NN\TT)(\NN\TT)^T\XX} = \|\XX\NN\TT \|^2_F$. Observe, since $\TT$ is an unbiased estimator of Frobenius norm, by Markov's inequality we can show with probability at least $1-c$, $\|\XX\NN\TT\|_F = c \| \NN\|_F = O(\sqrt{\eta})\|\AA\|_F$. Therefore, we can upper bound $\trace{\XX (\NN\TT)(\NN\TT)^T\XX}$ by $O(\eta)\|\AA \|^2_F$.  Now, we focus on the cross terms. By Cauchy-Schwartz, and a Markov bound, with probability at least $1-c$,
\begin{equation}
\label{eqn:pcp3}
    \trace{\XX(\AA\TT)(\NN\TT)^T\XX} \leq \|\AA \TT\|_F \cdot \|\NN\TT\|_F \leq O(\sqrt{\eta})\|\AA\|^2_F
\end{equation}

Combining equations \ref{eqn:pcp1}, \ref{eqn:pcp2_head}, \ref{eqn:pcp3} and union bounding over the success of the random events, with probability $1-c$, 
\begin{equation*}
    \| \AA - \XX\AA \|^2_F - O(\epsilon + \sqrt{\eta})\| \AA \|^2_F \leq \| \CC - \XX\CC \|^2_F \leq \| \AA - \XX\AA \|^2_F + O(\epsilon + \sqrt{\eta})\| \AA \|^2_F
\end{equation*}
\end{proof}

\noindent \textbf{Robust Row Projection Cost Preserving Sketches.} We now extend the diagonal sampling algorithm to construct a row projection cost preserving sketch for the matrix $\CC$. We note that following the construction for $\AA$ does not immediately give a row PCP for $\CC$ since $\CC$ is no longer PSD or even square matrix. Here, all previous approaches to construct a PCP with sublinear queries hit a roadblock, since the matrix $\CC$ need not have any well defined structure apart from being a scaled subset of the columns of a PSD matrix. However, we show that sampling rows of $\CC$ proportional to the diagonal entries of $\AA$ results in a row PCP.  

We begin by relating the row norms of $\CC$ to the row norms of $\AA$. Note, we do not expect to obtain concentration here, since such a sampling procedure would then help us estimate row norms of $\AA$ up to a constant and we would be done by using \cite{fkv04}. Therefore, we obtain the following one-sided guarantee: 
\begin{lemma}
\label{lem:row_norm_bound}
Let $\AA\TT \in \mathbb{R}^{n \times t}$ be a column projection-cost preserving sketch for $\AA$ as described in Theorem \ref{thm:pcp1}.
For all $i \in [n]$, with probability at least $1 - 1/n^c$,  $$\| (\AA\TT)_{i,*}\|^2_2 \leq O\left(\log(n)\max\left\{ \|\AA_{i,*}\|^2_2, \frac{ \phi_{\max} \sqrt{n}\|\AA \|_F\AA_{i,i}}{t} \right\}\right)$$ where $c$ is a fixed constant.
\end{lemma}
\begin{proof}
Observe that $\| (\AA\TT)_{i,*}\|^2_2 = \sum_{j \in [t]} (\AA\TT)_{i,j}^2$, where  $(\AA\TT)_{i,j}^2=\frac{\trace{\AA+\NN}}{t(\AA+\NN)_{j,j}} \AA_{i,j}^2$ 
with probability $\frac{(\AA+\NN)_{j,j}}{\trace{\AA+\NN}}$. Then, 
$\expec{}{\| (\AA\TT)_{i,*}\|^2_2} =  \sum^{n}_{i=1}\AA^2_{i,j} =
\|\AA_{i,*} \|^2_2$. Next, we compute the variance of $\|(\AA\TT)_{i,*}\|^2_2$. $\varf{\|(\AA\TT)_{i,*}\|^2_2} = t \varf{(\AA\TT)^2_{i,j}} \leq \expecf{}{(\AA\TT)^4_{i,j}}$. Then, 
\begin{equation*}
    \begin{split}
        t \expecf{}{(\AA\TT)^4_{i,j}}  =  \sum_{j \in [n]} \frac{1}{t q_j}\AA_{i,j}^4
        & \leq \sum_{j \in [n]} \frac{\trace{\AA+\NN}}{t (\AA+\NN)_{j,j}}\AA_{i,j}^2 \AA_{i,i} \AA_{j,j} \\
        & \leq \frac{\trace{\AA +\NN}\phi_{\max}\AA_{i,i}}{t}\|\AA_{i,*}\|^2_2 \\
        & \leq \left(\frac{2\phi_{\max}\sqrt{n}\|\AA\|_F\AA_{i,i}}{t}\right)^2 + \|\AA_{i,*}\|^4_2 \hspace{0.4in} \textrm{ [AM-GM]} \\
    \end{split}
\end{equation*}
where we use $\AA^2_{i,j} \leq \AA_{i,i}\AA_{j,j}$, which follows from applying Cauchy-Schwarz to $\langle \AA^{1/2}_{i,*},\AA^{1/2}_{j,*}\rangle$, i.e.,
$$\AA^2_{i,j} = \langle \AA^{1/2}_{i,*},\AA^{1/2}_{j,*}\rangle^2 \leq  \|\AA^{1/2}_{i,*}\|^2_2 \|\AA^{1/2}_{j,*}\|^2_2 = \AA_{i,i}\AA_{j,j} $$
Similarly, we bound $$(\AA\TT)^2_{i,j} = \frac{\trace{\AA+\NN}}{t(\AA+\NN)_{j,j}} \AA^2_{i,j} \leq \frac{2\phi_{\max}\sqrt{n}\|\AA\|_F}{t} \AA_{i,i}$$
Applying Bernstein's inequality, 
\begin{equation*}
    \begin{split}
        &\prob{}{|\|(\AA\TT)_{i,*} \|^2_2 - \|\AA_{i,*} \|^2_2  \geq  \gamma \max\left\{ \|\AA_{i,*}\|^2_2, \frac{ \phi_{\max} \sqrt{n}\|\AA \|_F\AA_{i,i}}{t} \right\}} \\
         & \leq 2 \exp\left(- \frac{\gamma^2 \max\left\{ \|\AA_{i,*}\|^4_2,  \left(\frac{ \phi_{\max} \sqrt{n}\|\AA \|_F\AA_{i,i}}{t}\right)^2 \right\}  }{ \left(\frac{2\phi_{\max}\sqrt{n}\|\AA\|_F\AA_{i,i}}{t}\right)^2 + \|\AA_{i,*}\|^4_2 + \gamma \max\left\{ \frac{ \phi_{\max} \sqrt{n}\|\AA \|_F\AA_{i,i} \|\AA_{i,*}\|^2_2}{t}, \left(\frac{ \trace{\AA}\AA_{i,i}}{t}\right)^2 \right\} } \right) \\
        & \leq 2 \exp\left(- \frac{\gamma \max\left\{ \|\AA_{i,*}\|^4_2,  \left(\frac{ \phi_{\max} \sqrt{n}\|\AA \|_F\AA_{i,i}}{t}\right)^2 \right\}  }{c' \left(\frac{\phi_{\max} \sqrt{n}\|\AA \|_F\AA_{i,i}}{t}\right)^2 + c'\|\AA_{i,*}\|^4_2  } \right) \\
    \end{split}
\end{equation*}
where  $\frac{\phi_{\max} \sqrt{n}\|\AA \|_F\AA_{i,i}}{t}\|\AA_{i,*}\|^2_2 \leq \left( \frac{\phi_{\max} \sqrt{n}\|\AA \|_F\AA_{i,i}}{t}\right)^2 + \| \AA_{i,*}\|^4_2$ follows from the AM-GM inequality. Setting $\gamma= \Omega(\log(n))$ completes the proof.  


\end{proof}

To construct a row projection cost preserving sketch of $\CC$, we sample $t$ rows of $\CC$ proportional to the corresponding diagonal entries of $\AA$. Formally, we consider a probability distribution, $p =\{p_1, p_2, \ldots p_{n}\}$, over the rows of $\CC$ such that $p_{i} = \frac{\AA_{i,i}}{\trace{\AA}}$. Let $\RR$ be a $t \times t$ matrix where each row of $\RR$ is set to $\frac{1}{\sqrt{tp_i}}\CC_{i,*}$ with probability $p_i$. As before, $\RR $ can be represented as $\SS \CC = \SS(\AA\TT + \NN \TT)$. We first obtain a spectral guarantee for $\SS\AA\TT$, while we cannot actually compute this. 

\begin{theorem}(Spectral Bounds.)\label{thm:pcp_spectral_row}
Let $\AA\TT$ be an $n \times t$ matrix constructed as shown in Theorem \ref{thm:pcp1}. Let $p = \{p_1, p_2 \ldots p_n\}$ be a probability distribution over the rows of $\AA\TT$ such that $p_i = \frac{(\AA+\NN)_{i,i}}{\trace{\AA+\NN}}$.  Let $t =  O\left(\frac{\sqrt{n} k^2}{\epsilon^2}\log(\frac{n}{\delta})\right)$. Construct a sampling matrix $\SS$ that samples $t$ rows of $\AA\TT$ such that row $(\AA\TT)_{i,*}$ is picked with with probability $p_i$ and scaled by $\frac{1}{\sqrt{t p_i}}$ . Then, with probability at least $1-\delta$,
\[
(\AA\TT)^T(\AA\TT)  - \frac{\epsilon}{k}\| \AA\|^2_F \II \preceq (\SS\AA\TT)^T (\SS\AA\TT) \preceq (\AA\TT)^T(\AA\TT) +  \frac{\epsilon}{k}\| \AA\|^2_F\II
\]
\end{theorem}
\begin{proof}
Let $\Y = (\SS\AA\TT)^T(\SS\AA\TT) - (\AA\TT)^T(\AA\TT)$. For notational 
convenience let $(\AA\TT)_{i} = (\AA\TT)_{{i,*}}$ and $(\SS\AA\TT)_{i} = (\SS\AA\TT)_{{i,*}}$. We can then write $\Y = \sum_{i \in [t]} \left((\SS\AA\TT)^T_{{i}} 
(\SS\AA\TT)_{{i}} -\frac{1}{t}(\AA\TT)^T(\AA\TT)\right) = 
\sum_{i\in [t]} \XX_i$, where $\XX_i = 
\frac{1}{t}\left(\frac{1}{p_i}(\AA\TT)^T_{i}(\AA\TT)_{i} - (\AA\TT)^T(\AA\TT) \right)$ with probability $p_i$. We 
observe that $\expec{}{\XX_i} = \expec{}{(\SS\AA\TT)^T_{i} (\SS\AA\TT)_{i} 
-\frac{1}{t}(\AA\TT)^T(\AA\TT)}= \sum_{i}\frac{p_i}{p_i}(\AA\TT)^T_{i}(\AA\TT)_{i} - (\AA\TT)^T 
(\AA\TT)= 0$, and therefore, $\expec{}{\Y} =0$. Next,
we bound the operator norm of $\Y$. To this end,
we use the Matrix Bernstein inequality, which in turn requires a bound on the operator norm of 
$\XX_i$ and variance of $\Y$. 
Recall, for some $i' \in [n]$
\begin{equation}
\label{eqn:spec_x_row}
    \begin{split}
        \| \XX_i \|_2  & = \left\| \frac{1}{tp_{i'}}(\AA\TT)^T_{{i'}}(\AA\TT)_{{i'}} - \frac{1}{t}(\AA\TT)^T(\AA\TT) \right\|_2 \\
        & \leq \frac{1}{t p_{i'}} \|(\AA\TT)^T_{{i'}}(\AA\TT)_{{i'}} \|_2 + \frac{1}{t} \| (\AA\TT)^T(\AA\TT)\|_2 \\
        & = \frac{ \|(\AA\TT)_{{i'}} \|^2_2}{tp_{i'}} + \frac{\| (\AA\TT) \|^2_2}{t}\\
        & \leq \frac{\log(n)}{tp_{i'}} \max\left\{ \|\AA_{i',*}\|^2_2, \frac{ \phi_{\max} \sqrt{n}\|\AA \|_F \AA_{i,i}}{t} \right\} + \frac{\|(\AA\TT)\|^2_2}{t} \hspace{0.5in} \textrm{[by Lemma \ref{lem:row_norm_bound}]}\\
        & \leq \frac{\log(n)}{t}\max\left\{ \phi_{\max} \sqrt{n}\|\AA \|_F\|\AA\|_2 , \frac{(\phi_{\max} \sqrt{n}\|\AA \|_F)^2}{t}, \|(\AA\TT)\|^2_2\right\} \hspace{0.55in} \textrm{[by Lemma \ref{lem:row_norm_spectral_norm}]}\\
        & \leq \frac{\phi_{\max}\sqrt{n}\log(n)\| \AA\|^2_F}{t}\left(1 + \frac{\sqrt{n}}{t}\right) \leq  \frac{2\phi_{\max}\sqrt{n}\log(n)\| \AA\|^2_F}{t}
    \end{split}
\end{equation}
where the last inequality uses that $t = \Omega(\sqrt{n})$.
Next, we bound $\var{\Y}\leq \expecf{}{\Y^2} $ as follows 

\begin{equation}
\label{eqn:spec_y_row}
    \begin{split}
        \expecf{}{\Y^2} & =  t \left(\sum_{i \in [n]} \frac{p_i}{t^2p_i^2} ((\AA\TT)^T_{i}(\AA\TT)_{i})^2 + \frac{1}{t^2}((\AA\TT)^T(\AA\TT))^2 - \sum_{i \in [n]} \frac{2p_i}{p_i t^2} (\AA\TT)^T_{i} (\AA\TT)_{i} (\AA\TT)^T (\AA\TT) \right) \\
        & = \frac{1}{t}\left( \sum_{i \in [n]} \frac{((\AA\TT)^T_{i}(\AA\TT)_{i})^2}{p_i} + ((\AA\TT)^T(\AA\TT))^2 - \sum_{i \in [n]} 2 (\AA\TT)^T_{i}(\AA\TT)_{i} (\AA\TT)^T(\AA\TT) \right) \\
        & \preceq \frac{1}{t}\left(  \sum_{i \in [n]} \frac{((\AA\TT)^T_{i}(\AA\TT)_{i})^2}{p_i} \right) \\
        & \preceq \frac{\log(n)}{t} \max\left\{\phi_{\max}\sqrt{n}\|\AA\|_F \|\AA\|_2 , \frac{(\phi_{\max}\sqrt{n}\|\AA\|_F)^2}{t} \right\} \left(  \sum_{i \in [n]} (\AA\TT)^T_{i}(\AA\TT)_{i} \right)  \\
        & \preceq \frac{c\log(n) \sqrt{n}\|\AA \|^2_F \|(\AA\TT) \|^2_2 }{t}\II_{n\times n}
    \end{split}
\end{equation}
where we again use Lemma \ref{lem:row_norm_spectral_norm} and Theorem \ref{thm:pcp_main_col}. Observe, 

Applying the Matrix Bernstein inequality with equations \ref{eqn:spec_x_row} and \ref{eqn:spec_y_row} 

\begin{equation}
\label{eqn:case1}
    \begin{split}
        \prob{}{\| \Y \|_2 \geq \epsilon \|\AA \|^2_F }  
        & \leq 2n \exp{\left(- \frac{\epsilon^2 \| \AA\|^4_F}{  \frac{c\log(n)\sqrt{n}\phi_{\max}\|\AA\|^2_F  \|(\AA\TT) \|^2_2 }{t}  +  \frac{\epsilon \sqrt{n}\log(n)\phi_{\max}}{3t}\| \AA \|^4_F }\right)} \\
        & \leq 2n \exp{\left(-  \frac{\epsilon^2 t }{  c' \phi_{\max}\sqrt{n} \log(n)  }\right)}\\
    \end{split}
\end{equation}
where the second inequality uses Theorem \ref{thm:pcp1}, to conclude that with probability at least $1-\delta/2$ $\|\AA\TT \|^2_2 \leq \|\AA \|^2_2 +\epsilon/k\|\AA \|^2_F \leq O(\|\AA \|^2_F) $.
Therefore, it suffices to set $t = \frac{c'\phi_{\max} \sqrt{n}\log(n))}{\epsilon^2}\log(n/\delta)$, to bound the above probability by $\delta/2$. Union bounding over the error for both PCPs, and setting $\epsilon = \epsilon/k$, we can conclude that with probability at least $1-\delta$, 
\[
(\AA\TT)^T(\AA\TT)  - \frac{\epsilon}{k}\| \AA\|^2_F \II \preceq (\SS\AA\TT)^T (\SS\AA\TT) \preceq (\AA\TT)^T(\AA\TT) +  \frac{\epsilon}{k}\| \AA\|^2_F\II
\]
when $t = \Omega\left( \frac{\phi_{\max} \sqrt{n} k^2}{\epsilon^2} \right)$.

\end{proof}

We use the spectral bound from Theorem \ref{thm:pcp_spectral_row} to obtain a row projection-cost preservation guarantee.  We follow the same proof strategy as Theorem \ref{thm:pcp_main_col}, while requiring modified version of the scalar Chernoff bound. We do away with the head-tail split from \cite{cohenmm17},\cite{mw17} and \cite{bakshi2018sublinear} and analyze the projection-cost guarantee directly. This enables us to obtain a better $\epsilon$ dependence than \cite{mw17} and \cite{bakshi2018sublinear}. Note, our $\epsilon$ dependence matches that of \cite{cohenmm17} but our row projection cost preserving sketch can be computed in sub-linear time, albeit for PSD matrices. 

\begin{theorem}(Row Projection-Cost Preservation.)
\label{thm:pcp_main_row}
Given as input $\AA+\NN$ let $\CC$ be an $n \times t$ matrix as defined in Theorem \ref{thm:pcp_main_col} such that $\CC = \AA\TT + \NN \TT$. Let $p = \{p_1, p_2 \ldots p_n\}$ be a probability distribution over the rows of $\CC$ such that $p_j = \frac{(\AA+\NN)_{j,j}}{\trace{\AA+\NN}}$.  Let $t =  O\left(\frac{\phi_{max}\sqrt{n}k^2 \log^2(n)}{\epsilon^2}\right)$. Then, construct $\RR$ using $t$ rows of $\CC$ and set each one to $\frac{\CC_{i,*}}{\sqrt{t p_i}}$ with probability $p_i$. With probability at least $1-c$, for any rank-$k$ orthogonal projection $\XX$,
\[
\| \RR - \RR\XX\|^2_F = \| \CC - \CC\XX\|^2_F \pm O(\epsilon+\sqrt{\eta})\| \AA \|^2_F
\]
for a fixed constant $c$. 
\end{theorem}

\begin{proof}
Note, $\RR = \SS \CC = \SS\AA\TT + \SS\NN\TT$, where $\SS$ and $\TT$ are the corresponding sampling matrices.  
Observe, $\| \CC - \CC\XX\|^2_F = \trace{  (\mathbb{I}-\XX) \CC^T \CC(\mathbb{I}-\XX) }$. Then, 
\begin{equation}
\label{eqn:tr1}
    \begin{split}
        \trace{  (\mathbb{I}-\XX) \CC^T \CC(\mathbb{I}-\XX) } & = \trace{  \CC^T \CC} + \trace{ \XX\CC^T \CC \XX} - \trace{\CC^T \CC \XX} - \trace{ \XX\CC^T \CC} \\
        & = \trace{  \CC^T \CC^T} + \trace{ \XX^T \CC^T \CC \XX} - \trace{\CC^T \CC \XX\XX} - \trace{ \XX\XX\CC^T \CC} \\
        & = \trace{  \CC^T \CC} + \trace{ \XX\CC^T \CC \XX} - \trace{\XX\CC^T \CC \XX} - \trace{ \XX\CC^T \CC \XX} \\
        & = \trace{  \CC^T \CC} - \trace{ \XX\CC^T \CC \XX} \\
        & = \trace{  \CC^T \CC} - \trace{ \XX(\AA\TT + \NN\TT)^T (\AA\TT + \NN\TT) \XX}
    \end{split}
\end{equation}
where we used the fact that for any projection matrix $X = X^2$ in addition to the cyclic property of the trace. Here, for analyzing the cross and tail terms, we observe that with probability $1- c$, $\|\XX (\AA\TT)^T \|_F \leq O(1)\|\AA\|_F$ and $\|\XX (\NN\TT)^T \|^2_F \leq O(\eta)\|\AA\|^2_F$. Therefore, 
\begin{equation}
\label{eqn:at_expansion}
    \trace{ \XX(\AA\TT + \NN\TT)^T (\AA\TT + \NN\TT) \XX} = \trace{\XX (\AA\TT)^T \AA\TT \XX} \pm O(\sqrt{\eta})\|\AA \|^2_F 
\end{equation}
Similarly, 
\begin{equation}
\label{eqn:tr2}
    \begin{split}
        \| \RR - \RR\XX\|^2_F  = \trace{(\mathbb{I}-\XX) \RR^T \RR (\mathbb{I}-\XX)} 
        & = \trace{  \RR^T \RR} - \trace{ \XX\RR^T \RR \XX}\\
        & = \trace{  \RR^T \RR} -  \trace{ \XX(\SS\AA\TT)^T (\SS\AA \TT) \XX} \pm O(\sqrt{\eta})\| \AA \|^2_F
    \end{split}
\end{equation}
Here, we observe $\| \SS\AA\TT\|^2_F$ is an unbiased estimator for $\|\AA \|^2_F$  and $\| \SS\NN\TT\|^2_F$ is an unbiased estimator for $\|\NN \|^2_F$. Using the same idea as above, we can bound the cross and tail terms by $O(\sqrt{\eta})\|\AA \|^2_F$. 
Our goal is show that Equations \ref{eqn:tr1} and \ref{eqn:tr2} are related up to additive error $O(\epsilon +\sqrt{\eta})\|\AA\|^2_F$.  We first relate $\trace{  \CC^T\CC}$ and $\trace{  \RR^T \RR}$. Recall, $\expecf{}{\trace{\RR^T \RR}} = \expecf{}{\|\RR\|^2_F} = \|\CC\|^2_F = \trace{\CC^T\CC}$. Using a scalar Chernoff bound, we show that with probability at least $1-1/\textrm{poly}(n)$, $| \|\RR \|^2_F - \|\CC \|^2_F | \leq \epsilon \| \AA \|^2_F$. Observe, for all $i \in [t]$, $\RR_{*,i} = \frac{1}{\sqrt{p_{i'} t}} \CC_{i',*}$ for some $i'\in [n]$. Then, 
\begin{equation}
\begin{split}
    \|\RR_{i,*} \|^2_2  = \frac{1}{p_{i'} t} \| \CC_{i',*} \|^2_2 
     &= \frac{\phi_{max}\sqrt{n}\| \AA\|_F \epsilon^2}{\phi_{max} \sqrt{nk}\log(n) \log(n/\delta) \AA_{i',i'}} \| \CC_{i',*} \|^2_2 \\
     &\leq \frac{c(1+\eta) \| \AA\|_F \epsilon^2}{ \sqrt{k} \log(n/\delta)\AA_{i',i'}} \max\left\{ \|\AA_{i',*}\|^2_2, \frac{ \phi_{\max} \sqrt{n}\|\AA \|_F\AA_{i',i'}}{t} \right\} \\
     &\leq \frac{c\epsilon^2}{ \sqrt{k} \log(n/\delta)} \max \left\{ \| \AA\|_2\|\AA\|_F, \frac{\|\AA \|_F^2 \epsilon^2}{ \sqrt{k} \log(n/\delta)} \right\} \\
     & \leq \frac{c\epsilon^2}{\sqrt{k} \log(n/\delta)} \|\AA+\NN\|^2_F
\end{split}
\end{equation}
where we use $\CC_{i',*} = (\AA\TT)_{i',*} + (\NN \TT)_{i',*}$,  $\|(\NN\TT)_{i,*} \|^2_{2} \leq (\eta)\|(\AA\TT)_{i,*}\|^2_2$ for all $i$ and Lemma \ref{lem:row_norm_bound} to bound $\|(\AA\TT)_{i,*}\|^2_2$.
Therefore, $\frac{ \sqrt{k} \log(n/\delta)}{c\epsilon^2 \|\AA \|^2_F} \|\RR_{i,*} \|^2_2 \in [0,1]$. Note, $\|\RR\|^2_F$ is an unbiased estimator for $\|\AA+\NN\|^2_F$. Using a Chernoff bound, 

\begin{equation}
    \begin{split}
        \prob{}{\|\RR \|^2_F \geq  (1 +\epsilon) \|\AA +\NN\|^2_F } & = \prob{}{ \frac{\sqrt{k} \log(n/\delta)}{\epsilon^2 \|\AA \|^2_F} \|\RR \|^2_F \geq \frac{\sqrt{k} \log(n/\delta)}{\epsilon^2} (1 +\epsilon)} \\ 
        & \leq \exp\left(- \frac{\sqrt{k} \epsilon^2 \log(n/\delta)}{\epsilon^2} \right)
         \leq \frac{\delta}{10}
    \end{split}
\end{equation}
Therefore, with probability at least $1 - \delta/10$, $|\| \RR\|^2_F - \|\AA + \NN\|^2_F | \leq \epsilon \|\AA+\NN\|^2_F$. Note, we can then bound $\|\AA +\NN \|^2_F \leq \|\AA \|^2_F+ 2\sqrt{\eta}\|\AA\|^2_F$. Therefore, 
\begin{equation*}
    |\| \RR\|^2_F - \|\AA \|^2_F | \leq O(\epsilon + \sqrt{\eta})\| \AA\|_F
\end{equation*}
Recall, by equation \ref{eqn:pcp1},  with probability $\delta/10$, $\|\AA\|^2_F = (1\pm  (\epsilon +2 \sqrt{\eta}))\|\CC\|^2_F$ and thus we have that  $\| \RR\|^2_F - \|\CC\|^2_F \leq  3\epsilon \|\AA\|^2_F$.  We can repeat the above argument to lower bound $\|\RR \|^2_F$.
Therefore, with probability $1- \delta$, we have 
\begin{equation}
    \label{eqn:pcp_head_terms}
    |\|\RR \|^2_F - \|\CC \|^2_F| \leq O(\epsilon+\sqrt{\eta}) \|\AA\|^2_F
\end{equation}
\noindent Next, we relate $\trace{ \XX(\SS\AA\TT)^T (\SS\AA \TT) \XX} $ and $\trace{\XX (\AA\TT)^T \AA\TT \XX} $.
Observe, $\XX$ is a rank $k$ projection matrix and we can represent it as $\ZZ \ZZ^T$, where $\ZZ \in \mathbb{R}^{n \times k}$ and has orthonormal columns. By the cyclic property of the trace, we have 
$$
\trace{\ZZ \ZZ^T(\SS\AA\TT)^T (\SS\AA \TT) \ZZ \ZZ^T}  =\trace{\ZZ^T (\SS\AA\TT)^T (\SS\AA \TT) \ZZ}  =  \sum_{j \in [k]} \ZZ^T_{*,j} (\SS\AA\TT)^T (\SS\AA \TT)\ZZ_{*,j}
$$
Similarly, $\trace{\ZZ \ZZ^T (\AA\TT)^T \AA\TT  \ZZ \ZZ^T}  = \sum_{j \in [k]} \ZZ^T_{*,j}  (\AA\TT)^T \AA\TT  \ZZ_{*,j}$. By Theorem \ref{thm:pcp_spectral_row} , we have 

\begin{equation}
\label{eqn:pcp_spec}
    \begin{split}
         \sum_{j \in [k]} \left( \ZZ^T_{*,j} (\AA\TT)^T \AA\TT  \ZZ_{*,j}  \right) 
          = \sum_{ j \in [k]} \left(  \ZZ^T_{*,j} (\SS\AA\TT)^T (\SS\AA \TT) \ZZ_{*,j} \pm \frac{\epsilon}{k} \| \AA \|_F^2 \ZZ^T_{*,j} \II \ZZ_{*,j}\right)
     \end{split}
\end{equation}
Since $\ZZ^T_{*,j} \ZZ_{*,j} = 1$ and $\trace{\ZZ^T (\SS\AA \TT)^T(\SS\AA \TT)\ZZ} = \trace{\XX (\AA\TT)^T \AA\TT\XX}$, we obtain
\begin{equation}
\label{eqn:pcp2_}
    \trace{\XX (\AA\TT)^T \AA\TT \XX} - \epsilon \| \AA \|_F^2 \leq  \trace{ \XX(\SS\AA \TT)^T (\SS\AA \TT)\XX} \leq \trace{\XX (\AA\TT)^T \AA\TT \XX} + \epsilon \| \AA \|_F^2 
\end{equation}
Combining equations \ref{eqn:pcp2_},\ref{eqn:pcp_head_terms}, \ref{eqn:at_expansion} and \ref{eqn:tr2} with probability $1-c$, 
\begin{equation*}
    \| \CC - \CC\XX\|^2_F - O(\epsilon+\sqrt{\eta})\| \AA \|^2_F \leq \| \RR - \RR\XX\|^2_F \leq \| \CC - \CC\XX\|^2_F + O(\epsilon+\sqrt{\eta})\| \AA \|^2_F
\end{equation*}
\end{proof}

\begin{Frame}[\textbf{Algorithm \ref{alg:robust_lra}} : Robust PSD Low-Rank Approximation]
\label{alg:robust_lra}
\textbf{Input}: A  Matrix $\AA + \NN$, integer $k$, $\epsilon >0$ and $\phi_{\max} = \max_{j} \AA_{j,j}/(\AA +\NN)_{j,j}$
\begin{enumerate}
    \item Let $t= \frac{c\phi^2_{\max} \sqrt{n}k\log^2(n)}{\epsilon^2}$, for some constant $c$. Let $q = \{ q_1, q_2 \ldots q_n \} $ denote a distribution over columns of $\AA+\NN$ such that $q_j = \frac{(\AA+\NN)_{j,j}}{\trace{\AA+\NN}}$. Construct a column PCP for $\AA+\NN$ by sampling $t$ columns of $\AA+\NN$ such that each column is set to $\frac{(\AA+\NN)_{*,j}}{\sqrt{t q_j}}$ with probability $q_j$. Let $\CC$ be the resulting $n \times t$ matrix that satisfies the guarantee of Theorem \ref{thm:pcp_main_col}.
    \item Let $p = \{ p_1, p_2 \ldots p_n \} $ denote a distribution over rows of $\CC$ such that $p_i = \frac{(\AA+\NN)_{i,i}}{\trace{\AA+\NN}}$. Construct a row PCP for $\CC$ by sampling $t$ rows of $\CC$ such that each row is set to $\frac{\CC_{i,*}}{\sqrt{t p_i}}$ with probability $p_i$. Let $\RR$ be the resulting $t \times t$ matrix that satisfies the guarantee of Theorem \ref{thm:pcp_main_row}. Sample $\Theta(n)$ entries uniformly at random from $\AA$ and rescale such that $\widetilde{v}^2 = \Theta(\|\AA\|^2_F)$. 
    \item Let $\mu = \phi_{\max}\sqrt{|(\AA+\NN)_{i,i}|\cdot|(\AA+\NN)_{i',i'}|} $
    For all $i \in [t]$, let $\X_i =\sum_{j\in[\epsilon^3 t/k^3]}  \X_{i,j}$ such that $\X_{i,j} = (k^3/\epsilon^3) \RR^2_{i,i'}$, with probability $1/t$, for all $i'\in[t]$. Here, we query the entry corresponding to $\RR_{i,i'}$ in $\AA+\NN$ and truncate it to $\mu$.   
    Let $\tau = \phi_{\max}^2 n \widetilde{v}^2/t^2$. If $\X_i > \tau$, sample row $\RR_{i,*}$ with probability $1$. For the remaining rows, sample $nk/(\epsilon t)$ rows uniformly at random.
    \item Run the sampling algorithm from Frieze-Kannan-Vempala \cite{fkv04} to compute a $t\times k$ matrix $\SS$ such that $\| \RR - \RR\SS \SS^T \|^2_F \leq \|\RR -\RR_k\|^2_F + \epsilon \| \RR\|^2_F$. 
    Consider the regression problem $$\min_{\XX \in \R^{n \times k}} \| \CC -\XX\SS^T \|^2_F.$$
    Sketch the problem using the leverage scores of $\SS^T$, as shown in Lemma \ref{lem:fast_regression}, to obtain a sampling matrix $\EE$ with $O(\frac{k}{\epsilon})$ columns. Compute $$\XX_{\CC} = \argmin_{\XX \in \R^{n \times k}}\| \CC\EE - \XX\SS^T\EE \|^2_F.$$
    Let $\XX_{\CC}\SS^T = \UU\VV^T$ be such that $\UU\in\mathbb{R}^{n\times k}$ has orthonormal columns.
     \item Consider the regression problem $$\min_{\XX \in \mathbb{R}^{k \times n}} \|\AA -\UU\XX\|^2_F.$$
     Sketch the problem as above, following Lemma \ref{lem:fast_regression} to obtain a sampling matrix $\EE'$ with $O(\frac{k}{\epsilon})$ rows. Compute $$\XX_{\AA} = \argmin_{\XX} \|\EE'\AA - \EE'\UU\XX \|^2_F$$
\end{enumerate}
\textbf{Output: $\MM = \UU$, $\NN^T = \XX_{\AA}$ }  
\end{Frame}

\noindent\textbf{Full Algorithm.}
Next, we describe a sublinear time and query robust algorithm for low-rank approximation of PSD matrices. We show that querying $\widetilde{O}(\phi^2_{\max}nk/\epsilon)$ entries of $\AA$ suffices. While we assume we know $\phi_{\max}$, in practice this need not be the case. Therefore, given a budget for the total number of queries, denoted by $\beta$, we can run the algorithm by querying a $\sqrt{\beta} \times \sqrt{\beta}$ submatrix (as described in Algorithm \ref{alg:robust_lra}), but correctness only holds when $\beta \geq \widetilde{\Theta}(\phi^2_{\max}nk/\epsilon)$.  Recall, whenever we read an entry in $(\AA + \NN)_{i,j}$, we can truncate it to $\phi_{\max}\sqrt{|(\AA+\NN)_{i,i}|\cdot|(\AA+\NN)_{j,j}|}$. We can compute these thresholds by simply reading the diagonal of $\AA+\NN$. 

We proceed by constructing column and row projection-cost preserving sketches of $\AA+\NN$, to obtain a $t \times t$ matrix $\RR$, where $t = \widetilde{O}(\phi_{\max}\sqrt{n}k^2/\epsilon^2)$. 
Instead of reading the entire matrix, we sample $\epsilon^3 t/k^3$ entries in each row of $\RR$, and read these entries. Ideally we would want to estimate $\ell^2_2$ norms of rows of $\RR$ to then use a result of Frieze-Kannan-Vempala \cite{fkv04} to show that there exists an $s \times t$ matrix $\SS$ such that the row space of $\SS$ contains a good rank-$k$ approximation, where $s =  c\phi_{\max}^2 nk/\epsilon t$, for some constant $c$. However, we show that is it not possible to obtain accurate estimates of the row norms of each row of $\RR$ with high probability. 

Instead, we describe a new sampling procedure that ends up sampling rows of $\RR$ with the same probability as Frieze-Kannan-Vempala. 
Once we compute a good low-rank approximation for $\RR$ we can follow the approach of \cite{cohenmm17},\cite{mw17} and \cite{bakshi2018sublinear}, where we set up two regression problems, and use fast approximate regression to compute a low rank approximation for $\AA$. The main theorem we prove in this section is as follows:

\begin{theorem}(Robust PSD LRA.)
\label{thm:main_thm}
Let $k$ be an integer and $\epsilon>\eta>0$. 
Given a matrix $\AA + \NN$, where $\AA$ is PSD and $\NN$ is a corruption term such that $\|\NN \|^2_F \leq \sqrt{\eta}\|\AA \|^2_F$ and for all $i \in [n]$ $\|\NN_{i,*}\|^2_2\leq c\|\AA_{i,*} \|^2_2$, for a fixed constant $c$, Algorithm \ref{alg:robust_lra} samples $\widetilde{O}\left(\phi^2_{\max}nk/\epsilon\right)$ entries in $\AA+\NN$ and computes matrices $\MM, \NN^T \in \mathbb{R}^{n \times k}$ such that with probability at least $99/100$,
\begin{equation*}
    \|\AA - \MM \NN \|^2_F \leq \|\AA - \AA_k \|^2_F + (\epsilon+\sqrt{\eta})\|\AA\|^2_F
\end{equation*}
\end{theorem}

We begin with the following simple lemma for approximating the Frobenius norm : 
\begin{lemma}(Approximating Frobenius Norm.)
\label{lem:frob_estimate}
Given as input an $n \times n$ matrix $\AA+\NN$, there exists an algorithm that reads $O(\phi_{\max}^2n)$ entries in $\AA$ and outputs an estimator $\tilde{v}$ such that with probability at least $1- \frac{1}{n^c}$, $\tilde{v} = \Theta(\|\AA\|^2_F)$. 
\end{lemma}
\begin{proof}
There are multiple ways to see this. Observe, in Theorem \ref{thm:pcp_main_row}, we show that sampling $ \frac{\phi_{\max}^2n\log(n)}{\epsilon^2}$ entries results in row projection-cost preserving sketch $\RR$ such that $\| \RR \|^2_F = (1\pm \epsilon )\| \AA + \NN \|^2_F$. Setting $\epsilon$ to be a small constant suffices. 
\end{proof}

Next, we provide intuition for why uniformly sampling columns of $\RR$ does not suffice for obtaining a sketch that spans a good low rank approximation. For simplicity, we assume there is no noise ($\eta =0$ and $\phi_{\max}=1$) and show that our techniques to bound the column norms of $\RR$ results in an estimate that is too large. We note that this lemma is not required for proving our result, and is just present for intuition. 

\begin{lemma}
\label{lem:large_frob_norm}
Let $\eta=0$.
Let $\RR \in \mathbb{R}^{t \times t}$ be a row projection-cost preserving sketch for $\CC$ as described in Theorem \ref{thm:pcp_main_row}.
For all $j \in [t]$, with probability at least $1 - 1/n^c$,  $$\| \RR_{*,j}\|^2_2 \leq O\left(\log(n)\max\left\{ \|\CC_{*,j}\|^2_2, \frac{ n \|\AA \|^2_F}{t^2} \right\}\right) = O\left(\log(n)\max\left\{\frac{\sqrt{n}\|\AA\|^2_F}{t }, \frac{ n \| \AA\|^2_F}{t^2} \right\}\right)$$ where $c$ is a fixed constant.
\end{lemma}
\begin{proof}
Observe, $\| \RR_{*,j}\|^2_2 = \sum_{i \in [t]} \RR_{i,j}^2$, where  $\RR_{i,j}^2=\frac{\trace{\AA}}{t\AA_{i',i'}} \CC_{i',j}^2$ 
with probability $\frac{\AA_{i',i'}}{\trace{\AA}}$ for all $i' \in [n]$. Then, 
$\expec{}{\| \RR_{*,j}\|^2_2} =  \sum^{n}_{i=1}\CC^2_{i,j} =
\|\CC_{*,j} \|^2_2$. Next, we compute the variance of $\|\RR_{*,j}\|^2_2$. $\varf{\|\RR_{*,j}\|^2_2} = t \varf{\RR_{i,j}} \leq t \expecf{}{\RR^4_{i,j}}$. Then, 
\begin{equation*}
    \begin{split}
        t \expecf{}{\RR^4_{i,j}}  =  \sum_{i' \in [n]} \frac{1}{t p_{i'}}\CC_{i',j}^4
        & \leq \sum_{i' \in [n]} \frac{\trace{\AA}^2}{t^2 \AA_{i',i'} \AA_{j,j}}\AA_{i',j}^2 \AA_{i',i'} \AA_{j,j} \\
        & \leq \frac{\trace{\AA}^2}{t^2}\|\AA_{*,j}\|^2_2 \\
        & = \frac{\trace{\AA}}{t}\|\CC_{*,j} \|^2_2\\
        & \leq \left(\frac{\trace{\AA}}{t}\right)^2 + \|\CC_{*,j}\|^4_2 \hspace{0.4in} \textrm{ [AM-GM]} \\
    \end{split}
\end{equation*}
where we use $\AA^2_{i,j} \leq \AA_{i,i}\AA_{j,j}$, which follows from applying Cauchy-Schwarz to $\langle \AA^{1/2}_{i,*},\AA^{1/2}_{j,*}\rangle$. Similarly, we bound $\RR^2_{i,j} \leq \frac{\trace{\AA}}{t} $. Applying Bernstein's inequality, 
\begin{equation*}
    \begin{split}
        &\prob{}{|\|\RR_{*,j} \|^2_2 - \|\CC_{*,j} \|^2_2  \geq  \eta \max\left\{ \|\CC_{*,j}\|^2_2, \frac{ \trace{\AA}}{t} \right\}} \\
         & \leq 2 \exp\left(- \frac{\eta^2 \max\left\{ \|\CC_{*,j}\|^4_2,  \left(\frac{ \trace{\AA}}{t}\right)^2 \right\}  }{\left(\frac{\trace{\AA}}{t}\right)^2 + \|\CC_{*,j}\|^4_2 + \eta \max\left\{ \frac{ \trace{\AA} \|\CC_{*,j}\|^2_2}{t}, \left(\frac{ \trace{\AA}}{t}\right)^2 \right\} } \right) \\
        & \leq 2 \exp\left(- \frac{\eta \max\left\{ \|\CC_{*,j}\|^4_2,  \left(\frac{ \trace{\AA}}{t}\right)^2 \right\}  }{c' \left(\frac{\trace{\AA}}{t}\right)^2 + c'\|\CC_{*,j}\|^4_2  } \right) \\
    \end{split}
\end{equation*}
where we use the AM-GM inequality on $\frac{\trace{\AA}}{t}\|\CC_{*,j}\|^2_2$ repeatedly. Setting $\eta = \Omega(\log(n))$ completes the proof. Finally, observe, for any $j \in [t]$, $\|\CC_{*,j}\|^2_2 = \frac{\trace{\AA}}{t \AA_{j',j'}}\|\AA_{*,j'} \|^2_2$ for some $j' \in [n]$. We then use $\trace{\AA} \leq \sqrt{n}\|\AA\|_F$.
\end{proof}



\noindent 
It is well-known that to recover a low-rank approximation for $\RR$, one can sample rows of $\RR$ proportional to row norm estimates, denoted by $\Y_i$ \cite{fkv04}. As shown in \cite{indyk2019sample} the following two conditions are a relaxation of those required in \cite{fkv04}, and suffice to obtain an additive error low-rank approximation :   
\begin{enumerate}
	\item For all $i \in [t]$, $\mathcal{Y}_i \geq \|\RR_{i,*} \|^2_2$.
	\item $\sum_{i \in [t]} \mathcal{Y}_i \leq \frac{n}{t} \| \RR \|^2_F$
\end{enumerate}
To satisfy the first condition, we need to obtain overestimates for each row. Since it is not immediately clear how to obtain overestimates for row norms of $\RR$, a na\"ive approach would be to bias the estimate for each row by an upper bound on the row norm. However, by Lemma \ref{lem:large_frob_norm}, a row norm could be as large as $ \sqrt{n} \|\AA \|^2_F/t$. Observe, we cannot afford to bias the estimator of each row, $\mathcal{Y}_i$, by this amount since $\sum_{i \in [t]} \mathcal{Y}_i \geq \sqrt{n} \|\AA \|^2_F \geq \sqrt{n} \|\RR \|^2_F $. Therefore, we would have to sample $\sqrt{n}k/\epsilon$ rows of $\RR$, resulting in us querying $\Omega(nk^{3}/\epsilon^{3})$ entries in $\AA$, even when $\eta=0$.

An alternative strategy would be to bias the estimator for each row by $n\|\RR\|^2_F/t^2$, as this would satisfy condition 2 above. 
We can now hope to detect rows of $\RR$ that have norm larger than $n\|\RR\|^2_F/t^2$ by sampling $\epsilon^3 t/k^3$ entries in each row of $\RR$, uniformly at random. Note, this way we can construct an unbiased estimator for the $\ell^2_2$ norm of each row. Ideally, we would want to show a high probability statement for concentration of our row norm estimates around the expectation. We could then union bound, and obtain concentration for all $i$ simultaneously.  

However, this is not possible since it may be the case that a row of $\RR$ is $\log(n)$-sparse with each entry being large in magnitude. In this case, uniformly querying the row would not observe any non-zero with good probability and thus cannot distinguish between such a row and an empty row. Instead, we settle for a weaker statement, that shows our estimate is accurate with $o(1)$ probability. All subsequent statements hold for $\eta>0$.

\begin{lemma}(Estimating large row norms.)
\label{lem:large_norm}
Let $\RR \in \mathbb{R}^{t \times t }$ be the row PCP output by Step 2 of Algorithm \ref{alg:robust_lra}. For all $i \in [t]$ let $\X_i = \sum_{j \in [\epsilon^3 t/k^3]} \X_{i,j}$ such that $\X_{i,j} = \frac{k^3\RR^2_{i,j'}}{\epsilon^3}$ with probability $\frac{1}{t}$, for all $j' \in [t]$. Then, for all $i \in [t]$, $\X_i = \left(1\pm \frac{1}{10}\right)\|\RR_{i,*}\|^2_2$ with probability at least $\frac{\|\RR_{i,*} \|^2_2 k}{\epsilon n}$.
\end{lemma}
\begin{proof}
Observe, $\X_i$ is an unbiased estimator of $\| \RR_{i,*}\|^2_2$ : 
\begin{equation*}
    \expecf{}{\X_i} = \frac{\epsilon^3 t}{k^3} \expecf{}{\X_{i,j}} = \frac{\epsilon^3 t}{k^3} \sum_{j' \in [t]}\frac{k^3}{\epsilon^3n}\RR^2_{i,j'} = \| \RR_{i,*}\|^2_2
\end{equation*}
Next, we compute the variance of $\X_i$. 

\begin{equation}
\label{eqn:bern_1}
    \begin{split}
        \varf{\X_i} = \frac{\epsilon^3 t}{k^3} \varf{\X_{i,j}} & \leq \sum_{j \in [t]} \frac{1}{\epsilon^3 } \RR^4_{i,j} \\
        & \leq \sum_{j \in [t]} \frac{k^3}{\epsilon^3 } \RR^2_{i,j}\left( \frac{\trace{\AA+\NN}^2}{t^2(\AA+\NN)_{i,i}(\AA+\NN)_{j,j} }(\AA+\NN)^2_{i,j} \right)\\
        & \leq \sum_{j \in [t]} \frac{k^3}{\epsilon^3 } \RR^2_{i,j}\left( \frac{\trace{\AA+\NN}^2}{t^2(\AA+\NN)_{i,i}(\AA+\NN)_{j,j} }( \AA_{i,i}\AA_{j,j} + \NN^2_{i,j})  \right)\\
        & \leq \sum_{j \in [t]} \frac{k^3}{\epsilon^3 } \RR^2_{i,j}\left( \frac{\trace{\AA+\NN}^2 \phi^2_{\max}}{t^2} + \frac{\trace{\AA+\NN}^2 \phi^2_{\max}(\AA+\NN)_{i,i}(\AA+\NN)_{j,j} }{t^2(\AA+\NN)_{i,i}(\AA+\NN)_{j,j} }  \right)\\
        & \leq \sum_{j \in [t]} \frac{k^3}{\epsilon^3 } \RR^2_{i,j}\left( \frac{\trace{\AA+\NN}^2 \phi^2_{\max}}{t^2}   \right)\leq O\left(\frac{\epsilon\|\AA \|^2_F}{k} \|\RR_{i,*} \|^2_2 \right)
    \end{split}
\end{equation}
Here, we use that $\NN^2_{i,j} \leq \phi^2_{\max}(\AA+\NN)_{i,i}(\AA+\NN)_{j,j}$, which follows from our truncation procedure. Further, using $t = \phi_{\max}\sqrt{n}k^2/\epsilon^2$ and $\trace{\AA + \NN} \leq \sqrt{n}\|\AA \|_F + \sqrt{\eta n}\|\AA \|_F$, we can bound 
$$ \frac{\trace{\AA+\NN}^2 \phi^2_{\max}}{t^2} \leq O\left(\frac{\epsilon\|\AA \|^2_F}{k} \right)$$
\noindent Further, using the same argument as above  
\begin{equation}
\label{eqn:bern_2}
\end{equation}
$$\X_{i,j} \leq \frac{k^3}{\epsilon^3} \RR^2_{i,j} \leq O\left(\frac{\epsilon \|\AA \|^2_F}{k}\right) $$ 
Using Equations \ref{eqn:bern_1} and \ref{eqn:bern_2} in Bernstein's inequality, 
\begin{equation*}
\begin{split}
    \Pr\left[ \big|\X_i -\expec{}{\X_i}\big|\geq  \delta \expecf{}{\X_i}\right] &\leq \exp\left(- \frac{\delta^2 \|\RR_{i,*} \|^4_2 }{\frac{\epsilon\|\AA \|^2_F}{k} \|\RR_{i,*} \|^2_2  + \frac{\delta \epsilon\|\AA \|^2_F}{3k}\|\RR_{i,*} \|^2_2}\right) \\
    & \leq \exp\left(- \frac{\delta^2 \|\RR_{i,*} \|^2_2 k \log^2(n) }{\epsilon \|\RR\|^2_F}\right)
\end{split}
\end{equation*}
where we use that $\|\AA\|^2_F = \Theta(\|\RR\|^2_F) $.
Setting $\eta = \frac{1}{10}$, $\X_i = \left(1 \pm \frac{1}{10} \right)\|\RR_{i,*} \|^2_2$ with probability at least $1-\exp\left(- \frac{ \|\RR_{i,*} \|^2_2 k \log^2(n) }{\epsilon \| \RR\|^2_F}\right)$. Let $\xi_i$ be the event that  $\X_i = \left(1 \pm \frac{1}{10} \right)\| \RR_{i,*}\|^2_2$. Then, union bounding over $t\leq n$ such events $\xi_i$, simultaneously for all $i$, $\xi_i$ is true with probability at least 
\begin{equation*}
    1-\exp\left(- \frac{ \|\RR_{i,*} \|^2_2 k \log(n) }{\epsilon \|\RR \|^2_F}\right) \geq \frac{ \|\RR_{i,*} \|^2_2 k \log(n) }{\epsilon \|\RR \|^2_F} \qedhere
\end{equation*}
\end{proof}

We now have two major challenges: first, the probability with which 
the estimators are accurate is too small to even detect all rows with 
norm larger than $\phi_{\max}^2 n\| \RR \|^2_F/t^2$, and second, there is no small query certificate for when an estimator is accurate in 
estimating the row norms. Therefore, we cannot even identify the rows 
where we obtain an accurate estimate of their norm.   

To address the first issue, we make the crucial observation that while we cannot estimate the norm of each row accurately, we can hope to sample the row with the same probability as Frieze-Kannan-Vempala \cite{fkv04}. Recall, their algorithm samples row $\RR_{i,*}$ with probability at least $\|\RR_{i,*}\|^2_2/\|\RR\|^2_F$, which matches the probability in Lemma \ref{lem:large_norm}. Therefore, we can focus on designing a weaker notion of identifiability, that may potentially include extra rows.  

We begin by partitioning rows of $\RR$ into two sets. Let $\mathcal{H} = \left\{ i \textrm{ }\big|\textrm{ } \|\RR_{i,*} \|^2_2 \geq \phi_{\max}^2 n \widetilde{v}^2/t^2 \right\}$ be the set of heavy rows and $[t] \setminus \mathcal{H}$ be the remaining rows. Since with 
probability at least $1-\frac{1}{\poly(n)}$, $\| \RR \|^2_F = \Theta(\| \AA\|^2_F) = \Theta(\widetilde{v}^2)$,   $$|\mathcal{H}| = O( t^2/ \phi_{\max}^2 n ) =O( k^4 \log^4(n)/\epsilon^4)$$
Observe, every row in $\mathcal{H}$ can potentially satisfy the 
threshold $\tau= \phi_{\max}^2 n \widetilde{v}^2/t^2$. Therefore, even if our estimator $\X_i$ is 
$\Theta(\|\RR_{i,*} \|^2_2)$ for all $i \in \mathcal{H}$, we include 
at most $\widetilde{O}(k^4/\epsilon^4)$ extra rows in $\SS$, which is well within our budget. Observe, we can then sample a row with probability $1$ whenever the corresponding estimate is larger than $\tau$. This sampling process ensures that we identify rows in $\mathcal{H}$ with the right probability and also doees not query more than $O(\phi_{\max}^2nk/\epsilon)$ entries in $\AA + \NN$. 
For all the remaining rows, we know the norm is at most 
$O(\phi_{\max}^2 n \widetilde{v}^2/t^2)$. We then modify the analysis of 
\cite{fkv04} to show that we can handle both cases separately.

\begin{theorem}(Existence \cite{fkv04}.)
Let $\RR$ be a row projection-cost preserving sketch output by Step 2 of Algorithm \ref{alg:robust_lra}. For all $i \in [t]$, let $\mathcal{X}_i$ be estimate for $\| \RR_{i,*}\|^2_2$ as described in Step 3 of Algorithm \ref{alg:robust_lra}. Let $\SS$ be a subset of $s= O(\phi_{\max}^2 nk/\epsilon t)$ columns of $\RR$ sampled according to distribution $r = \{ r_1, r_2, \ldots r_t \}$ such that $r_i$ is the probability of sampling the $i$-th row.
Then, with probability at least $99/100$, there exists a $t \times k$ matrix $\UU$ in the column span of $\SS$ such that
\begin{equation*}
    \|\RR  - \UU\UU^T\RR \|^2_F \leq \|\RR - \RR_k \|^2_F + \epsilon \|\RR\|^2_F
\end{equation*}
\end{theorem}

\begin{proof}
We follow the proof strategy of \cite{fkv04} and show how to directly bound the variance in our setting as opposed to reducing to the two conditions above. 
Let $\RR = \PP\Sig\QQ^T = \sum_{\ell \in [t]} \Sig_{\ell,\ell} \PP_{\ell,*} \QQ_{\ell,*}^T = \sum_{\ell \in [t]} \sigma_{\ell} \PP_\ell \QQ_\ell^T$. Recall, $\RR_k = \sum_{\ell \in [k]}\AA \QQ_\ell \QQ_\ell^T $. For $\ell \in [t]$, let  $\WW_{\ell} = \frac{1}{s} \sum_{i' \in [s]} \YY_{i'} $ where $\YY_{i'} = \frac{\PP_{i,\ell}}{r_i} \RR_{i,*}$ with probability $r_i$, for all $i \in [t]$. Then, 
\begin{equation}
\expecf{}{\WW_{\ell}}  = \expecf{}{\YY_{i'}} =  \sum_{j\in[t]}\frac{\PP_{i,\ell}}{r_i} \RR_{i,*} r_i = \sigma_{\ell}  \QQ_{\ell}
\end{equation}
Therefore, in expectation the span of the rows contain a good low-rank solution.  
Next, we bound the variance. 
Recall, here we consider the rows in $\mathcal{H}$ and its complement separately. From Lemma For all $i \in \mathcal{H}$, we know that $\X_i = \Theta(\|\RR_{i,*} \|^2_2)$ with probability at least $ \|\RR_{i,*} \|^2_2 k \log(n) / \epsilon \|\RR \|^2_F$. Since for all such $i$, $ \|\RR_{i,*} \|^2_2 \geq \tau$, the corresponding $r_i \geq \frac{ \|\RR_{i,*} \|^2_2 k \log(n) }{\epsilon \|\RR \|^2_F}$, since every time we pass the threshold we sample the row. For all $i \notin \mathcal{H}$, $r_i \geq 1/t$ since there can be at most $t$ such $i$, and we sample each such row with uniform probability. 
Once we have a lower bound on $r_i$ in both cases, we open up the analysis of the variance bound in \cite{fkv04} and show that our lower bounds suffice.  

\begin{equation}
\begin{split}
\expecf{}{\| \WW_{\ell} - \sigma_\ell \QQ_\ell  \|_2^2} = \frac{1}{s}\left(\sum_{j\in[t]}\frac{\PP^2_{i,\ell}}{r_i} \|\RR_{i,*}\|^2_2 \right)  - \frac{\sigma^2_\ell}{s} 
& \leq  \frac{1}{s}\left( \sum_{j \in \mathcal{H}}\frac{\PP^2_{i,\ell}}{r_i} \|\RR_{i,*}\|^2_2  + \sum_{i \in [t] \setminus \mathcal{H} } \frac{\PP^2_{i,\ell}}{r_i} \|\RR_{i,*}\|^2_2  \right) \\
& \leq  \frac{1}{s}\left( \sum_{j \in \mathcal{H}} \frac{ k \log(n) \PP^2_{i,\ell} \|\RR\|^2_F  }{ \epsilon }  + \sum_{j \in [t] \setminus \mathcal{H} } t \PP^2_{i,\ell} \|\RR_{i,*} \|^2_2   \right) \\
& \leq  \frac{1}{s}\left( \sum_{j \in \mathcal{H}} \frac{ k \log(n) \PP^2_{i,\ell} \|\RR\|^2_F  }{ \epsilon }  + \sum_{j \in [t] \setminus \mathcal{H} } \frac{n}{t} \PP^2_{i,\ell} \widetilde{v}^2   \right) \\
&\leq \frac{1}{s}\left(\frac{k\log(n)}{\epsilon} + \frac{\phi_{\max}^2 n}{t} \right) \| \RR \|^2_F
\end{split}
\end{equation}
Now, we can repeat the argument from \cite{fkv04} and it suffices to set $s = \left(\frac{\phi_{\max}^2 n}{t} + \frac{k}{\epsilon}\right) \frac{k}{\epsilon} = O(\frac{\phi_{\max}^2 nk}{\epsilon t})$. For completeness, we present the rest of the proof here. 
For all $\ell \in [t]$, let $\YY_{\ell}= \frac{1}{\sigma_{\ell}} \WW_{\ell}$. Let $\mathcal{V} =\textrm{span}(\YY_{1},\YY_{2},\ldots, \YY_{k})$. Let $\ZZ_{1}, \ZZ_{2}, \ldots \ZZ_{t}$ be an orthonormal basis for $\R^{t}$ such that $\mathcal{V} = \textrm{span}(\ZZ_{1},\ZZ_{2},\ldots, \ZZ_{k'})$, where $k'=\textrm{dim}(\mathcal{V})$. Let $\SS = \RR \sum_{\ell \in [k]}\ZZ_{\ell}\ZZ_{\ell}^T$ be the candidate low-rank approximation approximation. Then, 
\begin{equation}
\label{eqn:fkv_error1}
    \begin{split}
        \|\RR - \SS \|^2_F & = \sum_{\ell \in [t]} \| (\RR - \SS) \ZZ_{\ell} \|^2_2 \\
        & = \sum_{\ell \in [k'+1, t]} \| \RR  \ZZ_{\ell} \|^2_2 \\
        & = \sum_{\ell \in [k'+1, t]} \left\| \left((\RR - \RR \sum_{\ell' \in [k]}\QQ_{\ell'}\YY_{\ell'}^T \right)  \ZZ_{\ell} \right\|^2_2 \\
        &\leq \left\| \RR - \RR \sum_{\ell' \in [k]}\YY_{\ell'}\YY_{\ell'}^T \right\|^2_F
    \end{split}
\end{equation}
where the first equality follows from $\|\ZZ_{\ell}\|^2_2=1$, the seconds follows from $\ZZ_{\ell'}^T \ZZ_{\ell} =0$ for $\ell'\neq \ell$, the third follows from $\langle \YY_{\ell'}, \ZZ_{\ell} \rangle = 0$ for all $\ell' \leq k$ and $\ell > k'$. Let $\widehat{\SS} = \RR \sum_{\ell' \in [k]} \YY_{\ell'} \YY_{\ell'}^T$. Since $\PP_{1},\PP_{2}, \ldots, \PP_{t}$ forms an orthonormal basis
\begin{equation}
\label{eqn:fkv_error2}
\begin{split}
    \left\| \RR - \widehat{\SS} \right\|^2_F & \leq  \sum_{\ell \in[t]}\left\| \PP_{\ell} \left(\RR -\widehat{\SS}\right) \right\|^2_2 \\
    & = \sum_{\ell \in[k]}\left\| \sigma_{\ell} \QQ_{\ell} - \WW_{\ell}\right\|^2_2 + \sum_{\ell\in [k+1, t]} \sigma_{\ell}^2
\end{split}
\end{equation}
Taking expectations on both sides of equation \ref{eqn:fkv_error2}, we have 
\begin{equation}
\begin{split}
\expecf{}{\left\| \RR - \widehat{\SS} \right\|^2_F} & \leq \expecf{}{ \sum_{\ell \in[k]}\left\| \sigma_{\ell} \QQ_{\ell} - \WW_{\ell}\right\|^2_2} + \|\RR -\RR_{k}\|^2_F \\
& \leq \frac{k}{s}\left(\frac{k\log(n)}{\epsilon} + \frac{\phi_{\max}^2n}{t} \right) \| \RR \|^2_F + \|\RR -\RR_{k}\|^2_F
\end{split}
\end{equation}
Since $\widehat{\SS}$ is a rank $k$ matrix and $\RR_{k}$ is the best rank $k$ approximation to $\RR$, $\| \RR - \widehat{\SS}\|^2_F - \| \RR - \RR_k\|^2_F$ is a non-negative random variable. Thus, using Markov's inequality and Equation \ref{eqn:fkv_error1}, 
\begin{equation*}
\Pr\left[ \| \RR - \SS\|^2_F - \| \RR - \RR_k\|^2_F \geq \frac{100nk}{st} \|\RR \|^2_F \right] \leq \frac{1}{100}
\end{equation*}
Therefore, it suffices to sample $s = O\left(\frac{\phi_{\max}^2nk}{\epsilon t} \right)$ columns, read all of them and compute a low rank approximation for $\RR$ with probability at least $\frac{99}{100}$. Observe, the total entries read by this algorithm is $O\left(\frac{\phi_{\max}^2nk}{\epsilon t} \cdot t \right) = O\left(\frac{\phi_{\max}^2 nk}{\epsilon} \right)$. 
\end{proof}

It remains to show that we can now recover a low-rank approximation for $\AA$, in factored form, from the low-rank approximation for $\RR$. Here, we follow the approach of \cite{cohenmm17},\cite{mw17} and \cite{bakshi2018sublinear}, where we set up two regression problems, and use the sketch and solve paradigm to compute an approximate solution. We use the following Lemma from \cite{bakshi2018sublinear} that relates a good low-rank approximation of an additive error project-cost preserving sketch to a low-rank approximation of the original matrix. A similar guarantee for relative error appears in \cite{cohenmm17} and \cite{mw17}. 

\begin{lemma}(Lemma 4.4 in \cite{bakshi2018sublinear}.)
\label{lem:pcp_useful}
Let $\CC$ be a column PCP for $\AA$ satisfying the guarantee of Theorem \ref{thm:pcp_main_col}. Let $\PP^*_{\CC} = \argmin_{\textrm{rank}(\XX)\leq k}\| \CC - \XX\CC \|^2_F$
and $\PP^*_{\AA} =\argmin_{\textrm{rank}(\XX)\leq k} \|\AA - \XX\AA \|^2_F$.
Then, for any rank $k$ projection matrix $\PP$ such that $\|\CC - \PP \CC \|^2_F \leq  \| \CC - \PP^*_{\CC} \CC\|^2_F + (\epsilon+\sqrt{\eta}) \|\CC\|^2_F$, with probability at least $99/100$,
\[
\|\AA - \PP \AA \|^2_F \leq  \| \AA - \PP^*_{\AA} \AA\|^2_F + (\epsilon+\sqrt{\eta}) \| \AA \|^2_F
\]
A similar guarantee holds for a row PCP of $\AA$. 
\end{lemma}
\noindent  Note, while $\RR \SS \SS^T$ is an approximate rank-$k$ solution for $\RR$, it does not have the same dimensions as $\AA$. If we do not consider running time, we could construct a low-rank approximation to $\AA$ as follows: since projecting $\RR$ onto $\SS^T$ is approximately optimal, it follows from Lemma \ref{lem:pcp_useful} that with probability $99/100$,
\begin{equation}
\label{eqn:pcp_as}
\| \CC -\CC\SS\SS^T  \|^2_F =\| \CC - \CC_k \|^2_F \pm (\epsilon+\sqrt{\eta}) \| \CC \|^2_F
\end{equation}
Let $\CC_k = \UU'\VV'^T$ be such that $\UU'$ has orthonormal columns. Then, $\|\CC - \UU'\UU'^T \CC \|^2_F = \| \CC - \CC_k \|^2_F$ and by Lemma \ref{lem:pcp_useful} it follows that with probability $98/100$, $\|\AA - \UU'\UU'^T\AA\|^2_F \leq \|\AA - \AA_k \|^2_F + (\epsilon+\sqrt{\eta})\| \AA \|^2_F$. However, even approximately computing a column space $\UU'$ for $\CC_k$ using an input-sparsity time algorithm, such as \cite{clarkson2013low}, could require $\Omega\left(nt\right)$ queries. To get around this issue, we observe that an approximate solution for $\RR$ lies in the row space of $\SS^T$ and therefore, an approximately optimal solution for $\CC$ lies in the row space of $\SS^T$. We then set up the following regression problem:
\begin{equation}
\label{eqn:col_pcp_reg}
\min_{\textrm{rank}(\XX)\leq k} \|\CC - \XX \SS^T \|^2_F
\end{equation}
Note, this regression problem is still too large to be solved in sublinear time. Therefore, we sketch it by sampling columns of $\CC$ to set up a smaller regression problem. Observe, since $\SS$ has orthonormal columns, the leverage scores are simply $\ell^2_2$ norms of rows of $\SS$. Now, using  Lemma \ref{lem:fast_regression},  approximately solving this regression problem requires sampling $\Omega(k/\epsilon)$ rows of $\CC$, which in turn requires $\Omega(\frac{nk}{\epsilon})$ queries to $\AA+\NN$. 
Note, the above theorem applied to Equation \ref{eqn:col_pcp_reg} can take $O\left(nk+ \textrm{poly}(k,\epsilon^{-1})\right)$ time and thus is a lower order term.  
Since $\SS^T$ has orthonomal rows, the leverage scores are precomputed. With probability at least $99/100$, we can compute $\XX_{\CC} = \argmin_{\XX} \|\CC\EE - \XX \SS^T\EE \|^2_F$, where $\EE$ is a leverage score sketching matrix with $O\left(\frac{k}{\epsilon}\right)$ columns, as shown in Lemma \ref{lem:fast_regression}, and thus requires $O\left(\frac{nk}{\epsilon}\right)$ queries to $\AA$. 
Then,
\begin{equation}
\begin{split}
\|\CC - \XX_{\CC} \SS^T \|^2_F & \leq (1+\epsilon) \min_{\XX}\|\CC - \XX \SS^T \|^2_F \\
& \leq (1+\epsilon) \| \CC -\CC\SS\SS^T  \|^2_F \\ 
& = \| \CC - \CC_k \|^2_F \pm (\epsilon+\sqrt{\eta})\| \CC \|^2_F
\end{split}
\end{equation}
where the last two inequalities follow from equation \ref{eqn:pcp_as}. 
Let  $\XX_{\CC}\SS^T = \UU'\VV'^T $ be such that $\UU'$ has orthonormal columns. Then, the column space of $\UU'$ contains an approximately optimal solution for $\AA$, since $\|\CC - \UU'\VV'^T \|^2_F = \| \CC -\CC_k \|^2_F \pm \epsilon \|\CC \|^2_F$ and $\CC$ is a column PCP for $\AA$. It follows from Lemma \ref{lem:pcp_useful}  that with probability at least $98/100$, 
\begin{equation}
\label{eqn:additive_for_a}
\| \AA - \UU'\UU'^T\AA \|^2_F \leq \| \AA - \AA_k \|^2_F + (\epsilon+\sqrt{\eta}) \|\AA\|_F
\end{equation}
Therefore, there exists a good solution for $\AA$ in the column space of $\UU'$. Since we cannot compute this explicitly, we set up the following regression problem: 
\begin{equation}
\min_{\XX} \|\AA - \UU'\XX \|^2_F
\end{equation}
Again, we sketch the regression problem above by sampling columns of $\AA$ and apply Lemma \ref{lem:fast_regression}. We can then compute $\XX_{\AA} = \argmin_{\XX} \| \EE' \AA - \EE' \UU'\XX \|^2_F $ with probability at least $99/100$, where $\EE'$ is a sketching matrix with $\left(\frac{k}{\epsilon}\right)$ rows and $O\left(\frac{nk}{\epsilon}\right)$ queries to $\AA$. Then,
\begin{equation}
\begin{split}
\| \AA -  \UU'\XX_{\AA} \|^2_F & \leq (1+\epsilon)\min_{\XX} \| \AA - \UU'\XX \|^2_F \\
& \leq (1+\epsilon) \|\AA - \UU' \UU'^T\AA \|^2_F \\
& \leq  \| \AA - \AA_k \|^2_F + O(\epsilon+\sqrt{\eta}) \|\AA\|^2_F
\end{split}
\end{equation}
where the second inequality follows from $\XX$ being the minimizer and $\UU'^T\AA$ being some other matrix, and the last inequality follows from equation \ref{eqn:additive_for_a}.
Recall, $\UU'$ is an $n \times k$ matrix and the time taken to solve the regression problem is $O\left(nk + \textrm{poly}(k,\epsilon^{-1})  \right)$. 

Therefore, we observe that $\UU'\XX_{\AA}$ suffices and we output it in factored form by setting $\MM = \UU'$ and $\NN = \XX_{\AA}^T $. Union bounding over the probabilistic events, and rescaling $\epsilon$, with probability at least $9/10$, Algorithm \ref{alg:robust_lra} outputs $\MM \in \mathbf{R}^{n \times k}$ and $\NN^T \in \mathbf{R}^{n \times k}$ such that the total number of entries queried in $\AA$ are $\widetilde{O}\left(\frac{\phi^2_{\max}nk}{\epsilon}\right)$. This concludes the proof of Theorem \ref{thm:main_thm}.

\vspace{0.2in}
\textbf{Correlation Matrices.}
We introduce low-rank approximation of correlation Matrices, a special case of PSD matrices
where the diagonal is all $1$s. 
Correlation matrices are well studied in numerical linear algebra, statistics and finance since an important statistic of $n$ random variables $X_1, X_2, \ldots X_n$ is given by computing the pairwise correlation coefficient, $\textrm{corr}(X_i,X_j) = \textrm{cov}(X_i,X_j)/\sqrt{\textrm{var}(X_i)\cdot \textrm{var}(X_j)}$. A natural matrix representation of correlation coefficients results in a  $n \times n$ correlation matrix $\AA$ such that $\AA_{i,j} = \textrm{corr}(X_i, X_j)$. 

\begin{definition}(Correlation Matrices.)
$\AA$ is an $n \times n$ correlation matrix if $\AA$ is PSD and $\AA_{i,i} =1$, for all $i \in [n]$. 
\end{definition}

Often, in practice the correlation matrices obtained are close to being PSD, but corrupted by noise in the form of missing or asynchronous observations, stress testing or aggregation. Here the goal is to query few entries of the corrupted matrix and recover a rank-$k$ matrix close to the underlying correlation matrix, assuming that the underlying matrix is also close to low rank to begin with. 

Here we observe that since correlation matrices have all diagonal entries equal to $1$, we can compute $\phi_{\max}$ by simply reading the diagonal entries of $\AA + \NN$. However, we can do even better since we can discard the diagonal entries of $\AA + \NN$. 
The main insight here is that for correlation matrices, our algorithm simply uniformly samples columns and rows to construct our row and column PCPs, since we know what the true diagonals should be. 
In this case, no matter what the adversary does to the diagonal,  $\phi_{\max}=1$ and we obtain a $\widetilde{O}(nk/\epsilon)$ query algorithm.

\begin{corollary}(Robust LRA for Correlation Matrices.)
\label{cor:cor_matrix_lra}
Let $k$ be an integer and $1>\epsilon>\eta>0$. Given $\AA + \NN$, where $\AA$ is a correlation matrix and $\NN$ is a corruption term such that $\|\NN\|^2_F \leq \eta \|\AA\|^2_F$ and for all $i \in [n]$ $\|\NN_{i,*}\|^2_2 \leq c\|\AA_{i,*}\|^2_2$ for a fixed constant $c$, there exists an algorithm that samples $\widetilde{O}\left(nk/\epsilon\right)$ entries in $\AA+\NN$ and with probability at least $99/100$, computes a rank $k$ matrix $\BB$ such that 
\begin{equation*}
    \|\AA - \BB \|^2_F \leq \|\AA - \AA_k \|^2_F + (\epsilon+ \sqrt{\eta})\|\AA\|^2_F
\end{equation*}
\end{corollary}

Note, the sample complexity of this algorithm is optimal, since there is an $\Omega(nk/\epsilon)$ query lower bound for additive-error low-rank approximation of correlation matrices, even when there is no corruption (see Corollary \ref{cor:sample_lb_corr}).

\vspace{0.2in}
\textbf{Additive-Error PSD Low-Rank Approximation.}
In the limit where $\eta =0$, $\phi_{max} =1$, and we obtain an algorithm with query complexity $\widetilde{O}(nk/\epsilon)$. While this guarantee is already implied by our algorithm for \textit{relative-error} low-rank approximation, our additive-error algorithm is simpler to implement, since the sampling probabilities can be computed \textit{exactly} by simply reading the diagonal. 
\begin{corollary}(Sample-Optimal Additive-Error LRA.)
Given a PSD matrix $\AA$, rank parameter $k$, and $\epsilon>0$, there exists an algorithm that samples $\widetilde{O}(nk/\epsilon)$ entries in $\AA$ and outputs a rank-$k$ matrix $\BB$ such that with probability at least $99/100$,
\begin{equation*}
    \|\AA - \BB \|^2_F \leq \|\AA - \AA_k \|^2_F + \epsilon \|\AA\|^2_F
\end{equation*}
\end{corollary}

We show a matching lower bound on the query complexity of additive-error low-rank approximation of PSD matrices. Here, we simply observe that the lower bound construction introduced by \cite{mw17} of $\Omega\left(\frac{nk}{\epsilon}\right)$ also holds for additive error. As a consequence our algorithm is optimal in the setting where there is no corruption. 

\begin{corollary}(Correlation Matrix Lower Bound, Theorem 13 \cite{mw17}.)
\label{cor:sample_lb_corr}
Let $\AA$ be a PSD matrix,  $k\in \mathbb{Z}$ and $\epsilon>0$ be such that $\frac{nk}{\epsilon} = o(n^2)$. Any randomized algorithm, $\mathcal{A}$, that with probability at least $2/3$, computes a rank $k$ matrix $\BB$ such that 
\begin{equation*}
    \|\AA - \BB \|^2_F \leq \| \AA - \AA_k \|^2_F + \epsilon\|\AA \|^2_F
\end{equation*}
must read $\Omega\left(\frac{nk}{\epsilon}\right)$ entries of $\AA$ on some input, possibly adaptively, in expectation.
\end{corollary}
\begin{proof}
We observe that in the lower bound construction of \cite{mw17}, the matrix $\AA$ is binary, with all $1$s on a the diagonal, and $k$ off-diagonal blocks of all $1$s, each of size $\sqrt{\frac{2\epsilon n}{k}}\times\sqrt{\frac{2\epsilon n}{k}}$. Therefore, $\AA$ is a correlation matrix and $\|\AA\|^2_F = (1+2\epsilon)n $. Further, the optimal rank-$k$ cost, $\|\AA -\AA_k \|^2_F = \Theta(n)$. To compute an additive-error approximation, any algorithm must caputure $\epsilon\|\AA\|^2_F =\epsilon n$ mass among the off-diagonal entries of $\AA$. Note, the remaining proof is identical to Theorem 13 in \cite{mw17}.
\end{proof}

\newpage

\paragraph{Acknowledgments.} A. Bakshi and D. Woodruff would like to acknowledge support from the National Science Foundation under Grant No. CCF-1815840. Part of this work was also done while they were visiting the Simons Institute for the Theory of Computing.

\bibliographystyle{alpha}
\bibliography{lra}

\newcommand{\etalchar}[1]{$^{#1}$}
\begin{thebibliography}{CKM{\etalchar{+}}11}

\bibitem[AFKM01]{achlioptas2001web}
Dimitris Achlioptas, Amos Fiat, Anna~R Karlin, and Frank McSherry.
\newblock Web search via hub synthesis.
\newblock In {\em Proceedings 42nd IEEE Symposium on Foundations of Computer
  Science}, pages 500--509. IEEE, 2001.

\bibitem[ALN07]{arora2007frechet}
Sanjeev Arora, James~R Lee, and Assaf Naor.
\newblock Fr{\'e}chet embeddings of negative type metrics.
\newblock {\em Discrete \& Computational Geometry}, 38(4):726--739, 2007.

\bibitem[ALN08]{arora2008euclidean}
Sanjeev Arora, James Lee, and Assaf Naor.
\newblock Euclidean distortion and the sparsest cut.
\newblock {\em Journal of the American Mathematical Society}, 21(1):1--21,
  2008.

\bibitem[AM05]{achlioptas2005spectral}
Dimitris Achlioptas and Frank McSherry.
\newblock On spectral learning of mixtures of distributions.
\newblock In {\em International Conference on Computational Learning Theory},
  pages 458--469. Springer, 2005.

\bibitem[AM15]{alaoui2015fast}
Ahmed Alaoui and Michael~W Mahoney.
\newblock Fast randomized kernel ridge regression with statistical guarantees.
\newblock In {\em Advances in Neural Information Processing Systems}, pages
  775--783, 2015.

\bibitem[ARV09]{arora2009expander}
Sanjeev Arora, Satish Rao, and Umesh Vazirani.
\newblock Expander flows, geometric embeddings and graph partitioning.
\newblock {\em Journal of the ACM (JACM)}, 56(2):5, 2009.

\bibitem[BDN15]{BDN15}
Jean Bourgain, Sjoerd Dirksen, and Jelani Nelson.
\newblock Toward a unified theory of sparse dimensionality reduction in
  euclidean space.
\newblock In {\em Proceedings of the Forty-Seventh Annual {ACM} on Symposium on
  Theory of Computing, {STOC} 2015, Portland, OR, USA, June 14-17, 2015}, pages
  499--508, 2015.

\bibitem[BW18]{bakshi2018sublinear}
Ainesh Bakshi and David Woodruff.
\newblock Sublinear time low-rank approximation of distance matrices.
\newblock In {\em Advances in Neural Information Processing Systems}, pages
  3782--3792, 2018.

\bibitem[BY02]{bar2002complexity}
Ziv Bar-Yossef.
\newblock {\em The complexity of massive data set computations}.
\newblock PhD thesis, University of California, Berkeley, 2002.

\bibitem[CEM{\etalchar{+}}15]{cohen2015dimensionality}
Michael~B Cohen, Sam Elder, Cameron Musco, Christopher Musco, and Madalina
  Persu.
\newblock Dimensionality reduction for k-means clustering and low rank
  approximation.
\newblock In {\em Proceedings of the forty-seventh annual ACM symposium on
  Theory of computing}, pages 163--172. ACM, 2015.

\bibitem[CGR05]{chawla2005embeddings}
Shuchi Chawla, Anupam Gupta, and Harald R{\"a}cke.
\newblock Embeddings of negative-type metrics and an improved approximation to
  generalized sparsest cut.
\newblock In {\em Proceedings of the sixteenth annual ACM-SIAM symposium on
  Discrete algorithms}, pages 102--111. Society for Industrial and Applied
  Mathematics, 2005.

\bibitem[CKM{\etalchar{+}}11]{christiano2011electrical}
Paul Christiano, Jonathan~A Kelner, Aleksander Madry, Daniel~A Spielman, and
  Shang-Hua Teng.
\newblock Electrical flows, laplacian systems, and faster approximation of
  maximum flow in undirected graphs.
\newblock In {\em Proceedings of the forty-third annual ACM symposium on Theory
  of computing}, pages 273--282. ACM, 2011.

\bibitem[CLL{\etalchar{+}}10]{chen2010multiclass}
Pei-Chun Chen, Kuang-Yao Lee, Tsung-Ju Lee, Yuh-Jye Lee, and Su-Yun Huang.
\newblock Multiclass support vector classification via coding and regression.
\newblock {\em Neurocomputing}, 73(7-9):1501--1512, 2010.

\bibitem[CLM{\etalchar{+}}15]{cohen2015uniform}
Michael~B Cohen, Yin~Tat Lee, Cameron Musco, Christopher Musco, Richard Peng,
  and Aaron Sidford.
\newblock Uniform sampling for matrix approximation.
\newblock In {\em Proceedings of the 2015 Conference on Innovations in
  Theoretical Computer Science}, pages 181--190. ACM, 2015.

\bibitem[CLW18]{chia2018quantum}
Nai-Hui Chia, Han-Hsuan Lin, and Chunhao Wang.
\newblock Quantum-inspired sublinear classical algorithms for solving low-rank
  linear systems.
\newblock {\em arXiv preprint arXiv:1811.04852}, 2018.

\bibitem[CMM17]{cohenmm17}
Michael~B. Cohen, Cameron Musco, and Christopher Musco.
\newblock Input sparsity time low-rank approximation via ridge leverage score
  sampling.
\newblock In {\em Proceedings of the Twenty-Eighth Annual {ACM-SIAM} Symposium
  on Discrete Algorithms, {SODA} 2017, Barcelona, Spain, Hotel Porta Fira,
  January 16-19}, pages 1758--1777, 2017.

\bibitem[CNW15]{cohen2015optimal}
Michael~B Cohen, Jelani Nelson, and David~P Woodruff.
\newblock Optimal approximate matrix product in terms of stable rank.
\newblock {\em arXiv preprint arXiv:1507.02268}, 2015.

\bibitem[Coh16]{cohen2016nearly}
Michael~B Cohen.
\newblock Nearly tight oblivious subspace embeddings by trace inequalities.
\newblock In {\em Proceedings of the twenty-seventh annual ACM-SIAM symposium
  on Discrete algorithms}, pages 278--287. SIAM, 2016.

\bibitem[CRR{\etalchar{+}}96]{chandra1996electrical}
Ashok~K Chandra, Prabhakar Raghavan, Walter~L Ruzzo, Roman Smolensky, and
  Prasoon Tiwari.
\newblock The electrical resistance of a graph captures its commute and cover
  times.
\newblock {\em Computational Complexity}, 6(4):312--340, 1996.

\bibitem[CW09]{clarkson2009numerical}
Kenneth~L Clarkson and David~P Woodruff.
\newblock Numerical linear algebra in the streaming model.
\newblock In {\em Proceedings of the forty-first annual ACM symposium on Theory
  of computing}, pages 205--214. ACM, 2009.

\bibitem[CW13]{clarkson2013low}
Kenneth~L Clarkson and David~P Woodruff.
\newblock Low rank approximation and regression in input sparsity time.
\newblock In {\em Proceedings of the forty-fifth annual ACM symposium on Theory
  of computing}, pages 81--90. ACM, 2013.

\bibitem[CW17]{clarkson2017low}
Kenneth~L Clarkson and David~P Woodruff.
\newblock Low-rank psd approximation in input-sparsity time.
\newblock In {\em Proceedings of the Twenty-Eighth Annual ACM-SIAM Symposium on
  Discrete Algorithms}, pages 2061--2072. Society for Industrial and Applied
  Mathematics, 2017.

\bibitem[DFK{\etalchar{+}}04]{drineas2004clustering}
Petros Drineas, Alan Frieze, Ravi Kannan, Santosh Vempala, and V~Vinay.
\newblock Clustering large graphs via the singular value decomposition.
\newblock {\em Machine learning}, 56(1-3):9--33, 2004.

\bibitem[DKM06]{drineas2006fast}
Petros Drineas, Ravi Kannan, and Michael~W Mahoney.
\newblock Fast monte carlo algorithms for matrices i: Approximating matrix
  multiplication.
\newblock {\em SIAM Journal on Computing}, 36(1):132--157, 2006.

\bibitem[DKR02]{drineas2002competitive}
Petros Drineas, Iordanis Kerenidis, and Prabhakar Raghavan.
\newblock Competitive recommendation systems.
\newblock In {\em Proceedings of the thiry-fourth annual ACM symposium on
  Theory of computing}, pages 82--90. ACM, 2002.

\bibitem[DL09]{deza2009geometry}
Michel~Marie Deza and Monique Laurent.
\newblock {\em Geometry of cuts and metrics}, volume~15.
\newblock Springer, 2009.

\bibitem[Fis05]{fisk2005very}
Steve Fisk.
\newblock A very short proof of cauchy's interlace theorem for eigenvalues of
  hermitian matrices.
\newblock {\em arXiv preprint math/0502408}, 2005.

\bibitem[FKV04]{fkv04}
Alan~M. Frieze, Ravi Kannan, and Santosh Vempala.
\newblock Fast monte-carlo algorithms for finding low-rank approximations.
\newblock {\em J. {ACM}}, 51(6):1025--1041, 2004.

\bibitem[FSS13]{feldman2013turning}
Dan Feldman, Melanie Schmidt, and Christian Sohler.
\newblock Turning big data into tiny data: Constant-size coresets for k-means,
  pca and projective clustering.
\newblock In {\em Proceedings of the twenty-fourth annual ACM-SIAM symposium on
  Discrete algorithms}, pages 1434--1453. Society for Industrial and Applied
  Mathematics, 2013.

\bibitem[FT07]{friedland2007generalized}
Shmuel Friedland and Anatoli Torokhti.
\newblock Generalized rank-constrained matrix approximations.
\newblock {\em SIAM Journal on Matrix Analysis and Applications},
  29(2):656--659, 2007.

\bibitem[GLF{\etalchar{+}}10]{gross2010quantum}
David Gross, Yi-Kai Liu, Steven~T Flammia, Stephen Becker, and Jens Eisert.
\newblock Quantum state tomography via compressed sensing.
\newblock {\em Physical review letters}, 105(15):150401, 2010.

\bibitem[GLT18]{gilyen2018quantum}
Andr{\'a}s Gily{\'e}n, Seth Lloyd, and Ewin Tang.
\newblock Quantum-inspired low-rank stochastic regression with logarithmic
  dependence on the dimension.
\newblock {\em arXiv preprint arXiv:1811.04909}, 2018.

\bibitem[Gru17]{gruber2017improving}
Marvin Gruber.
\newblock {\em Improving Efficiency by Shrinkage: The James--Stein and Ridge
  Regression Estimators}.
\newblock Routledge, 2017.

\bibitem[GSLW19]{gilyen2019quantum}
Andr{\'a}s Gily{\'e}n, Yuan Su, Guang~Hao Low, and Nathan Wiebe.
\newblock Quantum singular value transformation and beyond: exponential
  improvements for quantum matrix arithmetics.
\newblock In {\em Proceedings of the 51st Annual ACM SIGACT Symposium on Theory
  of Computing}, pages 193--204. ACM, 2019.

\bibitem[Hig02]{h02}
Nicholas~J Higham.
\newblock Computing the nearest correlation matrix—a problem from finance.
\newblock {\em IMA journal of Numerical Analysis}, 22(3):329--343, 2002.

\bibitem[IVWW19]{indyk2019sample}
Piotr Indyk, Ali Vakilian, Tal Wagner, and David Woodruff.
\newblock Sample-optimal low-rank approximation of distance matrices.
\newblock {\em arXiv preprint arXiv:1906.00339}, 2019.

\bibitem[Kle99]{kleinberg1999authoritative}
Jon~M Kleinberg.
\newblock Authoritative sources in a hyperlinked environment.
\newblock {\em Journal of the ACM (JACM)}, 46(5):604--632, 1999.

\bibitem[KLM{\etalchar{+}}17]{kapralov2017single}
Michael Kapralov, Yin~Tat Lee, CN~Musco, Christopher~Paul Musco, and Aaron
  Sidford.
\newblock Single pass spectral sparsification in dynamic streams.
\newblock {\em SIAM Journal on Computing}, 46(1):456--477, 2017.

\bibitem[KMP14]{koutis2014approaching}
Ioannis Koutis, Gary~L Miller, and Richard Peng.
\newblock Approaching optimality for solving sdd linear systems.
\newblock {\em SIAM Journal on Computing}, 43(1):337--354, 2014.

\bibitem[KP16]{kerenidis2016quantum}
Iordanis Kerenidis and Anupam Prakash.
\newblock Quantum recommendation systems.
\newblock {\em arXiv preprint arXiv:1603.08675}, 2016.

\bibitem[KSV05]{kannan2005spectral}
Ravindran Kannan, Hadi Salmasian, and Santosh Vempala.
\newblock The spectral method for general mixture models.
\newblock In {\em International Conference on Computational Learning Theory},
  pages 444--457. Springer, 2005.

\bibitem[LMP13]{li2013iterative}
Mu~Li, Gary~L Miller, and Richard Peng.
\newblock Iterative row sampling.
\newblock In {\em 2013 IEEE 54th Annual Symposium on Foundations of Computer
  Science}, pages 127--136. IEEE, 2013.

\bibitem[LS15]{lee2015efficient}
Yin~Tat Lee and Aaron Sidford.
\newblock Efficient inverse maintenance and faster algorithms for linear
  programming.
\newblock In {\em 2015 IEEE 56th Annual Symposium on Foundations of Computer
  Science}, pages 230--249. IEEE, 2015.

\bibitem[McS01]{mcsherry2001spectral}
Frank McSherry.
\newblock Spectral partitioning of random graphs.
\newblock In {\em Proceedings 42nd IEEE Symposium on Foundations of Computer
  Science}, pages 529--537. IEEE, 2001.

\bibitem[MM13]{mm13}
Xiangrui Meng and Michael~W. Mahoney.
\newblock Low-distortion subspace embeddings in input-sparsity time and
  applications to robust linear regression.
\newblock In {\em Symposium on Theory of Computing Conference, STOC'13, Palo
  Alto, CA, USA, June 1-4, 2013}, pages 91--100, 2013.

\bibitem[MM17]{musco2017recursive}
Cameron Musco and Christopher Musco.
\newblock Recursive sampling for the nystrom method.
\newblock In {\em Advances in Neural Information Processing Systems}, pages
  3833--3845, 2017.

\bibitem[MST15]{madry2015fast}
Aleksander Madry, Damian Straszak, and Jakub Tarnawski.
\newblock Fast generation of random spanning trees and the effective resistance
  metric.
\newblock In {\em Proceedings of the twenty-sixth annual ACM-SIAM symposium on
  Discrete algorithms}, pages 2019--2036. Society for Industrial and Applied
  Mathematics, 2015.

\bibitem[MW17]{mw17}
Cameron Musco and David~P. Woodruff.
\newblock Sublinear time low-rank approximation of positive semidefinite
  matrices.
\newblock In {\em 58th {IEEE} Annual Symposium on Foundations of Computer
  Science, {FOCS} 2017, Berkeley, CA, USA, October 15-17, 2017}, pages
  672--683, 2017.

\bibitem[NN13]{NN13}
Jelani Nelson and Huy~L. Nguyen.
\newblock {OSNAP:} faster numerical linear algebra algorithms via sparser
  subspace embeddings.
\newblock In {\em 54th Annual {IEEE} Symposium on Foundations of Computer
  Science, {FOCS} 2013, 26-29 October, 2013, Berkeley, CA, {USA}}, pages
  117--126, 2013.

\bibitem[RSML18]{rebentrost2018quantum}
Patrick Rebentrost, Adrian Steffens, Iman Marvian, and Seth Lloyd.
\newblock Quantum singular-value decomposition of nonsparse low-rank matrices.
\newblock {\em Physical review A}, 97(1):012327, 2018.

\bibitem[RV07]{rudelson2007sampling}
Mark Rudelson and Roman Vershynin.
\newblock Sampling from large matrices: An approach through geometric
  functional analysis.
\newblock {\em Journal of the ACM (JACM)}, 54(4):21, 2007.

\bibitem[Sar06]{sarlos2006improved}
Tamas Sarlos.
\newblock Improved approximation algorithms for large matrices via random
  projections.
\newblock In {\em FOCS}, pages 143--152, 2006.

\bibitem[Sch38]{schoenberg1938metric}
Isaac~J Schoenberg.
\newblock Metric spaces and positive definite functions.
\newblock {\em Transactions of the American Mathematical Society},
  44(3):522--536, 1938.

\bibitem[SS11]{spielman2011graph}
Daniel~A Spielman and Nikhil Srivastava.
\newblock Graph sparsification by effective resistances.
\newblock {\em SIAM Journal on Computing}, 40(6):1913--1926, 2011.

\bibitem[SW19]{sw19}
Xiaofei Shi and David~P. Woodruff.
\newblock Sublinear time numerical linear algebra for structured matrices.
\newblock In {\em The Thirty-Third {AAAI} Conference on Artificial
  Intelligence, {AAAI} 2019, The Thirty-First Innovative Applications of
  Artificial Intelligence Conference, {IAAI} 2019, The Ninth {AAAI} Symposium
  on Educational Advances in Artificial Intelligence, {EAAI} 2019, Honolulu,
  Hawaii, USA, January 27 - February 1, 2019.}, pages 4918--4925, 2019.

\bibitem[Tan19]{tang2019quantum}
Ewin Tang.
\newblock A quantum-inspired classical algorithm for recommendation systems.
\newblock In {\em Proceedings of the 51st Annual ACM SIGACT Symposium on Theory
  of Computing}, pages 217--228. ACM, 2019.

\bibitem[TD87]{terwilliger1987classification}
Paul Terwilliger and Michel Deza.
\newblock The classification of finite connected hypermetric spaces.
\newblock {\em Graphs and Combinatorics}, 3(1):293--298, 1987.

\bibitem[Woo14]{w14}
David~P. Woodruff.
\newblock Sketching as a tool for numerical linear algebra.
\newblock {\em Foundations and Trends in Theoretical Computer Science},
  10(1-2):1--157, 2014.

\bibitem[Yao77]{yao_minimax}
Andrew Chi-Chin Yao.
\newblock Probabilistic computations: Toward a unified measure of complexity.
\newblock In {\em Proceedings of the 18th Annual Symposium on Foundations of
  Computer Science}, SFCS '77, pages 222--227, Washington, DC, USA, 1977. IEEE
  Computer Society.

\end{thebibliography}

\newpage

\end{document}